\documentclass[twoside,12pt]{article} 
\usepackage[margin=1.0in]{geometry}

\usepackage{amssymb,amsmath,amsthm}
\usepackage{mathtools}
\usepackage{bm,lscape}
\usepackage{bbm}
\usepackage{algpseudocode}
\usepackage[ruled,vlined]{algorithm2e}
\usepackage{mathrsfs,dsfont}
\usepackage{array}
\usepackage{newtxtext}
\usepackage{wasysym}
\RequirePackage[
    colorlinks=true,
    linkcolor=blue,
    urlcolor=blue,
    citecolor=blue]{hyperref}
\RequirePackage{natbib}
\usepackage[utf8]{inputenc}
\usepackage[english]{babel}

\usepackage{xparse}
\NewDocumentCommand{\evalat}{sO{\big}mm}{%
  \IfBooleanTF{#1}
   {\mleft. #3 \mright|_{#4}}
   {#3#2|_{#4}}%
}

\usepackage{graphicx}
\usepackage{color}
\usepackage{rotating}
\usepackage{authblk}
\usepackage{appendix}

\allowdisplaybreaks

\def\xb{\mathbf{x}}
\def\yb{\mathbf{y}}

\def\vb{\mathbf{v}}
\def\fb{\mathbf{f}}

\def\ub{\mathbf{u}}

\def\thetab{\boldsymbol{\theta}}

\def\simind{\stackrel{\mbox{\scriptsize{ind}}}{\sim}}
\def\simiid{\stackrel{\mbox{\scriptsize{iid}}}{\sim}}

\def\thetab{{\boldsymbol \theta}}

\def\eb{\mathbf{e}}
\def\vb{\mathbf{v}}
\def\xb{\mathbf{x}}
\def\kb{\mathbf{k}}
\def\Yb{\mathbf{Y}}
\def\Wb{\mathbf{W}}
\def\yb{\mathbf{y}}
\def\gb{\mathbf{g}}
\def\hb{\mathbf{h}}
\def\rb{\mathbf{r}}

\def\Mb{\mathbf{M}}
\def\Hb{\mathbf{H}}
\def\epsilonb{{\boldsymbol\epsilon}}

\def\vb{\mathbf{v}}

\def\ub{\mathbf{u}}

\def\Db{\mathbf{D}}

\def\thetab{\boldsymbol{\theta}}

\allowdisplaybreaks

\newcommand{\PP}{\mathds{P}}
\newcommand{\ddr}{\mathrm{d}}
\newcommand{\EE}{\mathds{E}}

\newtheorem{thm}{Theorem}[section]
\newtheorem{definition}[thm]{Definition}
\newtheorem{lem}[thm]{Lemma}

\newtheorem{prp}[thm]{Proposition}

\newtheorem{remark}[thm]{Remark}

\providecommand{\keywords}[1]
{
  \small	
  \textbf{\textit{Keywords:}} #1
}

\begin{document}

\title{Quasi-Bayes empirical Bayes: a sequential approach to the Poisson compound decision problem}


\author[1]{Stefano Favaro\thanks{stefano.favaro@unito.it}}
\author[2]{Sandra Fortini\thanks{sandra.fortini@unibocconi.it}}
\affil[1]{\small{Department of Economics and Statistics, University of Torino and Collegio Carlo Alberto, Italy}}
\affil[2]{\small{Department of Decision Sciences, Bocconi University, Italy}}

\maketitle

\begin{abstract}
The Poisson compound decision problem is a long-standing problem is statistics, for which empirical Bayes methods are commonly used to estimate Poisson means in static or batch settings. We consider this problem in a streaming, or online, framework. Building on a quasi-Bayesian approach based on Newton’s algorithm, we develop a sequential estimate that is easy to evaluate, computationally efficient, and has constant per-observation cost as the data accrue. We establish frequentist guarantees for the proposed estimate, including consistency and asymptotic optimality, with optimality understood as vanishing excess Bayes risk, or regret. Empirical performance is assessed through simulation studies and comparisons with benchmark procedures.
\end{abstract}



\section{Introduction}\label{sec1}
\subsection{Background and motivation}

Given $n\geq1$ observations modeled as independent Poisson random variables $Y_{1},\ldots,Y_{n}$ with means $\theta_{1},\ldots,\theta_{n}$, respectively, the Poisson compound decision problem concerns the estimation of $(\theta_{1},\ldots,\theta_{n})$ under squared-error loss. Empirical Bayes provides a general approach for compound decision problems \citep{Rob(56),Zha(03)}. If the $\theta_{i}$'s are i.i.d. as a prior $G$ on $\Theta\subseteq\mathbb{R}^{+}$, then the best estimate of $\theta_{i}$ is the posterior mean or Bayes estimate, namely
\begin{equation}\label{post_mean}
\hat{\theta}_{G}(y)=E_{G}[\theta_{i}\mid Y_{i}=y]=\frac{\int_{\Theta}\theta\frac{\text{e}^{-\theta}\theta^{y}}{y!}G(\ddr\theta)}{\int_{\Theta}\frac{\text{e}^{-\theta}\theta^{y}}{y!}G(\ddr\theta)}=(y+1)\frac{p_{G}(y+1)}{p_{G}(y)}\qquad y\in\mathbb{N}_{0},
\end{equation}
where $p_{G}$ denotes the probability mass function of the $Y_{i}$'s. Since $G$ is generally unknown, empirical Bayes proceeds by estimating either $p_G$, in the so-called \(f\)-modeling strategy, or $G$, in the so-called $g$-modeling strategy \citep{Efr(14),Efr(19)}.

A prominent $f$-modeling strategy is Robbins' method \citep{Rob(56)}, which replaces $p_{G}$ in \eqref{post_mean} with the empirical distribution of the $Y_{i}$'s. Despite its conceptual appeal and computational simplicity, $f$-modeling can be numerically unstable and lacks robustness. In particular, it is sensitive to outlying observations, or more precisely to counts that occur rarely, which may yield exceptionally small or large estimates \citep{Bro(13),She(24)}.

Examples of $g$-modeling include estimates of $G$ in \eqref{post_mean} based on maximum likelihood and minimum distance methods \citep{Jan(24)}. A fully Bayesian alternative places a suitable hyperprior on $G$, yielding what is termed Bayes empirical Bayes \citep{Dee(81),Efr(19)}. Compared with $f$-modeling, $g$-modeling typically produces more accurate estimates and allows prior information to be incorporated more naturally. Its main drawback is computational cost, especially in nonparametric and high-dimensional settings \citep{Jan(23),Teh(25)}.

The Poisson compound decision problem provides a flexible framework for modeling count data and has been successfully applied in settings where static or batch processing is appropriate, including  microarrays \citep{Efr(03)}, insurance analytics \citep[Chapter 6]{Efr(21)}, deconvolution \citep{Efr(16)}, data confidentiality \citep{Zha(05)} and sports analytics \citep{Jan(24)}. The growing prevalence of streaming data in some of these areas, as well as in fields such as social media, biosurveillance, healthcare and e-commerce, motivates extension of this framework to an online setting, where timely inference is required. Although existing empirical Bayes methods could in principle be adapted to streaming data, they typically rely on batch optimization or posterior computation and are therefore ill suited to sequential implementation. To our knowledge, no sequential methodology for the Poisson compound decision problem currently offers either theoretical guarantees or convincing empirical support.

\subsection{Preview of our contributions}

We propose a sequential $g$-modeling approach to the Poisson compound decision problem. Starting from an initial guess $G_0$ for $G$, for each $n\geq1$ the method updates $G_{n-1}$ after observing $Y_{n}$ according to the recursive procedure of \citet{Smi(78)}, known as Newton's algorithm \citep{New(98),Mar(08)}. This scheme can be interpreted as a quasi-Bayesian learning model and induces a predictive construction that is asymptotically equivalent, as $n\rightarrow+\infty$, to a Bayesian model \citep{For(20)}. A quasi-Bayes empirical Bayes estimate of $\hat{\theta}_{G}(y)$ is then obtained from \eqref{post_mean} by replacing $G$ with $G_{n}$. The resulting estimate, $\hat{\theta}_{G_{n}}(y)$, is straightforward to evaluate, computationally efficient, and has constant per-observation computational cost as data accrue. Moreover, viewing Newton’s algorithm as a learning model further enables the construction of asymptotic credible intervals through a Gaussian central limit theorem.

Frequentist guarantees are established for $\hat{\theta}_{G_{n}}(y)$ under the assumption that the observations $Y_n$'s are i.i.d. from a Poisson mixture model with ``true'' mixing distribution $G^{\ast}$, referred to as the oracle prior. We give sufficient conditions on $G^{\ast}$ under which $\hat{\theta}_{G_{n}}(y)$ is asymptotically optimal as $n\rightarrow+\infty$, where optimality is understood as vanishing excess Bayes risk, or regret.

We assess the proposed methodology through empirical studies on synthetic and real data. For fixed $n\geq1$, $\hat{\theta}_{G_{n}}(y)$ outperforms Robbins' method and achieves performance comparable to that of the nonparametric maximum likelihood and minimum distance estimates. 

\subsection{Related work}

Our approach links the empirical Bayes literature to a recent and growing body of work on predictive or recursive inference, often referred to as the ``post-Bayes" framework. This literature includes developments on predictive distributions and martingale posteriors \citep{Fon(23),Fon(24),Hol(23),Bat(25),Yun(25)}, as well as related work on generalized Bayesian updating \citep{Bis(16),Kno(22)}. 


\section{Sequential $g$-modeling for the Poisson compound decision problem}\label{sec:quasi}
\subsection{Modeling assumptions}
For a stream of observations $(Y_{n})_{n\geq1}$, we propose a sequential $g$-modeling methodology which, on  a probability space $(\Omega,\mathcal F,\PP)$, specifies the distribution of $(Y_{n})_{n\geq1}$ through the conditional distribution of $Y_{n+1}$ given  $(Y_1,\dots,Y_n)$, for each $n\geq0$. As in the Bayesian framework, the model is interpreted as a learning model and not as a data-generating process \citep{For(25)}.

Denote by $(\mathcal G_n)_{n\geq1}$ the natural filtration of $(Y_n)_{n\geq1}$. For each $n\geq0$, the conditional distribution of $Y_{n+1}$ given $\mathcal G_n$ is modeled as a mixture of Poisson distributions with the mean $\theta$ ranging in compact $\Theta\subset\mathbb{R}^{+}$ , namely
\begin{equation}\label{poiss_model}
p_{G_n}(y)=\PP(Y_{n+1}=y\mid \mathcal G_n)=\int_\Theta \text{Poisson}(y\mid\theta)G_n(\ddr\theta)\qquad  y\in\mathbb{N}_{0},
\end{equation}
where the mixing $G_{n}$ in \eqref{poiss_model} is defined recursively through Newton's algorithm \citep{New(98)}:
\begin{equation}\label{eq:newton}
G_{n+1}(\ddr\theta)=(1-\alpha_{n+1})G_{n}(\ddr\theta)+\alpha_{n+1}\frac{\text{Poisson}(Y_{n+1}\mid\theta)G_{n}(\ddr\theta)}{\int_\Theta \text{Poisson}(Y_{n+1}\mid\theta)G_{n}(\ddr\theta)},
\end{equation}
starting from an initial guess $G_{0}$ for $G$, where the $\alpha_{n}$'s in $(0,1)$ are such that  $\sum_{n\geq1}\alpha_{n}=+\infty$ and $\sum_{n\geq1}\alpha^{2}_{n}<+\infty$. According to \eqref{eq:newton}, after observing \(Y_{n+1}\), the model updates \(G_n\) by taking a weighted average of \(G_n\) and its posterior distribution based on \(Y_{n+1}\), with weight \(\alpha_{n+1}\). The sequence $(\alpha_{n})_{n\geq1}$ is referred to as the learning rate. A standard choice is $\alpha_{n}=(\alpha+n)^{-\gamma}$, with $\alpha>0$ and $\gamma\in (1/2,1]$; see \citet{For(20)} for a discussion on the choice of the learning rate.

\subsection{Quasi-Bayes properties}

We show that the learning model \eqref{poiss_model}-\eqref{eq:newton} is quasi-Bayesian in the sense of \citet{For(20)}, yielding a predictive construction that is asymptotically equivalent, as $n\rightarrow+\infty$, to a Bayesian model. 

\begin{thm}\label{th:cid}
Let $(Y_{n})_{n\geq1}$ and $(G_n)_{n\geq1}$ be as in \eqref{poiss_model}-\eqref{eq:newton}. As $n\rightarrow+\infty$, $G_n$ converges weakly $\PP-$a.s. to a random probability distribution $\tilde G$ on $\Theta$, such that, by defining $p_{\tilde G}$ as in \eqref{poiss_model} with $\tilde G$ in place of $G_n$,  it holds $\PP-$a.s. that $p_{\tilde G}(y)=\lim_{n\rightarrow\infty}n^{-1} \sum_{1\leq k\leq n}I(Y_k=y)$ and $\PP(Y_{n+k}=y\mid\mathcal G_n)=E( p_{\tilde G}(y)\mid\mathcal G_n)$, for every $n\geq 0$, $k\geq 1$ and $y\in\mathbb N_0$. In particular, $\PP(Y_n=y)=\int p_{\tilde G}(y)d\PP$. 
\end{thm}

See \ref{app3_0} for the proof of Theorem \ref{th:cid}. Under the learning model \eqref{poiss_model}-\eqref{eq:newton}, the future observations in $(Y_{n})_{n\geq1}$ are conditionally identically distributed given $\mathcal G_n$, for each $n \geq 0$. By \citet[Theorem 2.2]{Berti(04)},  it follows that the $Y_n$'s are asymptotically i.i.d., conditional on the random limit $\tilde G$, with distribution $p_{\tilde G}$. This places the model within the quasi-Bayesian framework developed in \citet{For(20)}. From this viewpoint, \(p_{\tilde G}\) is the random distribution that would be assigned with access to an infinite sample, so that \(\tilde G\), unavailable in practice, is the parameter of interest, and \eqref{post_mean} becomes
\begin{equation}\label{eq:param}
\hat{\theta}_{\tilde{G}}(y)=(y+1)\frac{p_{\tilde{G}}(y+1)}{p_{\tilde{G}}(y)}\qquad y\in\mathbb{N}_{0}.
\end{equation}
In contrast to the Bayesian framework, which places a prior on $G$ and assumes the $Y_n$'s to be i.i.d. conditional on $G$, the quasi-Bayesian approach avoids repeated recomputation of the  posterior distribution as new observations arrive, thereby yielding substantial computational savings in sequential applications. 

\subsection{Estimation and Gaussian approximation}\label{sec:estimation}

Under the learning model \eqref{poiss_model}-\eqref{eq:newton}, an estimate of $\hat{\theta}_{\tilde{G}}(y)$ is obtained by replacing the unknown $\tilde{G}$ in \eqref{eq:param} with $G_{n}$, namely
\begin{equation}\label{seq_estim}
\hat{\theta}_{G_{n}}(y)=(y+1)\frac{p_{G_{n}}(y+1)}{p_{G_{n}}(y)}\qquad y\in\mathbb{N}_{0}.
\end{equation}
We refer to $\hat{\theta}_{G_n}(y)$ as the quasi-Bayes empirical Bayes estimate. Because the map $G\mapsto \hat{\theta}_G(y)$ is nonlinear, the plug-in approximation $\hat{\theta}_{G_n}(y)$ to $\hat{\theta}_{\tilde G}(y)$ need not coincide with $E_{\tilde{G}}[\hat{\theta}_{\tilde G}(y)\,|\,\mathcal{G}_{n}]$. As a corollary of Theorem \ref{th:cid}, for every $y\in\mathbb{N}_{0}$, $\hat{\theta}_{G_n}(y)\to\hat{\theta}_{\tilde G}(y)$ $\PP-$a.s. as $n\to\infty$; the next theorem gives a related conditional Gaussian central limit theorem for $\hat{\theta}_{G_{n}}(y)$.

\begin{thm}\label{th:clt2}
Let \eqref{poiss_model}-\eqref{eq:newton} hold with $(\alpha_n)_{n\geq1}$  non increasing and $\sum_{n\geq 1}(\alpha_n^2/(\sum_{k\geq n}\alpha_k^2))^2<+\infty$, and let $b_n=(\sum_{k\geq n}\alpha_k^2)^{-1}$. For every $y\in\mathbb{N}_{0}$, the conditional distribution of $b_n^{1/2}(  \hat{\theta}_{\tilde{G}}(y)-\hat{\theta}_{G_{n}}(y))$, given $\mathcal G_n$, converges weakly $\PP-$a.s.  to a centered Gaussian kernel with the random variance 
\begin{equation}\label{eq:W}
W_{\tilde G}(y)=\hat\theta^{2}_{\tilde G}(y)
\sum_{z\in\mathbb N_0}p_{\tilde G}(z)\left(\frac{p_{\tilde G(\cdot\,|\,z)}(y)}{p_{\tilde G}(y)}-\frac{p_{\tilde G(\cdot\,|\,z)}(y+1)}{p_{\tilde G}(y+1)}\right)^2
\end{equation}
where $p_{\tilde G(\cdot\,|\,z)}(y)=\int_\Theta \text{Poisson}(y\mid\theta)\tilde G(\ddr\theta\mid z)$ with $\tilde G(\ddr\theta\mid z)\propto \text{Poisson}(z\mid\theta)\tilde G(\ddr\theta)$.
\end{thm}

See \ref{app3} for the proof of Theorem \ref{th:clt2}. If $W_{G_n}$ is defined as in \eqref{eq:W}, with $G_n$ in place of the unknown $\tilde G$, then, for every $y\in\mathbb{N}_{0}$, $W_{G_n}(y)\rightarrow W_{\tilde G}(y)$ $\PP-$a.s. as $n\rightarrow+\infty$. It follows that, conditionally on \(\mathcal G_n\), \(\hat{\theta}_{\tilde G}(y)\) is approximately Gaussian with mean \(\hat{\theta}_{G_n}(y)\) and variance \(b_n^{-1}W_{G_n}(y)\), for every $y\in\mathbb{N}_{0}$.

\subsection{Credible intervals}

For a fixed $\alpha\in(0,1)$, and with access to an infinite sample, the endpoints of a $(1-\alpha)$-level credible interval for $\theta$, given $Y=y$, are the $\alpha/2$ and $1-\alpha/2$ quantiles of the conditional distribution function
\begin{equation}\label{eq:FG}
    F_{\tilde G}(\theta\mid y)=\frac{\int_0^\theta\text{Poisson}(y\mid \vartheta)\tilde G(\ddr\vartheta)}{\int_\Theta \text{Poisson}(y\mid\vartheta)\tilde G(\ddr\vartheta)}.
\end{equation}
With $n$ observations, we replace $\tilde G$ by $G_n$ and incorporate into the credible interval the additional variability induced by this substitution. The next theorem is the basis for this construction.

\begin{thm}\label{th:clt3}
Let \eqref{poiss_model}-\eqref{eq:newton} hold with $(\alpha_n)_{n\geq1}$  non increasing and $\sum_{n\geq 1}(\alpha_n^2/(\sum_{k\geq n}\alpha_k^2))^2<+\infty$, and let $b_n=(\sum_{k\geq n}\alpha_k^2)^{-1}$. If $F_{G_n}(\cdot\mid y)$ denotes the distribution function in \eqref{eq:FG} with $G_n$ replacing $\tilde G$ then, for every $\theta\in\Theta$, the conditional distribution of $b_n^{1/2}(F_{G_{n}}(\theta\mid y)-  F_{\tilde{G}}(\theta\mid y))$, given $\mathcal G_n$, converges weakly $\PP-$a.s.  to a centered Gaussian kernel with the random variance 
\begin{equation}\label{eq:W'}
W'_{\tilde G,y}(\theta)=\sum_{z\in\mathbb{N}_{0}}p_{\tilde G}(z)\frac{p_{\tilde G(\cdot\mid z)}(y)^2}{p_{\tilde G}(y)^2}\left(F_{\tilde G(\cdot\mid z)}(\theta\mid y)-F_{\tilde G}(\theta\mid y)\right)^2,
\end{equation}
where $p_{\tilde G(\cdot\,|\,z)}(y)=\int_\Theta \text{Poisson}(y\mid\theta)\tilde G(\ddr\theta\mid z)$ with $\tilde G(\ddr\theta\mid z)\propto \text{Poisson}(z\mid\theta)\tilde G(\ddr\theta)$, and $F_{\tilde G(\cdot\mid z)}(\theta\mid y)$ as in \eqref{eq:FG} with $\tilde G(\cdot\mid z)$ in the place of $\tilde G$.
\end{thm}

See  \ref{app3new} for the proof of Theorem \ref{th:clt3}. Let $\tilde Y\mid\tilde\theta\sim \text{Poisson}(\cdot\mid\tilde\theta)$ and $\tilde\theta\mid \tilde G\sim\tilde G$. To construct a $(1-\alpha)$-level credible interval for $\tilde\theta$, let $W'_{G_n,y}$  be defined as in \eqref{eq:W'} with $G_n$ in place of the unknown $\tilde G$. For fixed $\beta\in(0,1)$, define
\begin{equation}\label{conf_ext1}
\underline{F}^{\beta}_{G_n}(\theta\mid y)=\max\{0,F_{G_n}(\theta\mid y)-z_{1-\beta/2}(b_{n}^{-1}W'_{G_n,y}(\theta))^{1/2}\}
\end{equation}
and
\begin{equation}\label{conf_ext2}
\overline{F}^\beta_{G_n}(\theta\mid y)=\min\{1,F_{G_n}(\theta\mid y)+z_{1-\beta/2}(b_{n}^{-1}W'_{G_n,y}(\theta))^{1/2}\},
\end{equation}
where $z_{1-\beta/2}$ is the quantile of the standard Gaussian distribution. For every $\theta\in\Theta$, $W'_{G_n,y}(\theta)\rightarrow W'_{\tilde G,y}(\theta)$ $\PP$-a.s. as $n\rightarrow+\infty$, and hence, for $n$ sufficiently large, $\PP( F_{\tilde G}(\theta\mid y)\leq \underline F_{G_n}^\beta(\theta\mid y))\leq\beta/2$ and $\PP(F_{\tilde G}(\theta\mid y)> \overline F_{G_n}^\beta(\theta\mid y))\leq \beta/2$. 
Therefore, for $\beta_1,\beta_2\in(0,1)$ such that $\beta_1+\beta_2=\alpha$, it follows from \eqref{conf_ext1}-\eqref{conf_ext2} that
\begin{equation}\label{qbeb_interval}
I_{G_{n}}(y)=[L_n(y),U_n(y)]\cap\Theta,
\end{equation}
with $L_n(y)=\sup\{\theta: \overline{F}^{\beta_1}_{G_n}(\theta\mid y)\leq \beta_2/2\}$ and $U_n(y)=\inf\{\theta: \underline{F}^{\beta_1}_{G_n}(\theta\mid y)\geq 1-\beta_2/2\}$, is a $(1-\alpha)$-level credible interval for $\tilde\theta$. The tuning parameters $\beta_1$ and $\beta_2$, which account for the uncertainty in $\tilde G$ and in $\theta$ conditionally on $\tilde G$, respectively, are chosen by solving an optimization problem aimed at minimizing (numerically) the length of the resulting credible interval $I_{G_n}$. See \ref{app5_num11}.

\subsection{Frequentist guarantees} \label{sec:regret}
Under the ``true"  Poisson mixture model, the $Y_n$'s are assumed to be i.i.d. from a mixture of Poisson distributions with the ``true'' mixing distribution $G^{\ast}$, or oracle prior, namely
\begin{displaymath}
p_{G^*}(y)=\int_\Theta \text{Poisson}(y\mid\theta)G^*(\ddr\theta)\qquad  y\in\mathbb{N}_{0}.
\end{displaymath}
Let $\PP_{G^*}$ be the probability measure under which the ``true" model holds. If $G^{\ast}$ were known, the Bayes estimate in \eqref{post_mean} would be $\hat\theta_{G^{*}}(y)$. Since $G^{\ast}$ is unknown, we instead use the quasi-Bayes empirical Bayes estimate $\hat{\theta}_{G_{n}}(y)$ in \eqref{seq_estim}. The regret incurring by using $\hat{\theta}_{G_{n}}(y)$ in place of $\hat\theta_{G^*}(y)$ is
\begin{equation}\label{reg}
\textsc{Regret}(G_n,G^*)=\sum_{y\in\mathbb{N}_{0}}(\hat\theta_{G_n}(y)-\hat\theta_{G^*}(y))^2p_{G^*}(y);
\end{equation}
see \citet[Section 3]{Efr(19)} for background on the regret \eqref{reg}. As a frequentist optimality guarantee for $\hat{\theta}_{G_{n}}(y)$, the next theorem shows that the regret \eqref{reg} converges to zero $\PP_{G^*}-$a.s. as $n\rightarrow+\infty$.

\begin{thm}\label{th:regret}
Let  $G^{\ast}$ be the oracle prior on $\mathbb{R}^{+}$ and, denoting by $\textsc{KL}$ the Kullback--Leibler divergence, let  $G^\dag=\text{argmin}_{G\in\overline{\mathcal S}}\textsc{KL}(p_{G^*} \,\|\, p_{G})$, where $\overline{\mathcal S}$ is the weak closure of the class $\mathcal S$  of probability measures on $\mathbb R^+$ that are absolutely continuous with respect to $G_0$. If $G_{n}$ is defined as in \eqref{eq:newton} with $\Theta$ compact, then for every $y\in\mathbb{N}_{0}$ as $n\rightarrow+\infty$
\begin{displaymath}
\hat\theta_{G_n}(y)\rightarrow \hat\theta_{G^{\dag}}(y)\qquad \PP_{G^*}-a.s.
\end{displaymath}
and 
\begin{displaymath}
\textsc{Regret}(G_n,G^\dag)\!=\!\sum_{y\in\mathbb N_0}(\hat\theta_{G_n}(y)-\hat\theta_{G^\dag}(y))^2p_{G^*}(y)\rightarrow 0\qquad \PP_{G^*}-a.s.
\end{displaymath}
If $G^*$ is absolutely continuous with respect to $G_0$, then $\lim_{n\rightarrow+\infty}\textsc{Regret}(G_n,G^*)\rightarrow 0$, $\PP_{G^*}-$a.s.
\end{thm}

See \ref{app:proofregret} for the proof of Theorem \ref{th:regret}. Under additional assumptions on Newton's algorithm, the rate of convergence of the regret to zero can also be derived; see \ref{app:freqfinite}.


\section{Numerical illustrations}

\subsection{Synthetic-data analysis}

We empirically assess the quasi-Bayes empirical Bayes methodology, hereafter QB, by comparing it, for fixed $n\geq 1$, with three benchmark nonparametric procedures: (i) Robbins' method (Robbins) \citep{Rob(56)}; (ii) maximum likelihood (ML) \citep{Jan(24)}; and (iii) minimum squared Hellinger distance (MHD) \citep{Jan(24)}. Details on these procedures are provided in \ref{app5_num11}.

For $n\in\{50,\,100,\,200,\,400\}$, generate i.i.d. data $\mathbf{Y}_{n}=(Y_{1},\ldots,Y_{n})$ from a Poisson mixture model with Weibull prior $G$ of scale $5$ and shape $3$. Newton's algorithm \eqref{eq:newton} requires numerical evaluation of an integral, i.e., the marginal likelihood, which we approximate by the trapezoidal rule. To this end, the density function of $G_{n}$ is represented through its values on a fixed grid of $d\geq 1$ quadratures points over $\Theta$ where $d$ control the numerical resolution. This representation is used only for numerical evaluation and imposes no modeling restriction on $\Theta$. Letting $U_\Theta=\max\{\max\{\mathbf{Y}_n\},\lceil Q_{n,0.99}+4(\max\{Q_{n,0.99},1\})^{1/2}\rceil\}$, with $Q_{n,0.99}=\textsc{Quantile}(\mathbf{Y}_n;0.99)$, we use a uniform grid with $d=1{,}000$ over $\Theta=(0,U_\Theta)$, set $G_0$ to be Uniform on $\Theta$, and take the learning rate to be $\alpha_n=(1+n)^{-0.99}$.

The performance of the QB estimate relative to the Robbins, ML, and MHD estimates is assessed through the empirical mean squared error (\textsc{e-MSE}). We also consider the oracle estimate (Oracle), obtained by replacing $G$ in \eqref{post_mean} with the Weibull prior distribution generating the data, which provides the empirical minimum mean squared error (\textsc{e-MMSE}). For the QB, Robbins, ML, and MHD estimates we further report the empirical regret, defined as $\textsc{e-REGRET}=\textsc{e-MSE}-\textsc{e-MMSE}$, that is, the excess empirical squared error relative to the oracle benchmark. See \ref{app5_num11}. A sensitivity analysis with respect to the the number of quadrature points $d\in\{5000,1000,500,100,50,10\}$, reported in \ref{app5_num13}, indicates that the performance of QB is robust to the choice of $d$.

\begin{table}[ht]
\centering
\caption{Values of \textsc{e-MSE} and \textsc{e-REGRET} for Oracle, Robbins, ML, MHD and QB estimates.}
{
\setlength{\tabcolsep}{0pt}
\begin{tabular}{@{}l@{\hspace{1cm}}*{5}{>{\centering\arraybackslash}p{1.75cm}}@{}}
\hline
\hline
 & Oracle & Robbins & ML & MHD & QB \\[0.1cm]
\hline

\multicolumn{6}{@{}l}{\underline{$n=100$}} \\[0.05cm]
\textsc{e-MSE}    & 1.796 & 11.319 & 1.848 & 1.977 & 1.989 \\
\textsc{e-REGRET} & 0 & 9.523 & 0.052 & 0.181 & 0.193 \\[0.3cm]

\multicolumn{6}{@{}l}{\underline{$n=400$}} \\[0.05cm]
\textsc{e-MSE}    & 1.585 & 6.072 & 2.910 & 1.678 & 1.947 \\
\textsc{e-REGRET} & 0 & 4.487 & 1.324 & 0.093 & 0.362 \\[0.1cm]
\hline
\hline
\end{tabular}
}
\label{weib_tab}
\end{table}%

\begin{figure}[t]
\centering
\includegraphics[width=.85\textwidth]{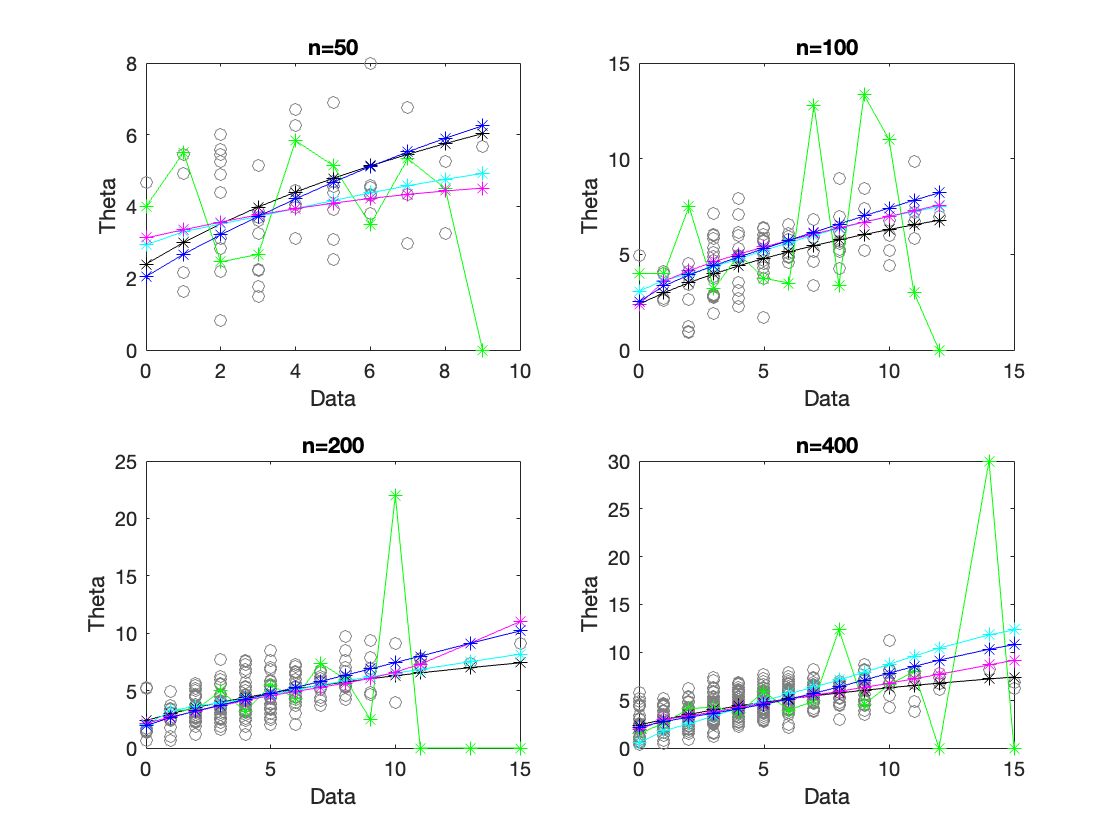}
\caption{Data points $Y$ versus true parameters $\theta$ (grey), and estimates: Oracle (black), Robbins (green), ML (cyan), MHD (magenta) and QB (blue).}
\label{weib_fig}
\end{figure}

QB performs competitively with the ML and MHD, whereas all three methods outperform Robbins' method, whose instability is well documented \citep{She(24)}. A main advantage of QB is computational efficiency. Unlike the other methods, which recompute the estimate from scratch whenever a new observation $Y_{n+1}$ is incorporated, QB updates the estimate using only $Y_{n+1}$. As a result, the computational cost per update is constant in the sample size $n$ and depends only on the number of quadrature points $d$. For $d=1,000$, QB requires $0.0019$ seconds per update, with CPU time increasing linearly in $d$; see \ref{app5_num1}. By comparison, at $n=400$, ML and MHD initialized on a grid of $d=1,000$ points require $1.84$ and $2.14$ seconds, respectively. All computations are performed on a MacBook Pro (M1 processor).

A more detailed and exhaustive synthetic-data analysis is reported in \ref{app5_num12}--\ref{app5_num15}. This includes a study of the credible interval $I_{G_n}$ in \eqref{qbeb_interval} and a comparison with the corresponding oracle interval, together with further experiments under other choices of the prior $G$, such as the Uniform, half-Gaussian, and square-root of half-Cauchy distributions.

\subsection{Real-data analysis}

We analyze a dataset of tweets that received at least $50$ retweets on Twitter between October 7 and November 6, 2011, available at \url{https://snap.stanford.edu/seismic/}. It contains $166{,}783$ tweets generating $34{,}617{,}704$ retweets, together with metadata such as tweet ID, posting time, retweet time, and follower counts. For each tweet, we consider the number of retweets received within the first $30$ seconds after posting, as rapid retweeting may be more indicative of automated than human activity. The retweet intensity is a latent measure of instantaneous virality, or propensity to attract immediate attention. Detecting tweets with unusually high intensity in real time may help identify automated activity or early virality, thus motivating a sequential method such as QB. Table \ref{tweet_tab} summarizes the fast-retweet counts, e.g., $40{,}259$ tweets received none and $28{,}339$ received one, and reports Robbins, ML, MHD, and QB estimates of tweet intensity.

\begin{table}[ht]
\centering
\caption{Counts $n_{y}$ of the number of tweets with $y$ fast retweets, and corresponding estimates.}
\begin{tabular}[t]{lccccccccccc}
\hline
\hline
$y$ &&& 0 & 1 & 2 & 3 & 4 & 5 & 6 & 7   \\[0.1cm]
$n_{y}$ &&& 40,259 & 28,339 & 21,581 & 16,479 & 12,130 & 9,238 & 7,193 & 5,464\\[0.1cm]
\hline 
Robbins&&& 0.704 & 1.523 & 2.291 & 2.944 & 3.808 & 4.672&5.317 & 6.324  \\
ML&&& 0.391 & 2.134 & 2.931 & 3.697 & 4.459 & 5.329 &6.334 &  7.424 \\
MHD&&& 0.750 & 2.198 & 3.142 & 4.149 & 4.986 &5.694 & 6.390& 7.163  \\[0.2cm]
QB&&& 0.607 & 0.818 & 1.632 & 3.076 & 4.007 &  4.519&  5.023& 5.707 \\[0.1cm]
\hline
\hline
\end{tabular}
\label{tweet_tab}
\end{table}

Further analyses are reported in \ref{app5_num2}-\ref{app5_num3}, including an additional study of the Twitter dataset focused on the followers generated by retweets, and an application to the benchmark accident dataset \citep[Table 6.1]{Efr(21)}.


\section*{Supplementary Material}

The Supplementary Material contains auxiliary results in \ref{App_as}; proofs of Theorems \ref{th:cid}--\ref{th:regret} and additional frequentist guarantees in \ref{App_proof}-\ref{app:freqfinite}; further numerical studies for synthetic and real data in \ref{app5_num1}-\ref{app5_num3}; a multidimensional extension of the quasi-Bayes empirical Bayes methodology in  \ref{sec:multiple}.



\section*{Appendix}

\renewcommand{\thesection}{\Alph{section}}
\renewcommand{\theequation}{\thesection.\arabic{equation}}

\setcounter{section}{0}
\setcounter{equation}{0}
\setcounter{thm}{0}

\section{Preliminaries}\label{App_as}

\subsection{Almost sure conditional convergence.} 

Given a probability space $(\Omega,\mathcal F,\PP)$, a  kernel on $\mathbb R^d$ is a function $K:\mathcal B(\mathbb R^d)\times\Omega$ satisfying: i) for every $\omega\in\Omega$, $K(\cdot,\omega)$ is a probability measure on $\mathcal B(\mathbb R^d)$; ii) for each $B\in\mathcal B(\mathbb R^d)$, the function $K(B,\cdot)$ is measurable with respect to $\mathcal F$. A kernel $K$ is called a Gaussian kernel if $K(\cdot,\omega)$ is a Gaussian distribution with mean vector $\mu(\omega)$ and covariance matrix  $\Sigma(\omega)$, where $\mu$ and $\Sigma$ are random variables defined on $(\Omega,\mathcal F,\PP)$. We denote the Gaussian kernel by $\mathcal N(\mu,\Sigma)$. 

\begin{definition}[\cite{crimaldi2009}]
    Let $(X_n)$ be a sequence of random variables  adapted to a filtration $(\mathcal G_n)$ and taking values in $\mathbb R^d$. For every $n\geq 0$, let $K_n$ denote a regular version of the conditional distribution of $X_{n+1}$, given $\mathcal G_n$. If there exists a kernel $K$ such that the sequence $(K_n(\cdot,\omega))_{n}$ converges weakly to $K(\cdot,\omega)$ for almost every $\omega\in\Omega$, then we say that the sequence $(X_n)$ converges to $K$ in the sense of almost-sure conditional convergence with respect to the filtration $(\mathcal G_n)$.
\end{definition}
Whenever the underlying filtration is clear from the context, we denote the almost sure conditional convergence of $X_n$ to $K$ by $$X_n\stackrel{a.s.\; cond.}{\longrightarrow}K.$$

The following results from \cite{crimaldi2009} and \cite{crimaldi2016} provide the main tool in the proof of Theorem \ref{th:clt2}.

\begin{thm}[\cite{crimaldi2009}, Theorem A.1]
\label{th:app:almost sure conditional}
For each $n\geq 1$, let $(M_{n,j})_{j\geq 1}$ be a real valued martingale with respect to the filtration $(\mathcal F_{n,j})_{j\geq 1}$, satisfying $M_{n,0}=0$, and converging in $L^1$ to a random variable $M_{n,\infty}$. Set
$$
X_{n,j}=M_{n,j}-M_{n,j-1}\mbox{ for }j\geq 1, \quad U_n=\sum_{j\geq 1}X_{n,j}^2,\quad X_n^*=\sup_{j\geq 1}|X_{n,j}|.
$$
Assume that
\begin{itemize}
    \item[(a)] $(X_n^*)_n$ is dominated in $L^1$ and converges to zero a.s.
    \item[(b)] $(U_n)_n$ converges a.s. to a non-negative random variable $U$.
\end{itemize}
Then the sequence $(M_{n,\infty})$ 
                    converges to the Gaussian kernel $\mathcal N(0,U)$ in the sense of almost-sure conditional convergence with respect to the filtration $(\mathcal F_{n,1})$.
\end{thm}

\begin{remark}
   Requesting that the almost sure limit of $(U_n)$ is non-negative guarantees the well-definedness of the kernel $\mathcal N(0,U)$ for every $\omega \in \Omega$. Theorem A.1 in {\rm \cite{crimaldi2009}}  specifies that $U$ should be a \lq\lq positive random variable''. Although not explicitly stated, this requirement should be understood as $U \geq 0$, as becomes evident upon a careful examination of the proof.
\end{remark}

\begin{thm}[\cite{crimaldi2016}, Lemma 4.1]
\label{th:app:almost sure2}
    Let $(Z_n)$ be a sequence of real valued random variables adapted to the filtration $(\mathcal G_n)$ and such that $E(Z_{n+1}\mid \mathcal G_n)\rightarrow Z$ a.s. for some random variable $Z$. Moreover, let $(a_n)$ and $(b_n)$ be sequences of real numbers such that
    $$
    b_n\uparrow +\infty,\quad \sum_{k=1}^\infty \frac{E(Z_k^2)}{a_k^2b_k^2}<+\infty.
    $$
Then we have:
\begin{itemize}
    \item[(a)] If $\frac{1}{b_n}\sum_{k=1}^n \frac{1}{a_k}\rightarrow c$ for some constant $c$, then $\frac{1}{b_n}\sum_{k=1}^n \frac{Z_k}{a_k}\rightarrow cZ$ a.s.
    \item[(b)] If $b_n\sum_{k\geq n}\frac{1}{a_kb_k^2}\rightarrow c$ for some constant $c$, then ${b_n}\sum_{k\geq n} \frac{Z_k}{a_kb_k^2}\rightarrow c Z$ a.s.
\end{itemize}
\end{thm}

\subsection{Delta method for almost-sure conditional convergence}

In this section, we extend the delta method to almost sure conditional convergence.
Let $\hb$ be a function defined on a convex set $C\in\mathcal B(\mathbb R^d)$ with values in $\mathbb R^p$. The function $\hb$ is said to be uniformly differentiable on $ C $ if there exists, for every $\vb^*\in C $ a linear function $\Db\hb_{\vb^*}: C \rightarrow \mathbb R^p$ such that the following condition holds: for every $\epsilon >0$ there exists $\delta$ such that
\begin{equation}\label{eq:unifdiff}
\frac{||\hb(\vb)-\hb(\vb^*)-\Db\hb_{\vb^*}(\vb-\vb^*)||}{||\vb-\vb^*||}<\epsilon \quad \mbox{ for every }\quad \vb\neq \vb^*,||\vb-\vb^*||<\delta.
\end{equation}

\begin{thm}\label{th:delta}
           Let $(\vb_n)$ be a sequence of random vectors taking values in a convex set $ C \in\mathcal B(\mathbb R^d)$ with values in $\mathbb R^p$, adapted to the filtration $(\mathcal G_n)$, and let $\tilde \vb$ be a random vector taking values in $ C $. Assume that there exist an increasing and diverging sequence $(r_n)$ of positive numbers and a kernel $K(\cdot,\omega)$ tight with respect to $\omega$, $\PP$-a.s., such that $r_n(\vb_n-\tilde \vb)$ converges in the sense of almost sure conditional convergence to $K(\cdot,\omega)$. If $\hb$ is a uniformly differentiable function on $ C $, then $r_n(\hb(\vb_n)-\hb(\tilde \vb))$ converges in the sense of almost sure conditional convergence to the kernel $K(\Db\hb_{\tilde \vb}^{-1}(\cdot),\omega)$.
\end{thm}

The proof Theorem \ref{th:delta} is split into several lemma.
 
\begin{lem}\label{lem:delta1}
Under the assumptions of Theorem \ref{th:delta},
for every $\epsilon >0$ 
$$
\PP\left(\frac{||\hb(\vb_n)-\hb(\tilde \vb)-\Db\hb_{\tilde \vb}(\vb_n-\tilde \vb)||}{||\vb_n-\tilde \vb||}>\epsilon \mid\mathcal G_n\right)\rightarrow 0 \quad \PP\mbox{-a.s.}
$$
as $n\rightarrow\infty$.
\end{lem}
\begin{proof}
Fix $\epsilon >0$ and let $\delta$ be such that \eqref{eq:unifdiff} holds. Then
    \begin{align*}
        \PP\left(\frac{||\hb(\vb_n)-\hb(\tilde \vb)-D\hb_{\tilde \vb}(\vb_n-\tilde \vb)||}{||\vb_n-\tilde \vb||}>\epsilon \mid\mathcal G_n\right)&\leq 
        \PP(||\vb_n-\tilde \vb||\geq \delta\mid\mathcal G_n)\\&\leq 
        \PP(r_n||\vb_n-\tilde \vb||\geq r_n\delta\mid\mathcal G_n).
           \end{align*}
Fix $\eta$, and let $M$ be such that $K([-M,M]^d,\omega)>1-\eta$ for almost every $\omega$. Then, there exists $n_0=n_0(\omega)$ such that, for $\omega$ in a set of probability one, and $n>n_0(\omega)$,
$$
\PP(r_n||\vb_n-\tilde\vb||\geq M\mid\mathcal G_n)(\omega)<2\eta.
$$
It follows that
$$
\limsup_n  \PP\left(\frac{||\hb(\vb_n)-\hb(\tilde \vb)-\Db\hb_{\tilde \vb}(\vb_n-\tilde \vb)||}{||\vb_n-\tilde \vb||}>\epsilon \mid\mathcal G_n\right)<2\eta\quad\quad \PP\mbox{-a.s.}
$$
The claim follows from the arbitrariness of $\eta$ in a countable dense subset of $\mathbb R^+$.
\end{proof}

\begin{lem}\label{lem:delta2}
    Under the assumptions of Theorem \ref{th:delta}, for every $\epsilon>0$,
    $$
    \PP\left(||r_n(\hb(\vb_n)-\hb(\tilde \vb))-r_n\Db\hb_{\tilde \vb}(\vb_n-\tilde \vb)||>\epsilon\mid\mathcal G_n\right)\rightarrow 0\quad \PP\mbox{-a.s.}
    $$
    as $n\rightarrow\infty$.
\end{lem}
\begin{proof}
It can be proved as in Lemma \ref{lem:delta1} that, for every $\eta>0$, there exists $n_0=n_0(\omega)$ such that, for $\omega$ in a set of probability one and $n\geq n_0(\omega)$,
$$
\PP(r_n||\vb_n-\tilde \vb||\geq M\mid\mathcal G_n)(\omega)<\eta.
$$
Thus, for $n>n_0(\omega)$, it holds, $\PP$-a.s. that
    \begin{align*}
    & \PP\left(||r_n(\hb(\vb_n)-\hb(\tilde \vb))-r_n\Db\hb_{\tilde \vb}(\vb_n-\tilde \vb)||>\epsilon\mid\mathcal G_n\right)(\omega)\\
       &\quad\leq \PP\left(\frac{||\hb(\vb_n)-\hb(\tilde \vb)-\Db\hb_{\tilde \vb}(\vb_n-\tilde \vb)||}{||\vb_n-\tilde \vb||}\;r_n||\vb_n-\tilde \vb||>\epsilon\mid\mathcal G_n \right)(\omega)\\
       &\quad \leq 
        \PP(r_n||\vb_n-\tilde \vb||\geq M\mid\mathcal G_n)(\omega)+ \PP\left(M\frac{||\hb(\vb_n)-\hb(\tilde \vb)-\Db\hb_{\tilde \vb}(\vb_n-\tilde \vb)||}{||\vb_n-\tilde \vb||}>\epsilon \mid\mathcal G_n\right)(\omega)\\
        &\quad\leq 
        \eta+\PP\left(\frac{||\hb(\vb_n)-\hb(\tilde \vb)-\Db\hb_{\tilde \vb}(\vb_n-\tilde\vb)||}{||\vb_n-\tilde\vb||}>\frac{\epsilon}{M} \mid\mathcal G_n\right)(\omega).
           \end{align*}
The claim follows from Lemma \ref{lem:delta1} and the arbitrariness of $\eta$.

\end{proof}

\begin{proof}[of Theorem \ref{th:delta}]
    Let $N^c$ be an event such that $P(N^c)=1$, $\{K(\cdot,\omega):\omega\in N^c\}$ is tight, $\PP(r_n(\vb_n-\tilde \vb)\in\cdot\mid\mathcal G_n)(\omega)\rightarrow K(\cdot,\omega) $ and 
    $\PP(||r_n(\hb(\vb_n)-\hb(\tilde\vb))-r_n\Db\hb_{\tilde \vb}(\vb_n-\tilde\vb)||>\epsilon\mid\mathcal G_n)(\omega)\rightarrow 0$ 
    for every $\epsilon >0$. Fix $\omega\in N^c$ and let $A\in\mathcal B(\mathbb R^p)$ be such that $K(\Db\hb_{\tilde \vb}^{-1}(\partial A),\omega)=0$. For every $\epsilon>0$, let $A^\epsilon=\{\xb\in \mathbb R^p:{\rm d}(\xb,A)\leq \epsilon\}$. Then
    \begin{align*}
        &\limsup_n \PP(r_n(\hb(\vb_n)-\hb(\tilde \vb))\in A\mid\mathcal G_n)(\omega)\\&\quad\leq 
        \limsup_n \PP(||r_n(\hb(\vb_n)-\hb(\tilde \vb))-r_n\Db\hb_{\tilde\vb}(\vb_n-\tilde\vb)||>\epsilon\mid\mathcal G_n)(\omega)\\
        &\quad\quad\quad\quad+\limsup_n\PP(\Db\hb_{\tilde\vb}(r_n(\vb_n-\tilde\vb))\in A^\epsilon\mid\mathcal G_n)(\omega)\\
         &\quad\leq K(\Db\hb_{\tilde\vb}^{-1}(A^\epsilon),\omega).
            \end{align*}
            Analogously,
             \begin{align*}
&\limsup_n \PP(r_n(\hb(\vb_n)-\hb(\tilde\vb))\in A^c\mid\mathcal G_n)(\omega)\leq K(\Db\hb_{\tilde\vb}^{-1}((A^c)^\epsilon),\omega).
             \end{align*}
             Thus,
\begin{align*}
&\liminf_n \PP(r_n(\hb(\vb_n)-\hb(\tilde\vb))\in A\mid\mathcal G_n)(\omega)\geq K(\Db\hb_{\tilde\vb}^{-1}(A_\epsilon),\omega),
             \end{align*}
             where $A_\epsilon$ is the complement of $(A^c)^\epsilon$.
             Thus, for every $\epsilon>0$,
\begin{align*}
    K(\Db\hb_{\tilde\vb}^{-1}(A_\epsilon),\omega)&\leq 
    \liminf_n \PP(r_n(\hb(\vb_n)-\hb(\tilde\vb))\in A\mid\mathcal G_n)(\omega)\\&\leq 
    \limsup_n \PP(r_n(\hb(\vb_n)-\hb(\tilde\vb))\in A\mid\mathcal G_n)(\omega)\\&\leq K(\Db\hb_{\tilde\vb}^{-1}(A^\epsilon),\omega).
\end{align*}           
Since $K(\Db\hb_{\tilde\vb}^{-1}(\partial A,\omega)=0$, then $ K(\Db\hb_{\tilde\vb}^{-1}(A^\epsilon),\omega)-K(\Db\hb_{\tilde\vb}^{-1}(A_\epsilon),\omega)\rightarrow 0$, as $\epsilon\rightarrow 0$. Thus
$\lim_n \PP(r_n(\hb(\vb_n)-\hb(\tilde\vb))\in A\mid\mathcal G_n)(\omega)= K(\Db\hb_{\tilde\vb}^{-1}(A),\omega).$
    \end{proof}


\section{Proofs}\label{App_proof}

Throughout, we denote by $k(\cdot\,|\,\theta)$ the Poisson kernel with mean $\theta>0$, i.e. $k(y\,|\,\theta)=\theta^{y}\text{e}^{-\theta}/y!$ for $y\in\mathbb{N}_{0}$.

\subsection{Proof of Theorem \ref{th:cid}}\label{app3_0}
By Theorem 1  in \cite{For(20)}, the sequence $(G_n)_{n\geq 0}$ is a martingale measure with respect to the filtration $(\mathcal G_n)_{n\geq 0}$, converging weakly $\PP$-a.s. to a random probability measure $\tilde G$. Furthermore, by Proposition 1 in \cite{For(20)}, for every $y\in\mathbb N_0$, $(p_{G_n}(y))_{n\geq 0}$ is a martingale with respect to to $(\mathcal G_n)_{n\geq 0}$ converging $\PP$-a.s.
to $p_{\tilde G}(y)=\int_\Theta k(y\mid\theta)\tilde G(\ddr\theta)$.
 It follows that, for every $y\in\mathbb N_0$, every $n\geq 0$ and every $k\geq 1$,
$$
\PP(Y_{n+k}=y\mid \mathcal G_n)=
E(p_{G_{n+k-1}}(y)\mid\mathcal G_n)=p_{G_n}(y)=E(p_{\tilde G}(y)\mid\mathcal G_n),
$$
which implies that $\PP(Y_n=y)=\int p_{\tilde G}(y)d\PP$, and that the sequence $(Y_n)$ is conditionally identically distributed with respect to the filtration $(\mathcal G_n)$ \citep{Berti(04)}. By Lemma 2.1 and Theorem 2.2 in \cite{Berti(04)}, with $f(Y_n)=I(Y_n=y)$, $p_{\tilde G}(y)$ is the $\PP$-a.s. limit of $\frac 1 n \sum_{k=1}^n I(Y_k=y)$.

\subsection{Proof of Theorem \ref{th:clt2}}\label{app3}

The proof is based on [\cite{crimaldi2009}, Theorem A.1], [\cite{crimaldi2016}, Lemma 4.1]
(see theorems \ref{th:app:almost sure conditional} and \ref{th:app:almost sure2}) and on the Delta method for almost sure conditional convergence (see Theorem \ref{th:delta}).
Define, for every $y\in\mathbb N_0$ and $n\geq 0$,
$$
\Delta_n(y)=p_{G_n}(y)-p_{G_{n-1}}(y)=\alpha_n\left[\int_\Theta k(y\mid\theta)G_{n-1}(\ddr\theta\mid Y_n)-\int_\Theta k(y\mid\theta)G_{n-1}(\ddr\theta)\right],
$$
and fix a vector $\ub=[u_1,u_2]$ such that $||\ub||=1$. For every $n\geq 0$, let 
$$M_{n,0}=0, \quad M_{n,j}=b_n^{1/2}\sum_{i=1}^{j-1} (u_1\Delta_{n+i}(y)+u_2 \Delta_{n+i}(y+1))\quad (j\geq 1),$$
 $$
       \mathcal F_{n,0}=\mathcal G_n,\quad \mathcal F_{n,j}=\mathcal G_{n+j-1}\mbox{ for }j\geq 1,
       $$
For every $n\geq 0$, the sequence $(M_{n,j})_{j\geq 0}$ is a martingale with respect to $(\mathcal F_{n,j})_{j\geq 0}$, converging in $L^1$ to 
$$M_{n,\infty}=b_n^{1/2}\left[ u_1(p_{\tilde G}(y)-p_{G_n}(y))+u_2(p_{\tilde G}(y+1)-p_{G_n}(y+1))\right].$$
Set
$$
X_{n,j}=M_{n,j}-M_{n,j-1}\mbox{ for }j\geq 1, \quad U_n=\sum_{j\geq 1}X_{n,j}^2,\quad X_n^*=\sup_{j\geq 1}|X_{n,j}|.
$$
Then $X_{n,1}=0$ and, for $j\geq 2$,
\begin{align*}
    X_{n,j}=b_n^{1/2}\alpha_{n+j-1}\left[u_1(p_{G_{n+j-2}}(y\mid Y_{n+j-1})-p_{G_{n+j-2}}(y))
    +u_2(p_{G_{n+j-2}}(y+1\mid Y_{n+j-1})-p_{G_{n+j-2}}(y+1))\right],
  \end{align*}
  where, for every $G$, $y$ and $z$, $$p_G(y\mid z)=\int_\Theta k(y\mid\theta)G(\ddr\theta\mid z)\quad\mbox{ with }\quad G(\ddr\theta\mid z)=\frac{k(z\mid\theta)G(\ddr\theta)}{\int_\Theta k(z\mid\theta)G(\ddr\theta)}.$$
Since
$
|X_n^*|\leq b_n^{1/2}\sup_{j\geq 1}\alpha_{n+j}
$, then $X_n^*$ is dominated in $L^1$ and converges to zero $\PP$-a.s.
Furthermore
$$
U_n=b_n\sum_{k>  n}\alpha_{k}^2 Z_k,
$$
with
$$
Z_k=\left[u_1(p_{G_{k-1}}(y\mid Y_{k})-p_{G_{k-1}}(y))
    +u_2(p_{G_{k-1}}(y+1\mid Y_{k})-p_{G_{k-1}}(y+1))\right]^2
$$
For every $k\geq 1$, let $a_k=(b_k\alpha_k)^{-2}$. The sequence $(Z_k)_{k\geq 1}$ is adapted to the filtration $(\mathcal G_k)_{k\geq 1}$ and $\EE(Z_{k+1}\mid \mathcal G_k)$ converges $\PP$-a.s. to $Z=\ub^T\Sigma \ub$, with 
    \begin{align*}
        &\Sigma_{1,1}=\sum_{z\in\mathbb N_0}p_{\tilde G}(y\mid z)^2p_{\tilde G}(z)-p_{\tilde G}(y)^2,\\ &\Sigma_{1,2}=\Sigma_{2,1}=\sum_{z\in\mathbb N_0}p_{\tilde G}(y\mid z)p_{\tilde G}(y+1\mid z)p_{\tilde G}(z)-p_{\tilde G}(y)p_{\tilde G}(y+1),\\
        &\Sigma_{2,2}=\sum_{z\in\mathbb N_0}p_{\tilde G}(y+1\mid z)^2p_{\tilde G}(z)-p_{\tilde G}(y+1)^2.
      \end{align*}
Furthermore,
$$
\sum_{k=1}^\infty \frac{E(Z_k^2)}{a_k^2b_k^2}\leq \sum_{k=1}^\infty b_k^2\alpha_k^4<\infty.
$$
Since  $b_n\sum_{k\geq  n}\frac{1}{a_kb_k^2}\rightarrow 1$, then by Theorem \ref{th:app:almost sure2}, $$b_n\sum_{k\geq n}\alpha_k^2Z_k=b_n\sum_{k\geq   n} \frac{Z_k}{a_kb_k^2}\rightarrow  Z\quad \PP\mbox{-a.s.}$$
It follows that
\begin{align*}
    U_n=b_n\sum_{k>   n} \alpha_k^2Z_k\rightarrow  Z\quad \PP\mbox{-a.s.}
\end{align*}
since $b_n\alpha_n^2Z_n\rightarrow  0$  $\PP$-a.s. (as $\sum_{n\geq 1} \alpha_n^4b_n^2<\infty$ by assumption and $\sup_n|Z_n|\leq 1$). 
 By Theorem \ref{th:app:almost sure conditional} and the arbitrariness of $\ub$ in a countable dense subset of the unit sphere,
$$
b_n^{1/2}\left(
\left[ 
\begin{array}{c}
p_{G_n}(y)\\
p_{G_n}(y+1)
\end{array}
\right]-
\left[ 
\begin{array}{c}
p_{\tilde G}(y)\\
p_{\tilde G}(y+1)
\end{array}
\right]\right)
\stackrel{a.s.\;cond}{\rightarrow}\mathcal N(0,\Sigma).
$$
To conclude, we apply the Delta method for almost sure conditional convergence (see Theorem \ref{th:delta}). We can write that 
$$
\hat\theta_G(y)=h(p_G(y),p_G(y+1)),
$$
with $h(v_1,v_2)=(y+1)v_2/v_1$. The function $h$ is continuously differentiable in the convex set $C=(0,1]\times (0,1]$ and the kernel $\mathcal N(0,\Sigma)$ is $\PP$-a.s. tight since $\Sigma_{1,1}$ and $\Sigma_{2,2}$ are bounded by 1. By Theorem \ref{th:delta} and standard calculus we obtain
$$
b_n^{1/2}(\hat\theta_{G_n}-\hat\theta_{\tilde G}(y))\stackrel{a.s.\;cond.}{\longrightarrow}\mathcal N(0, W_{\tilde G}(y)),
$$
with 
$$
W_{\tilde G}(y)=\hat\theta_{\tilde G}(y)^2\sum_zp_{\tilde G}(z)\left(\frac{p_{\tilde G}(y\mid z)}{p_{\tilde G}(y)}-\frac{p_{\tilde G}(y+1\mid z)}{p_{\tilde G}(y+1)}\right)^2.
$$

 \subsection{Proof of Theorem \ref{th:clt3}}\label{app3new}
The proof is based on \cite{crimaldi2009}, Theorem A.1, \cite{crimaldi2016}, Lemma 4.1
(see theorems \ref{th:app:almost sure conditional} and \ref{th:app:almost sure2}) and on the Delta method for almost sure conditional convergence (see Theorem \ref{th:delta}).
We first notice that, since $\tilde G$ and $\tilde G_n$ ($n\geq 1$) have compact support, then for every $y$ there exists $\delta(y)>0$ such that $p_{\tilde G}(y)>\delta(y) $ and $p_{G_n}(y)<\delta(y)$ for every $n\geq 1$.
We can write for every mixing distribution $G$
$$
F_G(\theta\mid y)=\frac{A_{G}(\theta\mid y)}{p_G(y)}
$$
where
$$
A_{G}(\theta\mid y)=\int_0^\theta k(y\mid\vartheta)G(\ddr\vartheta),
$$
with $k(y\mid\vartheta)=e^{-\vartheta}\vartheta^y/y!$. Define, for every $y\in\mathbb N_0$ and $n\geq 0$,
$$
\Delta p_n(y)=p_{G_n}(y)-p_{G_{n-1}}(y)=\alpha_n\left[\int_\Theta k(y\mid\theta)G_{n-1}(\ddr\theta\mid Y_n)-\int_\Theta k(y\mid\theta)G_{n-1}(\ddr\theta)\right],
$$
and
$$
\Delta A_n(\theta\mid y)=A_{G_n}(\theta\mid y)-A_{G_{n-1}}(\theta\mid y)=\alpha_n\left[  
\int_0^\theta k(y\mid\vartheta)G_{n-1}(\ddr\vartheta\mid Y_n)-\int_0^\theta k(y\mid\vartheta)G_{n-1}(\ddr\vartheta).
\right]
$$
Fix a vector $\ub=[u_1,u_2]$ such that $||\ub||=1$. For every $n\geq 0$, let 
$$M_{n,0}=0, \quad M_{n,j}=b_n^{1/2}\sum_{i=1}^{j-1} (u_1\Delta A_{{n+i}}(\theta\mid y)+u_2 \Delta p_{n+i}(y))\quad (j\geq 1),$$
 $$
       \mathcal F_{n,0}=\mathcal G_n,\quad \mathcal F_{n,j}=\mathcal G_{n+j-1}\mbox{ for }j\geq 1,
       $$
For every $n\geq 0$, the sequence $(M_{n,j})_{j\geq 0}$ is a martingale with respect to $(\mathcal F_{n,j})_{j\geq 0}$, converging in $L^1$ to 
$$M_{n,\infty}=b_n^{1/2}\left[ u_1(A_{\tilde G}(\theta\mid y)-A_{G_n}(\theta\mid y))+u_2(p_{\tilde G}(y)-p_{G_n}(y))\right].$$
Set
$$
X_{n,j}=M_{n,j}-M_{n,j-1}\mbox{ for }j\geq 1, \quad U_n=\sum_{j\geq 1}X_{n,j}^2,\quad X_n^*=\sup_{j\geq 1}|X_{n,j}|.
$$
Then $X_{n,1}=0$ and, for $j\geq 2$,
\begin{align*}
    X_{n,j}=b_n^{1/2}\alpha_{n+j-1}\left[u_1(A_{G_{n+j-2}(\cdot\mid Y_{n+j-1})}(\theta\mid y)-A_{G_{n+j-2}}(\theta\mid y))
    +u_2(p_{G_{n+j-2}(\cdot\mid  Y_{n+j-1}) }(y)-p_{G_{n+j-2}}(y))\right],
  \end{align*}
  where, for every $G$, $y$ and $z$, 
  $$A_{G(\cdot\mid z)}(\theta\mid y)=\int_0^\theta k(y\mid\vartheta)G(\ddr\vartheta\mid z),\quad 
  p_{G(\cdot\mid z)}(y)=\int_\Theta k(y\mid\theta)G(\ddr\theta\mid z)\quad\mbox{ with }\quad G(\ddr\theta\mid z)=\frac{k(z\mid\theta)G(\ddr\theta)}{\int_\Theta k(z\mid\theta)G(\ddr\theta)}.$$
Since
$
|X_n^*|\leq b_n^{1/2}\sup_{j\geq 1}\alpha_{n+j}
$, then $X_n^*$ is dominated in $L^1$ and converges to zero $\PP$-a.s.
Furthermore
$$
U_n=b_n\sum_{k>  n}\alpha_{k}^2 Z_k,
$$
with
$$
Z_k=\left[u_1(A_{G_{k-1}(\cdot\mid Y_k)}(\theta\mid y)-A_{G_{k-1}}(\theta\mid y))
    +u_2(p_{G_{k-1}(\cdot\mid Y_k)}(y)-p_{G_{k-1}}(y))\right]^2
$$
For every $k\geq 1$, let $a_k=(b_k\alpha_k)^{-2}$. The sequence $(Z_k)_{k\geq 1}$ is adapted to the filtration $(\mathcal G_k)_{k\geq 1}$ and $\EE(Z_{k+1}\mid \mathcal G_k)$ converges $\PP$-a.s. to $Z=\ub^T\Sigma \ub$, with 
    \begin{align*}
        &\Sigma_{1,1}=\sum_{z\in\mathbb N_0}A_{\tilde G(\cdot\mid z)}(\theta\mid y)^2p_{\tilde G}(z)-A_{\tilde G}(\theta\mid y)^2,\\ 
        &\Sigma_{1,2}=\Sigma_{2,1}=\sum_{z\in\mathbb N_0}A_{\tilde G(\cdot\mid z)}(\theta\mid y)p_{\tilde G(\cdot\mid z)}(y)p_{\tilde G}(z)-A_{\tilde G}(\theta\mid y)p_{\tilde G}(y),\\
        &\Sigma_{2,2}=\sum_{z\in\mathbb N_0}p_{\tilde G(\cdot\mid z)}(y)^2p_{\tilde G}(z)-p_{\tilde G}(y)^2.
      \end{align*}
Furthermore,
$$
\sum_{k=1}^\infty \frac{E(Z_k^2)}{a_k^2b_k^2}\leq 2\sum_{k=1}^\infty b_k^2\alpha_k^4<\infty.
$$
Since  $b_n\sum_{k\geq  n}\frac{1}{a_kb_k^2}\rightarrow 1$, then by Theorem \ref{th:app:almost sure2}, $$b_n\sum_{k\geq n}\alpha_k^2Z_k=b_n\sum_{k\geq   n} \frac{Z_k}{a_kb_k^2}\rightarrow  Z\quad \PP\mbox{-a.s.}$$
It follows that
\begin{align*}
    U_n=b_n\sum_{k>   n} \alpha_k^2Z_k\rightarrow  Z\quad \PP\mbox{-a.s.}
\end{align*}
since $b_n\alpha_n^2Z_n\rightarrow  0$  $\PP$-a.s. (as $\sum_{n\geq 1} \alpha_n^4b_n^2<\infty$ by assumption and $\sup_n|Z_n|\leq 1$). 
 By Theorem \ref{th:app:almost sure conditional} and the arbitrariness of $\ub$ in a countable dense subset of the unit sphere,
$$
b_n^{1/2}\left(
\left[ 
\begin{array}{c}
A_{G_n}(\theta\mid y)\\
p_{G_n}(y)
\end{array}
\right]-
\left[ 
\begin{array}{c}
A_{\tilde G}(\theta\mid y)\\
p_{\tilde G}(y)
\end{array}
\right]\right)
\stackrel{a.s.\;cond}{\rightarrow}\mathcal N(0,\Sigma).
$$
To conclude, we apply the Delta method for almost sure conditional convergence (see Theorem \ref{th:delta}). We can write that 
$$
F_G(\theta\mid y)=h(A_G(\theta\mid y),p_G(y)),
$$
with $h(v_1,v_2)=v_1/v_2$. The function $h$ is uniformly differentiable in the convex set $C=[0,1]\times [\delta(y),1]$ and the kernel $\mathcal N(0,\Sigma)$ is $\PP$-a.s. tight since $\Sigma_{1,1}$ and $\Sigma_{2,2}$ are bounded by 1. By Theorem \ref{th:delta} and standard calculus we obtain
$$
b_n^{1/2}(F_{G_n}(\theta\mid y)-F_{\tilde G}(\theta\mid y))\stackrel{a.s.\;cond.}{\longrightarrow}\mathcal N(0, W'_{\tilde G}(y)),
$$
with 
\begin{align*}
    W'_{\tilde G}(y)&=
\sum_zp_{\tilde G}(z)\frac{p_{\tilde G(\cdot\mid z)}(y)^2}{p_{\tilde G}(y)^2}\left(F_{\tilde G(\cdot\mid z)}(\theta\mid y)-F_{\tilde G}(\theta\mid y)\right)^2.
\end{align*}

\subsection{Proof of Theorem \ref{th:regret}}\label{app:proofregret}
The proof is based on Corollary~4.6 in \cite{MarTok(09)}. Since $\Theta$ is compact and $k(\cdot\mid\theta)$ is the Poisson kernel, assumptions A1--A4 in \cite{MarTok(09)} hold, except possibly for $\sum_{n\ge1} a_n\alpha_n^2<\infty$ (with $a_n=\sum_{i=1}^n\alpha_i$), which is not used in the proof of Corollary~4.6. Hence,
$
p_{G_n}=\int_\Theta k(\cdot\mid\theta)\,G_n(\mathrm d\theta)$ converges 
$\PP_{G^*}$-a.s. to $p_{G^\dag}=\int_\Theta k(\cdot\mid\theta)\,G^\dag(\mathrm d\theta)$,
both pointwise and in $L^1$.
Since $\Theta$ is compact, for every $y\in\mathbb N_0$ $p_{G_n}(y)$ are bounded away from zero uniformly with respect to $n$. Therefore $p_{G^\dag}(y)\neq 0$ for every $y\in\mathbb N_0$. By the continuous mapping Theorem, for every $y\in\mathbb N_0$,
$\hat\theta_{G_n}(y)=(y+1)p_{G_n}(y+1)/p_{G_n}(y)$ converges $\PP_{G^*}$-a.s. to $\hat\theta_{G^\dag}(y)=(y+1)p_{G^\dag}(y+1)/p_{G^\dag}(y)$. Denoting by $\theta_M=\max_{\theta\in\Theta} \theta$, we can write that 
$(\hat\theta_{G_n}(y)-\hat\theta_{G^\dag}(y))^2 \leq (y+1)^2 \theta_M^2$, which is summable in $\mathbb N_0$ with  respect $p_{G^*}$. Hence,
$$
\textsc{Regret}(G_n,G^\dag)=\sum_{y\in\mathbb N_0}(\hat\theta_{G_n}(y)-\hat\theta_{G^\dag}(y))^2p_{G^*}(y)\rightarrow 0,
$$
$\PP_{G^*}$-a.s., as $n\rightarrow\infty$.


\section{Rate of convergence}\label{app:freqfinite}

Throughout, we denote by $k(\cdot\,|\,\theta)$ the Poisson kernel with mean $\theta>0$, i.e. $k(y\,|\,\theta)=\theta^{y}\text{e}^{-\theta}/y!$ for $y\in\mathbb{N}_{0}$.

\subsection{Grid construction and rate of convergence}\label{app:results}

The proposition below introduces a procedure for constructing a grid $\Theta_\eta$ such that finite mixtures provide a good approximation for mixing distributions subject to moment constraints.
This construction plays a central role in establishing asymptotic frequentist guarantees. We denote by $\textsc{KL}(p||q)$ the Kullback-Leibler divergence between the probability distributions $p$ and $q$.

\begin{prp}\label{prp_grid}
Let the mixing distribution  $G^*$ on $\Theta\subseteq \mathbb R^+$ be such that $\sum_{y\in\mathbb{N}_{0}}y^{k}p_{G^*}(y)\leq m_{k}$ for some $k\in\mathbb{N}$ with $k\geq 2$ and $m_{k}\in\mathbb{R}$. For a fixed $\eta>0$ define the equally spaced grid $\Theta_{\eta}=\{\vartheta_{i}=i\eta\text{ : }i=1,\ldots,d_{\eta}\}$ with 
\begin{displaymath}
d_{\eta}=\inf\left\{n\in\mathbb{N}\text{ : }n>\eta^{-1}\text{ and }n^{1-k}\log(n\eta)m_{k}\leq\eta^{k}\right\}.
\end{displaymath}
If $\mathcal{S}_{\eta}$ is the set of all the probability distributions supported on $\Theta_{\eta}$, then $G^*_{\eta}=\text{argmin}_{G\in\mathcal{S}_{\eta}}\textsc{KL}(p_{G^*}||p_{G})$ is unique and $\textsc{KL}(p_{G^*}||p_{G^*_{\eta}})<2\eta$. 
\end{prp}

See Supplementary Material \ref{app3_00} for the proof of Proposition \ref{prp_grid}. The probability distribution 
$G^*_{\eta}$ in Proposition \ref{prp_grid} is referred to as the $d_{\eta}$-\textsc{KL}-discretization of $G^*$ on $\Theta\subseteq\mathbb{R}^{+}$. This form of discretization is of independent interest in the context of $g$-modeling approaches to the Poisson compound decision problem \citep{Jan(24),She(24)}.

\begin{thm}\label{th:regretdiscr}
Let  $G^{\ast}$ be an oracle prior on $\Theta\subseteq\mathbb{R}^{+}$ such that $\sum_{y\in\mathbb{N}_{0}}y^{k}p_{G^{\ast}}(y)\leq m_{k}$ for some $k\in\mathbb{N}$ with $k\geq 2$ and $m_{k}\in\mathbb{R}$, and for a fixed $\eta>0$ let $G^{\ast}_{\eta}$ be the $d_{\eta}$-\textsc{KL}-discretization of $G^{\ast}$ on $\Theta_{\eta}$. For each $n\geq1$, let $G_{n}$  be the as in \eqref{eq:newton}, with $G_0(\{\vartheta_j\})>0$ for every $j=1,\dots,d_\eta$, and with $\alpha_n=(\alpha+n)^{-\gamma}$ for $\alpha>0$ and $\gamma\in (1/2,1)$. If $\sum_{y\in\mathbb{N}_{0}}d_{\eta}^{2y}p_{G^{\ast}}(y)<+\infty$ then, for every $\delta<2-1/\gamma$
\begin{equation}\label{optim}
    \textsc{Regret}(G_{n},G_{\eta}^{\ast})=o(n^{-\delta})\quad \PP_{G^*}\,\text{-}\,a.s.
\end{equation}
\end{thm}

The proof is in Section \ref{app4}.  Under the ``true" model for the observations $Y_n$'s, with the oracle prior $G^{\ast}$ belonging to the class of light-tail distributions, Theorem \ref{th:regretdiscr} establishes that $\hat{\theta}_{g_{n}}(y)$ is an asymptotically optimal  estimate of $\hat\theta_{G^{\ast}_{\eta}}(y)$. Although Theorem \ref{th:regret} is stated with respect to the standard choice $\alpha_n=(\alpha+n)^{-\gamma}$, alternative learning rates may also be considered \citep{For(20)}.

\subsection{Notation}
We denote vectors and their components by bold letters as
$\vb=[\vb_1,\dots, \vb_k]^T$, scalar product by $\langle\cdot,\cdot\rangle$, component-wise products by $\circ$, and the Euclidean norm by $||\cdot||$. 

A probability mass functions $g$ defined on $\Theta=\{\vartheta_1,\dots,\vartheta_d\}$ can also be interpreted as a vector $\gb$ in the  simplex $\Delta=\{\vb=[\vb_1,\dots,\vb_{d-1}]^T \vb_i\geq 0,(i=1,\dots,d-1),\sum_{i=1}^{d-1}\vb_i\leq 1\}$, or in $\mathbb R^d$.
In both cases $\gb_i=g(\vartheta_i)$ for all $i$. Similarly, 
$k(y\mid \cdot)$ is sometimes regarded as a vector $\kb_y$ in $\mathbb R^d$. We use the notations $\kb_{y,i}$, $k_{y}(\vartheta_i)$ and $k(y\mid \vartheta_i)$ interchangeably to refer to its $i$th coordinate, for $i=1,\dots,d$. 
When the probability mass function $p_G$ is seen as a function of the vector $\gb$, it is denoted by $p_{\gb}$.
The  closure and boundary of a set $A$ are denoted by $\overline A$ and $\partial A$, respectively.

\subsection{Proof of Proposition \ref{prp_grid}}\label{app3_00}

We first prove that there exists $G\in \mathcal S_\eta$ such that 
$
    \textsc{KL}(p_{G^*}||p_G)< 2\eta.
$
Since the $k$-moment of the Poisson distribution with parameter $\lambda$ is greater than $\lambda^k$, then the assumption $\sum_{y\in \mathbb N_0}y^kp_{G^*}(y)\leq m_k$  implies that 
$$
\int_0^\infty \theta^k G^*(\ddr\theta)\leq m_k.
$$
    For every $i=1,\dots,d_\eta$, let $w_i=G^*((\vartheta_{i-1},\vartheta_i])$ with $\vartheta_0=0$. Let $F$ be the  probability measure
    $$
    G=\sum_{i=1}^{d_\eta}w_i \delta_{\vartheta_i}+G^*(\vartheta_{d_\eta},+\infty)\delta_{\vartheta_{d_\eta}}.
    $$
        Then
    $$
    G=
    \sum_{i=1}^{d_\eta}\delta_{\vartheta_i}\int_{I_i}
    G^*(\ddr\theta)+\delta_{\vartheta_{d_\eta}}\int_{I_{d_\eta+1}}G^*(\ddr\theta),
    $$
    where $I_i=(\vartheta_{i-1},\vartheta_i]$ for $i=1,\dots,d_\eta$ and $I_{d_\eta+1}=(\vartheta_{d_\eta},+\infty)$.
    Since the Kullback-Leibler divergence is convex, then, denoting $M=d_\eta \eta$, we can write that
    \begin{align*}
        \textsc{KL}(p_{G^*}||p_{G})&
        \leq \sum_{y\in\mathbb N_0}\sum_{i=1}^{d_\eta}\int_{I_i}\log\left(\frac{e^{-\theta}\theta^y}{e^{-\vartheta_i}\vartheta_i^y}\right)\frac{e^{-\theta}\theta^y}{y!}G^*(\ddr\theta)
        +\sum_{y\in\mathbb N_0}\int_M^\infty \log\left(\frac{e^{-\theta}\theta^y}{e^{-M}M^y}\right)\frac{e^{-\theta}\theta^y}{y!}G^*(\ddr\theta)\\
       &< \int_0^M\eta G^*(\ddr\theta)+\int_M^\infty \sum_{y\in\mathbb N_0}y\log(\theta)\frac{e^{-\theta}\theta^y}{y!}G^*(\ddr\theta)\\
       &\leq \eta +\int_M^\infty \theta\log(\theta)G^*(\ddr\theta)\\
       &\leq \eta+\frac{\log M}{M^{k-1}}\int_M^\infty \theta^kG^*(\ddr\theta)\\
       &\leq 2\eta.
    \end{align*}
   We now prove that the minimization problem
$
\mbox{\sl argmin}_{G\in\mathcal S_\eta}\textsc{KL}(p_G^*||p_G)
$
has a unique solution. 
Representing each probability mass function $g$ with support $\Theta_\eta$ as a vector $\gb$ taking values in the simplex
$
\Delta_\eta=\{\vb=(\vb_1,\dots,\vb_{d_\eta-1}):\vb_i\geq 0,\sum_{i=1}^{d_\eta-1}\vb_i\leq 1\},
$
we can see $\ell (\fb)=\textsc{KL}(p_{G^*}||p_G)$ as a strictly convex function  defined on a convex and compact set. Thus the global minimum $\gb_\eta$ exists and is unique.
    $\square$

\subsection{Proof of Theorem \ref{th:regretdiscr}}\label{app4}

The proof of Theorem \ref{th:regretdiscr} builds on Lemmas \ref{lem:lagrange} and \ref{lem:stochapp1}, together with the convergence rate of $g_n$ toward $g_\eta^*$ under i.i.d.\ observations. This latter problem has been investigated in the literature, although existing results either yield slower convergence  (\cite{MarTok(09)}) or rely on the assumption that $G^*$ and $G_n$ have the same  finite support (\cite{Mar(12)}).

In this section, we denote the probability mass function associated with a probability distribution by the same symbol, but in lowercase. Furthermore, we occasionally represent probability mass functions $g$ on $\Theta_\eta$ as vectors $\gb$, taking values either in the simplex $\Delta_\eta=\{\vb=(\vb_1,\dots,\vb_{d_\eta-1}):\vb_i\geq 0 \;(i=1,\dots,d_\eta-1),\;\sum_{i=1}^{d_\eta-1}\vb_i\leq 1\}$ or in $\mathbb R^{d_\eta}$, depending on the context, which will be specified when needed. Similarly, the Poisson kernel is represented as a vector in $\Delta_\eta$ or $\mathbb R^{d_\eta}$, by defining $\kb_y(\vartheta_j)=k(y\mid\vartheta_j)$. We denote component-wise multiplication of vectors by $\circ$.

\begin{lem}\label{lem:lagrange}
If 
\begin{equation}\label{eq:momentgf}
\sum_{y\in \mathbb N_0}d_\eta^{2y}p_{G^*}(y)<+\infty,
\end{equation}
then the $d_\eta$-\textsc{KL}-discretization $g_\eta^*$ of $G^*$  satisfies  
for every $i=1,\dots,d_\eta$ the equation in $g$:
\begin{equation}\label{eq:minKL}
    g(\vartheta_i)\left(\sum_{y\in\mathbb N_0}\frac{k(y\mid\vartheta_i)}{p_{G}(y)}p_{G^*}(y)-1\right)=0.
\end{equation}
\end{lem}
\begin{proof}
Representing probability mass functions on $\Theta_\eta$ as vectors in
$
\Delta_\eta,
$
we can see $\ell (\fb)=\textsc{KL}(p_{G^*}||p_F)$ as a function on $\Delta_\eta$. 
The assumption  \eqref{eq:momentgf} implies that
\begin{align}\label{eq:condfin}
    \int_0^\infty e^{\theta(d_\eta^2-1)}G^*(d\theta)&=\int_0^\infty \sum_{y\in\mathbb N_0}e^{\log(d_\eta^2)y}e^{-\theta}\frac{\theta^y}{y!}G^*(d\theta)\\
    &=\sum_{y\in\mathbb N_0}e^{2y\log d_\eta}p_{G^*}(y)<+\infty.\nonumber
\end{align}
Thus, for every $i=1,\dots,d_\eta$,
$$
\frac{\partial}{\partial \gb_i}\ell (\gb)=\sum_{y\in\mathbb N_0}\frac{k(y\mid\vartheta_i)}{p_{G}(y)}p_{G^*}(y)
\leq e^{\vartheta_1-\vartheta_{d_\eta}}\int e^{\theta(d_\eta-1)}G^*(\ddr\theta)<+\infty
$$
by \eqref{eq:condfin}
    Analogously, for every $i=1,\dots,d_\eta$, 
    \begin{equation}\label{eq:finite}
            \sum_{y\in\mathbb N_0}\frac{k(y\mid\vartheta_i)^2}{p_{G}(y)^2}p_{G^*}(y)
            \leq e^{\vartheta_1-\vartheta_{d_\eta}}\int e^{\theta(d_\eta^2-1)}G^*(\ddr\theta)<+\infty.
     \end{equation}
         Thus, $\ell(\gb)$ is twice continuously differentiable on $\Delta_\eta$.
Denoting by $k_y(\vartheta_i)=k(y\mid\vartheta_i)$ for $i=1,\dots,\vartheta_{d_\eta}$, we can write
$$
\nabla_{\gb} \textsc{KL}(p_{G^*}||p_{G})=-\sum_{y\in\mathbb N_0}\frac{k_y-k(y\mid\vartheta_{d_\eta})}{p_{G}(y)}p_{G^*}(y).
$$
The $(i,j)$ element of the Hessian matrix $\Mb$ of $\ell(\gb)$ is
\begin{align*}
    \Mb_{i,j}&=\sum_{y\in\mathbb N_0}\frac{(k(y\mid\vartheta_i)-k(y\mid\vartheta_{d_\eta}))(k(y\mid\vartheta_j)-k(y\mid\vartheta_{d_\eta}))}{p_{G}(y)^2}p_{G^*}(y).
\end{align*}
for $i,j=1,\dots,d-1$.

To prove \eqref{eq:minKL}, we apply the method of Lagrange multipliers for the constrained minimum of twice continuously differentiable functions.
The solution $\gb_\eta^*$ satisfies, for same $\lambda_i\geq 0 $ ($i=1,\dots,d_\eta$) with $\lambda_i>0$ for at least one $i$,
\begin{equation}\label{eq:deriv}
    \sum_{y\in\mathbb N_0}\frac{k(y\mid\vartheta_i)-k(y\mid \vartheta_{d_\eta})}{p_{G_\eta^*}(y)}p_{G^*}(y)-\lambda_i+\lambda_{d_\eta}=0\quad (i=1,\dots,d_\eta-1),
\end{equation}
and
\begin{equation}\label{eq:constlambda}
\lambda_1 g_\eta^*(\vartheta_1)=0,\;\dots, \;\lambda_{d_\eta-1}g_\eta^*(\vartheta_{d_\eta-1})=0,\;\lambda_{d_\eta}(\sum_{i=1}^{d_\eta-1}g_\eta^*(\vartheta_i)-1)=0.
\end{equation}
Multiplying \eqref{eq:deriv} by $g_\eta^*(\vartheta_i)$ and taking the sum with respect to $i=1,\dots,d_\eta-1$, we obtain by \eqref{eq:constlambda}
$$
\lambda_{d_\eta}=-\sum_{i=1}^{d_\eta-1}\sum_{y\in\mathbb N_0}g_\eta^*(\vartheta_i)\frac{k(y\mid\vartheta_i)}{p_{G_\eta^*}(y)}p_{G^*}(y)+(1-g_\eta^*(\vartheta_{d_\eta}))\sum_{y\in\mathbb N_0}\frac{k(y\mid\vartheta_{d_\eta})}{p_{G_\eta^*}(y)}p_{G^*}(y).
$$
Thus
$$
\lambda_{d_\eta}-\sum_{y\in\mathbb N_0}\frac{k(y\mid\vartheta_{d_\eta})}{p_{G_\eta^*}(y)}p_{G^*}(y)=-\sum_{i=1}^{d_\eta}\sum_{y\in\mathbb N_0}g_\eta^*(\vartheta_i)\frac{k(y\mid\vartheta_i)}{p_{G_\eta^*}(y)}p_{G^*}(y)=-1.
$$
Substituting in \eqref{eq:deriv}, we obtain
$$
\sum_{y\in\mathbb N_0}\frac{k(y\mid\vartheta_i)}{p_{G_\eta^*}(y)}p_{G^*}(y)-\lambda_i-1=0\quad (i=1,\dots,d_\eta-1),
$$
Thus, by \eqref{eq:constlambda}, 
\begin{equation}\label{eq:ts}
    g_\eta^*(\vartheta_i)\left(\sum_{y\in\mathbb N_0}\frac{k(y\mid\vartheta_i)}{p_{G_\eta^*}(y)}p_{G^*}(y)-1\right)=0
\end{equation}
for every $i=1,\dots,d_\eta-1$. Since
$$
    \sum_{i=1}^{d_\eta}\sum_{y\in\mathbb N_0}g_\eta^*(\vartheta_i)\sum_{y\in\mathbb N_0}\frac{k(y\mid\vartheta_i)}{p_{G_\eta^*}(y)}p_{G^*}(y)=1,
$$
then \eqref{eq:ts} holds also for $i=d_\eta$.
\\
$\square$
\end{proof}

\begin{lem}\label{lem:stochapp3}
 Under the assumptions of Theorem \ref{th:regret} with $\gamma<1$, for every $\delta<2-1/\gamma$
    $$
    ||\gb_n-\gb_\eta^*||^2=o(n^{-\delta})\quad \PP_{G^*}\mbox{-a.s.}
    $$
\end{lem}
\begin{proof}
The random vectors $\gb_n$ and $\gb_\eta^*$ take values in the compact set $\{\vb\in\mathbb R^{d_\eta}:||\vb||\leq 1\}$. By Corollary 4.7 in \cite{MarTok(09)} $\gb_n$ converges to $\gb_\eta^*$, $\PP_{G^*}$-a.s.
Moreover, $(\gb_n)$ satisfies the 
 stochastic approximation in $\mathbb R^{d_\eta}$:
\begin{align*}
\label{eq:newton_proof}
\gb_{n+1}&=\gb_{n}+\alpha_{n+1} \gb_{n}\circ \left(\frac{\kb_{Y_{n+1}}}{p_{\gb_{n}}(Y_{n+1})}-\bf 1\right)\\
&=\gb_{n}+\alpha_{n+1} \hb(\gb_{n})+\alpha_{n+1}\epsilonb_{n+1},
\end{align*}
with initial value $\gb_0$, where
$$
\hb(\gb)=\gb\circ \left(
\sum_{y\in \mathbb N_0}\frac{\kb_y}{p_{\gb}(y)}
p_{G^*}(y)-\bf 1
\right),
$$
and
$$
\epsilonb_{n+1}=\gb\circ \left(
\frac{\kb_{Y_{n+1}}}{p_{\gb}(Y_{n+1})}
-
\sum_{y\in \mathbb N_0}\frac{\kb_y}{p_{\gb}(y)}
p_{G^*}(y)
\right).
$$
   The thesis follows from Theorem 3.1.1 in \cite{Chen(05)} if we can prove that the following conditions hold:\begin{itemize}
        \item[T1] $\alpha_n\rightarrow 0$, $\sum_{n=1}^\infty \alpha_n=\infty$, $\alpha_{n+1}^{-1}-\alpha_n^{-1}\rightarrow 0$;
        \item[T2] $\sum_{n=1}^\infty \alpha_n^{1-\delta/2}\epsilon_n$ converges, $\PP_{G^*}$-a.s.;
        \item[T3] $\hb$ is measurable and locally bounded, and is differentiable at $\gb_\eta^*$. Let $\Hb$ be a matrix such that, as $\gb\rightarrow \gb_\eta^*$,
        $$
        \hb(\gb)=\Hb(\gb-\gb_\eta^*)+\rb(g),\quad \rb(g_\eta^*)=\mathbf 0,\quad ||\rb(\gb)||=o(||\gb-\gb_\eta^*||).
        $$
        All the eigenvalues of the matrix $\Hb$ have negative real parts.
    \end{itemize}
\begin{remark}
        Theorem 3.1.1 in \cite{Chen(05)} also requests  the following condition: T4  
  {\em    There exists a continuously differentiable function $\ell:\Delta_\eta\rightarrow \mathbb R$ such that $\sup_{\delta_1\leq {\rm d}(\gb,\gb_\eta^*)\leq \delta_2}\hb(\gb)^T\nabla_{\gb}\ell<0$
    for every $\delta_2>\delta_1>0$.}   
   However, a close inspection of the proof of Theorem 3.1.1 in Chen  shows that T4 is employed only to ensure convergence of $\gb_n$ towards $\gb_\eta^*$. In our case, this property is known to be true.
   \end{remark}
    T1 is obvious. To prove T2, notice that $(\sum_{k=1}^n\alpha_k^{1-\delta/2}\epsilon_k)_{n\geq 1}$ is a martingale. Moreover, there exists a constant $C$ such that, for every $i=1,\dots,{d_\eta}$,
    \begin{align*}
        \sup_nE_{G^*}((\sum_{k=1}^n\alpha_k^{1-\delta/2}\epsilonb_{k,i})^2)
        &= \sup_n\sum_{k=1}^n\alpha_k^{2-\delta}E_{G^*}(\epsilonb_{k,i})^2)\\
        &\leq C \sum_{k=1}^\infty\left(\frac{1}{\alpha+k}\right)^{\gamma(2-\delta)}<\infty,
        \end{align*}
        since $\gamma(2-\delta)>1$ and $(\epsilonb_{n,i})$ are uniformly bounded.
        We now prove T3. The function $\hb$ is continuously differentiable at $\gb_\eta^*$ and
        $$
        \Hb_{i,j}=-\gb_{\eta,i}^*\sum_{y\in\mathbb N_0}\frac{k_y(\vartheta_i)k_y(\vartheta_j)}{p_{g_\eta^*}(y)^2}p_{G^*}(y).
        $$
        The matrix $\Hb$ can be seen as the product of the diagonal matrix ${\rm diag}(\gb_\eta^*)$, which is positive definite, times $-\Mb$, with $\Mb_{i,j}=\sum_{y\in\mathbb N_0}\frac{k_y(\vartheta_i)k_y(\vartheta_j)}{p_{G_\eta^*}(y)^2}p_{G^*}(y)$. We now prove that $\Mb$ is positive definite. Given a vector $\ub\in \mathbb R^{d_\eta}$ with $||\ub||=1$, we can write that
        \begin{align*}
        \ub^T \Mb\ub=&\sum_{i,j=1}^{d_\eta}\ub_i\ub_j\sum_{y\in\mathbb N_0}\frac{k_y(\vartheta_i)k_y(\vartheta_j)}{p_{G_\eta^*}(y)^2}p_{G^*}(y)=
        \sum_{y\in\mathbb N_0}\frac{p_{G^*{(y)}}}{p_{G_\eta^*}(y)^2}\left(\sum_{i=1}^{d_\eta} \ub_i k_y(\vartheta_i)\right)^2\geq 0,
        \end{align*}
        with equality if and only if $\sum_{i=1}^{d_\eta} \ub_i k_y(\vartheta_i)=0$ for every $y\in\mathbb N_0$. The only solution to the last equation in $\ub$ is the zero vector, which is inconsistent with $||\ub||=1$. Since the function $\ub^T\Mb\ub$ is continuous on the compact set $||\ub||=1$, then the minimum is attained, and by the above reasoning it is strictly positive. It follows that $\Mb$ is positive definite.
    \end{proof}

    \begin{lem}\label{lem:stochapp1}
  Under the assumptions of Theorem \ref{th:regret}, 
       there exists a constant $C=C(d, \vartheta_1,\vartheta_d)$ such that for every $n\geq 1$
  $$
  \textsc{Regret}(G_n,G_\eta^*)\leq C ||\gb_n-\gb_\eta^*||^2.
  $$
 \end{lem}
\begin{proof}
    We can write that
    \begin{align*}
       &\sum_{y\in \mathbb N_0} (\hat\theta_{g_n}(y)-\hat\theta_{g_\eta^*}(y))^2p_{g_\eta^*}(y)\\
         &\leq \sum_{y\in \mathbb N_0} 
         \frac{
         \left(
         \sum_{i=1}^{d_\eta}\vartheta_i^{y+1}e^{-\vartheta_i}\left(g_n(\vartheta_i)-g_\eta^*(\vartheta_i)\right)
         \right)^2
         }{
         (\sum_{i=1}^{d_\eta} \vartheta_i^{y}e^{-\vartheta_i}g_\eta^*(\vartheta_i))^2
         }p_{G^*}(y)\\
         &\quad+\sum_{y\in \mathbb N_0}  
         \frac{
         \left(
         \sum_{i=1}^{d_\eta}\vartheta_i^{y+1}e^{-\vartheta_i}g_n(\vartheta_i)
         \right)^2\left(\sum_{i=1}^{d_\eta} \vartheta_i^{y}e^{-\vartheta_i}(g_n(\vartheta_i)-g_\eta^*(\vartheta_i)\right)^2
         }{
         \left(\sum_{i=1}^{d_\eta} \vartheta_i^{y}e^{-\vartheta_i}g_n(\vartheta_i)\right)^2\left(\sum_{i=1}^{d_\eta} \vartheta_i^{y}e^{-\vartheta_i} g_\eta^*(\vartheta_i)\right)^2
         }p_{G^*}(y)
         \\
         &\leq 4{d_\eta}\vartheta_{d_\eta}^2e^{2(\vartheta_{d_\eta}-\vartheta_1)}||\gb_n-\gb_\eta^*||^2
         \int_0^\infty e^{-\theta} \sum_{y\in \mathbb N_0} \frac{1}{y!}\left(\frac{\theta\vartheta_{d_\eta}^2}{\vartheta_1^2} \right)^yG^*(d\theta)
         \\
         &\leq C ||\gb_n-\gb_\eta^*||^2,
    \end{align*}
    with $C=4{d_\eta}\vartheta_{d_\eta}^2e^{2(\vartheta_{d_\eta}-\vartheta_1)}         \int_0^\infty e^{\theta({d_\eta}^2-1)} G^*(d\theta)<+\infty,
         $
         by \eqref{eq:condfin}.
                 Thus,
         $$
         \textsc{Regret}(G_n,G_\eta^*)\leq C ||\gb_n-\gb_\eta^*||^2,
         $$
         with $C<+\infty$.
\end{proof}
    \textit{Proof of Theorem \ref{th:regret}.} The proof is an immediate consequence of lemmas \ref{lem:stochapp3} and \ref{lem:stochapp1}.


\section{Numerical illustrations: synthetic data}\label{app5_num1}
\subsection{Preliminaries}\label{app5_num11}
We generate synthetic data from a Poisson mixture model, for various choices of the prior (mixing) distribution $G$. In particular, let $n\geq1$ and let $(Y_{1},\theta_{1}),\ldots,(Y_{n},\theta_{n})$ be random vectors distributed as follows:
\begin{equation}\label{eq:mixture_model_poisson}
\begin{aligned}
Y_i \mid \theta_i &\quad \simind \quad \text{Poisson}(\cdot \mid \theta_i) \qquad (i=1,\ldots,n),\\[-0.2cm]
\theta_i &\quad \simiid \quad G.
\end{aligned}
\end{equation}
First, we assume $G$ to be a Uniform distribution on $[a,b]$; precisely, we set $a=0$ and $b=3$. Then, we consider two examples of $G$ belonging to the class $\mathcal{G}$ of sub-exponential distribution of order $s$, which is defined as
\begin{displaymath}
\mathcal{G}=\left\{G\text{ on }\mathbb{R}^{+}:\;G([t,\infty)) \leq 2e^{-t/s}\ \text{for all } t>0\right\},\qquad s>0.
\end{displaymath}
In particular, we assume $G$ to be: i) a Weibull distribution with scale $a$ and shape $b$, which belongs to $\mathcal{G}$ for $b\geq1$; ii) a half-Gaussian distribution, namely the distribution of the positive part of a Gaussian random variable with mean $0$ and variance $\sigma^{2}$, which belongs $\mathcal{G}$ for $\sigma>0$. Precisely we set $a=5$ and $b=3$ for the Weibull, and $\sigma=1$ for the half-Gaussian; under this parameterization the Weibull tail is lighter than the half-Gaussian tail. Finally, we consider an example of $G$ belonging to the moment class $\mathcal{M}$ defined, for any real $M_{p}$, as
\begin{displaymath}
\mathcal{M}=\left\{G\text{ on }\mathbb{R}^{+}:\;\int_{\mathbb{R}^{+}}\theta^{p} G(\ddr\theta)<M_{p}\right\},\qquad p>0;
\end{displaymath}
see \citet[Section 1]{She(24)}. In particular, we assume $G$ to be square-root of half-Cauchy distribution, namely the distribution of the square-root of the positive part of a standard Cauchy random variable. This distribution has heavier tail than both the Weibull distribution and the half-Gaussian distribution.

For each of the above choices of the prior $G$, we assess the QB estimate $\hat{\theta}_{G_{n}}$ in \eqref{seq_estim} by comparing it, for fixed $n\geq 1$, with the Oracle, Robbins, ML, and MHD estimates. Consider the posterior mean, or Bayes estimate,
\begin{equation}\label{rob_for}
\hat{\theta}_{G}(y)=E_{G}[\theta_i \mid Y_i=y]=\frac{\int_{\Theta}\theta\,\text{Poisson}(y\mid\theta)\,G(\ddr\theta)}{\int_{\Theta}\text{Poisson}(y\mid\theta)\,G(\ddr\theta)}=(y+1)\frac{p_G(y+1)}{p_G(y)},
\qquad y\in\mathbb{N}_0.
\end{equation}
The Oracle estimate is obtained by replacing $G$ in \eqref{rob_for} with the ``true'', or oracle, prior $G^\ast$, namely the Uniform, Weibull, half-Gaussian, or square-root of half-Cauchy priors that generate the $\theta_{i}$'s in \eqref{eq:mixture_model_poisson}, and evaluating the marginal likelihood $p_{G^\ast}$ numerically via the trapezoidal rule. The Robbins estimate is obtained by replacing $p_G$ in \eqref{rob_for} with the empirical distribution $\hat p_n(y)=n^{-1}\sum_{1\leq i\leq n}I(Y_i=y)$, for $y\in\mathbb{N}_{0}$. The ML and MHD estimates are obtained by replacing $G$ in \eqref{rob_for} with its nonparametric maximum likelihood estimate and its minimum Hellinger distance estimate, respectively, both computed via the vertex direction algorithm of \citet[Algorithm 1]{Jan(24)}, which is based on a discretization of $G$ over a finite grid.

As a measure of estimation accuracy, we consider the empirical mean squared error (\textsc{e-MSE}). For $n\in\mathbb{N}$, let $(\theta_1,\ldots,\theta_n)$ be the values generated from $G$ in \eqref{eq:mixture_model_poisson}, and let $(\hat\theta_{G}(y_1),\ldots,\hat\theta_{G}(y_n))$ be the corresponding estimates, namely the Oracle, Robbins, ML, MHD, of QB estimates. The \textsc{e-MSE} is defined as
\begin{displaymath}
\textsc{e-MSE}=\frac{1}{n}\sum_{i=1}^n \bigl(\hat\theta_{G}(y_i)-\theta_i\bigr)^2.
\end{displaymath}
When $\hat\theta_{G}$ is the oracle estimate $\hat\theta_{G^\ast}$, the \textsc{e-MSE} is referred to as the empirical minimum mean squared error (\textsc{e-MMSE}), i.e.
\begin{displaymath}
\textsc{e-MMSE}=\frac{1}{n}\sum_{i=1}^n \bigl(\hat\theta_{G^\ast}(y_i)-\theta_i\bigr)^2.
\end{displaymath}
Finally, for the Robbins, ML, MHD, and QB estimates, we define the empirical regret (\textsc{e-REGRET}) as $\textsc{e-Regret}=\textsc{e-MSE}-\textsc{e-MMSE}$, namely the excess empirical squared error with respect to the oracle benchmark.

For each of the above choices of the prior $G$, we also provide the QB credible interval $I_{G_{n}}$ in \eqref{qbeb_interval}, which, for a fixed level $1-\alpha$, is constructed by considering a uniform grid of $50$ pairs of tuning parameters $(\beta_{1},\beta_{2})$ satisfying the constraint $\beta_1+\beta_2=\alpha$ and then selecting, by a simple grid search, the pair of parameters that yields the shortest credible interval. Specifically, we consider an equally spaced grid $\mathcal B_\alpha=\{(\beta_{1,\ell},\beta_{2,\ell}):\ell=1,\ldots,50\}$, with $\beta_{1,\ell}\in(0,\alpha)$ and $\beta_{2,\ell}=\alpha-\beta_{1,\ell}$. For each pair $(\beta_{1,\ell},\beta_{2,\ell})\in\mathcal B_\alpha$, we compute the lower and upper envelopes $\underline F^{\beta_{1,\ell}}_{G_n}(\cdot\mid y)$ and $\overline F^{\beta_{1,\ell}}_{G_n}(\cdot\mid y)$ defined in \eqref{conf_ext1}--\eqref{conf_ext2}. We then obtain the QB credible interval
\begin{displaymath}
I^{(\ell)}_{G_n}(y)=[L^{(\ell)}_n(y),U^{(\ell)}_n(y)]\cap\Theta,
\end{displaymath}
where
\begin{displaymath}
L^{(\ell)}_n(y)=\sup\left\{\theta\in\Theta:\overline F^{\beta_{1,\ell}}_{G_n}(\theta\mid y)\leq\frac{\beta_{2,\ell}}{2}\right\}
\end{displaymath}
and
\begin{displaymath}
U^{(\ell)}_n(y)=\inf\left\{\theta\in\Theta:\underline F^{\beta_{1,\ell}}_{G_n}(\theta\mid y)\geq1-\frac{\beta_{2,\ell}}{2}\right\}.
\end{displaymath}
Now, the selected pair of tuning parameters is the one that minimizes the length of the credible interval, namely
\begin{displaymath}
\ell^\ast\in\arg\min_{1\leq \ell\leq 50}\left\{U^{(\ell)}_n(y)-L^{(\ell)}_n(y)\right\},
\end{displaymath}
and the reported QB credible interval is $I_{G_n}(y)=I^{(\ell^\ast)}_{G_n}(y)$. This interval is compared with the corresponding Oracle credible interval of level $1-\alpha$, which is the $(1-\alpha)$-level posterior quantile obtained under the ``true'', or oracle, prior $G^\ast$, namely the Uniform, Weibull, half-Gaussian, or square-root of half-Cauchy priors that generate the $\theta_{i}$'s in \eqref{eq:mixture_model_poisson}. In particular, in order to define the Oracle credible interval, for $y\in\mathbb{N}_0$ let
\begin{displaymath}
F_{G^{\ast}}(\theta\mid y)=\frac{\int_0^\theta \text{Poisson}(y\mid\vartheta)\,G^{\ast}(\ddr\vartheta)}{\int_{\Theta}\text{Poisson}(y\mid\vartheta)\,G^{\ast}(\ddr\vartheta)},\qquad \theta\in\Theta
\end{displaymath}
be the posterior distribution function of $\theta$ given $Y=y$ under the prior $G^{\ast}$. The Oracle credible interval is then defined as
\begin{displaymath}
I_{G^{\ast}}(y)=[L_{\ast}(y),U_{\ast}(y)],
\end{displaymath}
where
\begin{displaymath}
L_{\ast}(y)=\inf\{\theta\in\Theta\text{ : }F_{G^{\ast}}(\theta\mid y)\geq \alpha/2\},
\qquad
U_{\ast}(y)=\inf\{\theta\in\Theta\text{ : }F_{G^{\ast}}(\theta\mid y)\geq 1-\alpha/2\},
\end{displaymath}
with the posterior distribution function $F_{G^{\ast}}$ being evaluated numerically via the trapezoidal rule, as for the Oracle estimate.

As a measure of comparison between the QB credible interval $I_{G_n}$ and the corresponding Oracle credible interval $I_{G^\ast}$, we consider the empirical mean absolute relative length discrepancy (\textsc{e-MARLD}). For $y\in\mathbb N_0$, let
\begin{displaymath}
\ell_n(y)=U_{n}(y)-L_{n}(y),\qquad
\ell_{\ast}(y)=U_{\ast}(y)-L_{\ast}(y)
\end{displaymath}
denote the lengths of the QB and Oracle credible intervals, respectively. We define the absolute relative length discrepancy
\begin{displaymath}
D_n(y)=\left|\frac{\ell_n(y)}{\ell_{\ast}(y)}-1\right|
=\frac{|\ell_n(y)-\ell_{\ast}(y)|}{\ell_{\ast}(y)},
\end{displaymath}
and the \textsc{e-MARLD} by
\begin{displaymath}
\textsc{e-MARLD}=\frac{1}{n}\sum_{i=1}^n D_n(y_i).
\end{displaymath}
The closer $\textsc{e-MARLD}$ is to $0$, the closer, on average, the lengths of the QB credible intervals are to those of the Oracle credible intervals. In particular, a value $\textsc{e-MARLD}=x\in[0,1]$ means that, on average, the length of the QB credible interval differs from that of the corresponding Oracle credible interval by a proportion $x$.

\subsection{Uniform prior}\label{app5_num12}

For $n\in\{50,\,100,\,200,\,400\}$, we generate data $\mathbf{Y}_{n}=(Y_{1},\ldots,Y_{n})$ from the Poisson mixture model \eqref{eq:mixture_model_poisson} with a Uniform prior $G$ on the set $[0,3]$. For the QB estimate, the marginal likelihood in Newton's algorithm \eqref{eq:newton} is evaluated numerically via the trapezoidal rule. To perform this evaluation, the density function of the mixing distribution $G_{n}$ is represented through its values on a fixed uniform grid of $d\in\{5,000;\,1,000;\,500;\,100;\,50;\,10\}$ quadrature points over $\Theta=(0,U_{\Theta})$, where $U_\Theta=\max\{\max\{\mathbf{Y}_n\},\lceil Q_{n,0.99}+4(\max\{Q_{n,0.99},1\})^{1/2}\rceil\}$, with $Q_{n,0.99}=\textsc{Quantile}(\mathbf{Y}_n;0.99)$. Further, we set $G_{0}$ to be Uniform over $\Theta$, and take the learning rate to be $\alpha_{n}=(1+n)^{-0.99}$. For this setting of Newton's algorithm, Table \ref{uniform_tab_sens} provides \textsc{e-MSE}, \textsc{e-REGRET} and CPU time as the sample size $n$ and the grid resolution $d$ vary. The CPU time refers to the time (in seconds) for processing a new observation on a laptop MacBook Pro (M1 processor).

\begin{table}[ht]
\centering
\caption{Uniform prior: \textsc{e-MSE}, \textsc{e-REGRET} and CPU time (in seconds) of QB as $n$ and $d$ vary.}
{
\setlength{\tabcolsep}{0pt}
\begin{tabular}{@{}l@{\hspace{1cm}}*{6}{>{\centering\arraybackslash}p{1.75cm}}@{}}
\hline
\hline
 & $d=5{,}000$ & $d=1{,}000$ & $d=500$ & $d=100$ & $d=50$ & $d=10$ \\[0.1cm]
\hline
\multicolumn{7}{@{}l}{\underline{$n=50$}} \\[0.05cm]
\textsc{e-MSE}    & 0.515 & 0.515 & 0.515 & 0.517 & 0.522 & 0.719 \\
\textsc{e-REGRET} & 0.045 & 0.045 & 0.045 & 0.047 & 0.052 & 0.249 \\
CPU time & 0.003 & 0.002 & 0.002 & 0.002 & 0.002 & 0.001 \\[0.2cm]

\multicolumn{7}{@{}l}{\underline{$n=100$}} \\[0.05cm]
\textsc{e-MSE}    & 1.045 & 1.045 & 1.045 & 1.045 & 1.047 & 1.637 \\
\textsc{e-REGRET} & 0.539 & 0.539 & 0.539 & 0.539 & 0.541 & 1.131 \\
CPU time & 0.003 & 0.002 & 0.002 & 0.002 & 0.002 & 0.001 \\[0.2cm]

\multicolumn{7}{@{}l}{\underline{$n=200$}} \\[0.05cm]
\textsc{e-MSE}    & 0.597 & 0.596 & 0.595 & 0.591 & 0.589 & 0.841 \\
\textsc{e-REGRET} & 0.113 & 0.113 & 0.112 & 0.108 & 0.106 & 0.357 \\
CPU time & 0.003 & 0.002 & 0.002 & 0.002 & 0.002 & 0.002 \\[0.2cm]

\multicolumn{7}{@{}l}{\underline{$n=400$}} \\[0.05cm]
\textsc{e-MSE}    & 0.646 & 0.646 & 0.646 & 0.646 & 0.647 & 0.961 \\
\textsc{e-REGRET} & 0.159 & 0.159 & 0.159 & 0.159 & 0.160 & 0.474 \\
CPU time & 0.003 & 0.002 & 0.002 & 0.002 & 0.002 & 0.001 \\[0.1cm]
\hline
\hline
\end{tabular}
}
\label{uniform_tab_sens}
\end{table}

Table \ref{uniform_compare_tab} compares the \textsc{e-MSE} and \textsc{e-REGRET} for the Oracle, Robbins, ML, MHD and QB estimates. For the QB estimate we consider: i) a fixed uniform grid of $d=1,000$ quadrature points over $\Theta=(0,U_{\Theta})$; ii) an initial guess $G_{0}$ that is Uniform over $\Theta$; iii) a learning rate $\alpha_{n}=(1+n)^{-0.99}$. This setting corresponds to the second column of Table \ref{uniform_tab_sens}. Figure \ref{uniform_fig} displays the Oracle, Robbins, ML, MHD and QB estimates against the true values.

\begin{table}[ht]
\centering
\caption{Uniform prior: \textsc{e-MSE} and \textsc{e-REGRET} of Oracle, Robbins, ML, MHD and QB.}
{
\setlength{\tabcolsep}{0pt}
\begin{tabular}{@{}l@{\hspace{1cm}}*{5}{>{\centering\arraybackslash}p{1.75cm}}@{}}
\hline
\hline
 & Oracle & Robbins & ML & MHD & QB \\[0.1cm]
\hline
\multicolumn{6}{@{}l}{\underline{$n=50$}} \\[0.05cm]
\textsc{e-MSE}    & 0.470 & 1.114 & 0.499 & 0.526 & 0.515 \\
\textsc{e-REGRET} & 0.000 & 0.645 & 0.030 & 0.056 & 0.045 \\[0.3cm]

\multicolumn{6}{@{}l}{\underline{$n=100$}} \\[0.05cm]
\textsc{e-MSE}    & 0.506 & 0.991 & 1.755 & 0.941 & 1.045 \\
\textsc{e-REGRET} & 0.000 & 0.485 & 1.249 & 0.435 & 0.539 \\[0.3cm]

\multicolumn{6}{@{}l}{\underline{$n=200$}} \\[0.05cm]
\textsc{e-MSE}    & 0.483 & 1.374 & 0.661 & 0.666 & 0.596 \\
\textsc{e-REGRET} & 0.000 & 0.891 & 0.178 & 0.183 & 0.113 \\[0.3cm]

\multicolumn{6}{@{}l}{\underline{$n=400$}} \\[0.05cm]
\textsc{e-MSE}    & 0.487 & 1.236 & 0.564 & 0.650 & 0.646 \\
\textsc{e-REGRET} & 0.000 & 0.749 & 0.077 & 0.164 & 0.159 \\[0.1cm]
\hline
\hline
\end{tabular}
}
\label{uniform_compare_tab}
\end{table}

\begin{figure}[t]
\centering
\includegraphics[width=.85\textwidth]{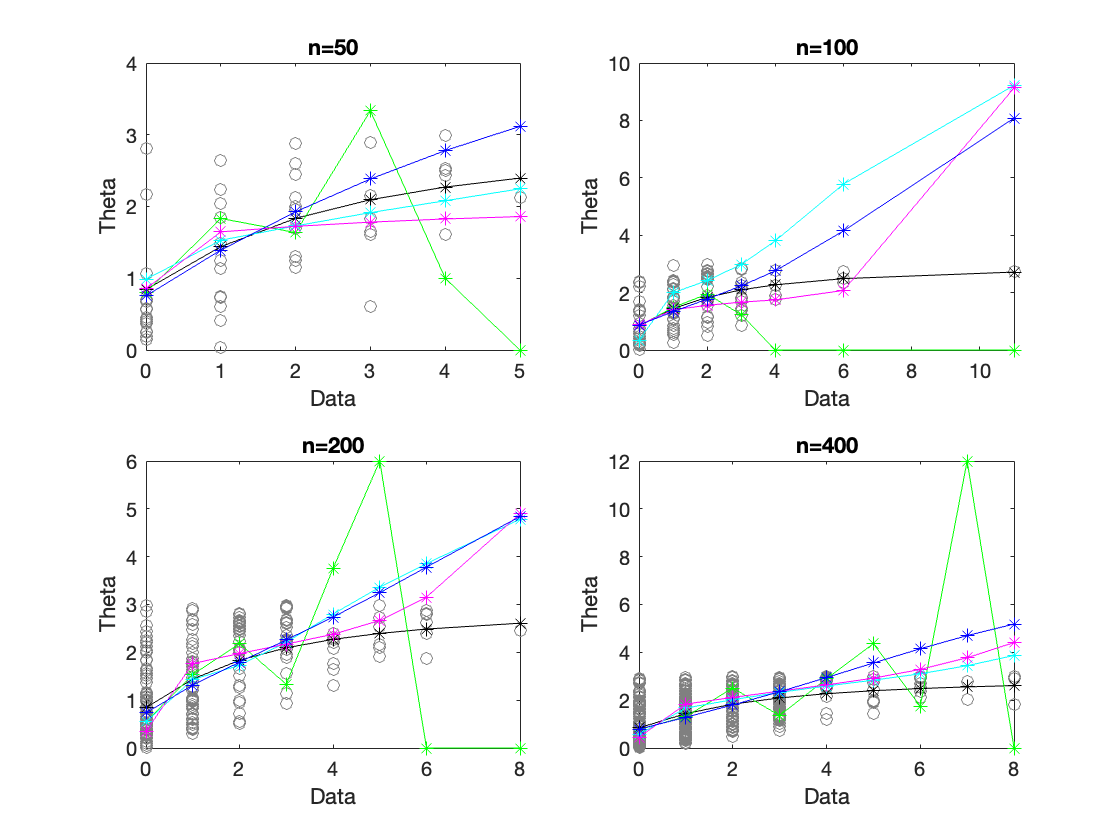}
\caption{Data points $Y$ versus true parameters $\theta$ (grey) under Uniform prior, and estimates: Oracle (black), Robbins (green), ML (cyan), MHD (magenta) and QB (blue).}
\label{uniform_fig}
\end{figure}

To conclude, we provide QB credible intervals at level $1-\alpha=0.95$. Intervals are constructed by relying on Newton's algorithm initialized as in Table \ref{uniform_tab_sens}, with the optimal choice of the tuning parameters $\beta_{1}$ and $\beta_{2}$. For such QB credible intervals, Table \ref{uniform_tab_sens_interval} provides \textsc{e-MARLD} as the sample size $n$ and the grid resolution $d$ vary. 

\begin{table}[ht]
\centering
\caption{Uniform prior: \textsc{e-MARLD} of QB as $n$ and $d$ vary.}
{
\setlength{\tabcolsep}{0pt}
\begin{tabular}{@{}l@{\hspace{1cm}}*{6}{>{\centering\arraybackslash}p{1.75cm}}@{}}
\hline
\hline
 & $d=5{,}000$ & $d=1{,}000$ & $d=500$ & $d=100$ & $d=50$ & $d=10$ \\[0.1cm]
\hline
\multicolumn{7}{@{}l}{\underline{$n=50$}} \\[0.05cm]
\textsc{e-MARLD} & 0.662 & 0.662 & 0.665 & 0.697 & 0.726 & 1.164 \\[0.2cm]

\multicolumn{7}{@{}l}{\underline{$n=100$}} \\[0.05cm]
\textsc{e-MARLD} & 0.708 & 0.708 & 0.712 & 0.784 & 0.932 & 1.469 \\[0.2cm]

\multicolumn{7}{@{}l}{\underline{$n=200$}} \\[0.05cm]
\textsc{e-MARLD} & 0.571 & 0.571 & 0.574 & 0.596 & 0.631 & 1.075 \\[0.2cm]

\multicolumn{7}{@{}l}{\underline{$n=400$}} \\[0.05cm]
\textsc{e-MARLD} & 0.643 & 0.643 & 0.647 & 0.686 & 0.775 & 1.283 \\[0.1cm]
\hline
\hline
\end{tabular}
}
\label{uniform_tab_sens_interval}
\end{table}

Under the setting of Table \ref{uniform_tab_sens_interval}, and assuming $d=1,000$ quadrature points over $\Theta=(0,U_{\Theta})$, Figure \ref{uniform_fig_interval} display the QB and Oracle credible intervals. Figure \ref{uniform_fig_interval1} display the Oracle and QB credible intervals for datasets with larger sample sizes, namely $n\in\{1,000,\,2,000,\,4,000,\,8,000\}$, generated under the same Poisson mixture model.

\begin{figure}[t]
\centering
\includegraphics[width=.85\textwidth]{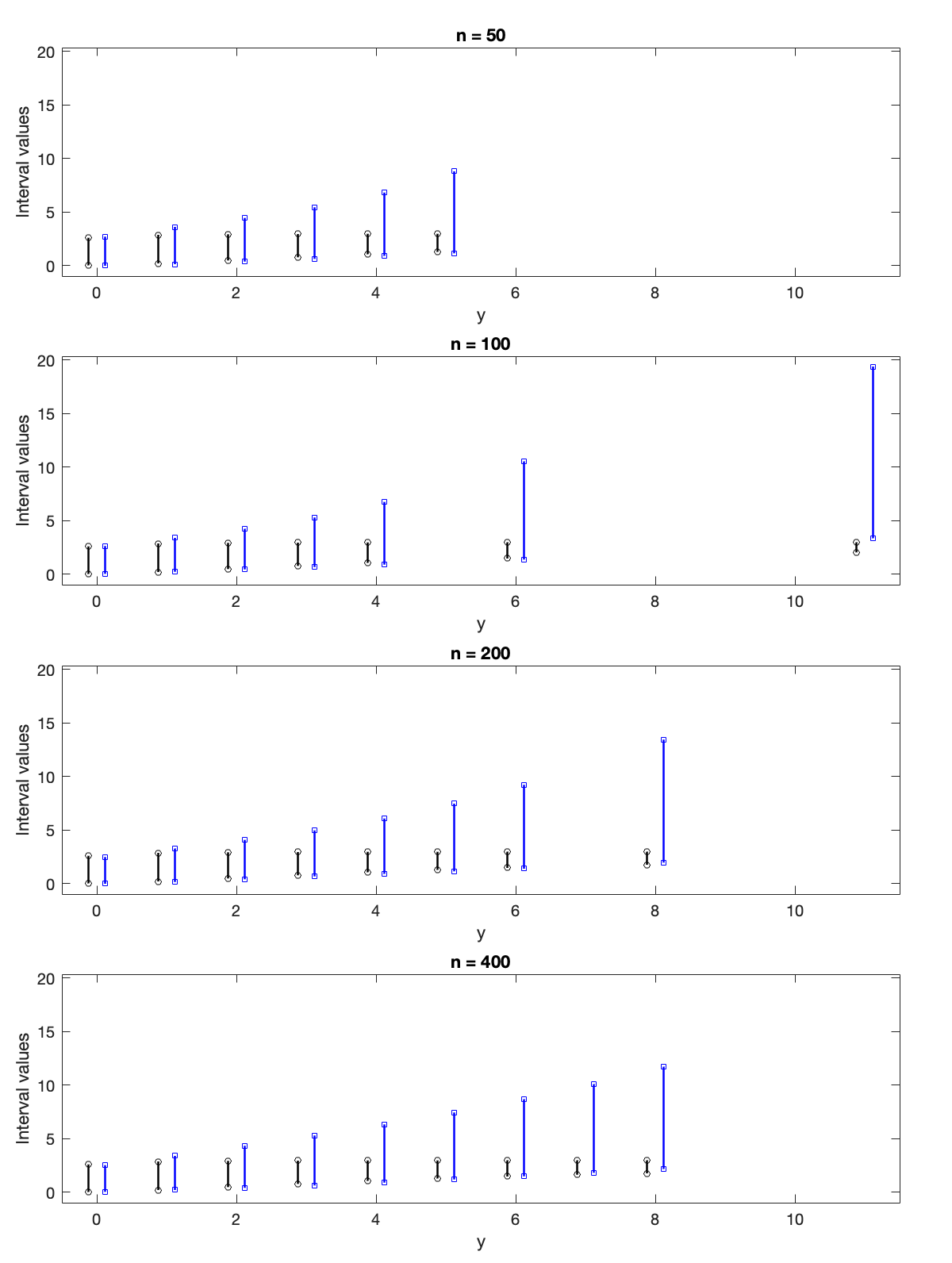}
\caption{Uniform prior: Oracle credible interval $I_{G^{\ast}}(y)$ (black), QB credible interval $I_{G_{n}}(y)$ (blue) for $Y\in\{Y_1,\dots,Y_n\}$.}
\label{uniform_fig_interval}
\end{figure}

\begin{figure}[t]
\centering
\includegraphics[width=.85\textwidth]{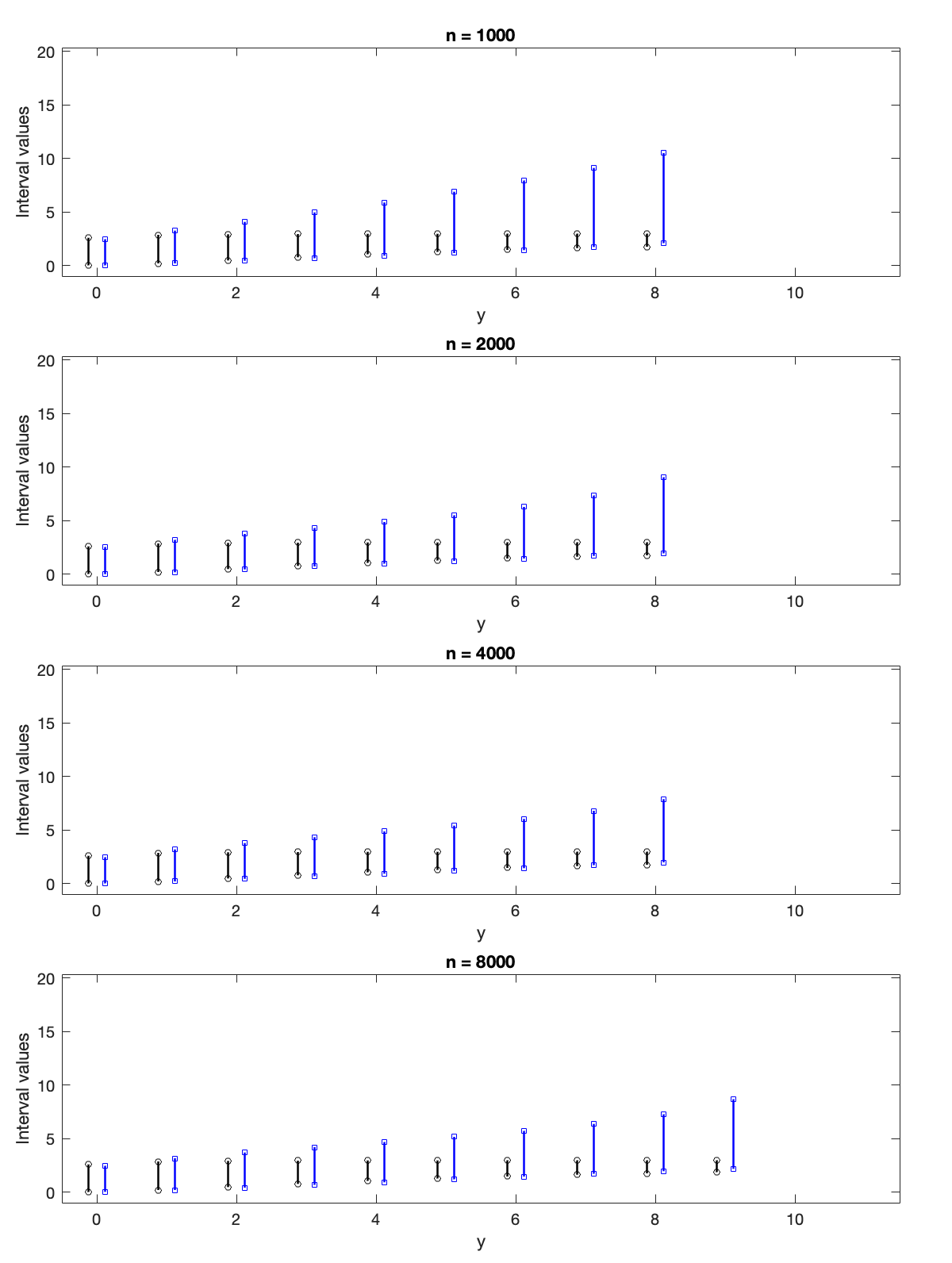}
\caption{Uniform prior: Oracle credible interval $I_{G^{\ast}}(y)$ (black), QB credible interval $I_{G_{n}}(y)$ (blue) for $Y\in\{Y_1,\dots,Y_n\}$.}
\label{uniform_fig_interval1}
\end{figure}

\subsection{Weibull prior}\label{app5_num13}

For $n\in\{50,\,100,\,200,\,400\}$, we generate data $\mathbf{Y}_{n}=(Y_{1},\ldots,Y_{n})$ from the Poisson mixture model \eqref{eq:mixture_model_poisson} with a Weibull prior $G$ with scale $5$ and shape $3$. For the QB estimate, the marginal likelihood in Newton's algorithm \eqref{eq:newton} is evaluated numerically via the trapezoidal rule. To perform this evaluation, the density function of the mixing distribution $G_{n}$ is represented through its values on a fixed uniform grid of $d\in\{5,000;\,1,000;\,500;\,100;\,50;\,10\}$ quadrature points over $\Theta=(0,U_{\Theta})$, where $U_\Theta=\max\{\max\{\mathbf{Y}_n\},\lceil Q_{n,0.99}+4(\max\{Q_{n,0.99},1\})^{1/2}\rceil\}$, with $Q_{n,0.99}=\textsc{Quantile}(\mathbf{Y}_n;0.99)$. Further, we se $G_{0}$ to be Uniform over $\Theta$, and take the learning rate to be $\alpha_{n}=(1+n)^{-0.99}$. For this setting of Newton's algorithm, Table \ref{weib_tab_sens} provides \textsc{e-MSE}, \textsc{e-REGRET} and CPU time as the sample size $n$ and the grid resolution $d$ vary. The CPU time refers to the time (in seconds) for processing a new observation on a laptop MacBook Pro (M1 processor). 

\begin{table}[ht]
\centering
\caption{Weibull prior: \textsc{e-MSE}, \textsc{e-REGRET} and CPU time (in seconds) of QB as $n$ and $d$ vary.}
{
\setlength{\tabcolsep}{0pt}
\begin{tabular}{@{}l@{\hspace{1cm}}*{6}{>{\centering\arraybackslash}p{1.75cm}}@{}}
\hline
\hline
 & $d=5{,}000$ & $d=1{,}000$ & $d=500$ & $d=100$ & $d=50$ & $d=10$ \\[0.1cm]
\hline
\multicolumn{7}{@{}l}{\underline{$n=50$}} \\[0.05cm]
\textsc{e-MSE}    & 2.485 & 2.485 & 2.485 & 2.484 & 2.483 & 2.472 \\
\textsc{e-REGRET} & 0.100 & 0.100 & 0.100 & 0.100 & 0.099 & 0.088 \\
CPU time & 0.003 & 0.002 & 0.002 & 0.002 & 0.001 & 0.001 \\[0.2cm]

\multicolumn{7}{@{}l}{\underline{$n=100$}} \\[0.05cm]
\textsc{e-MSE}    & 1.989 & 1.989 & 1.989 & 1.988 & 1.986 & 1.971 \\
\textsc{e-REGRET} & 0.193 & 0.193 & 0.193 & 0.192 & 0.190 & 0.175 \\
CPU time & 0.003 & 0.002 & 0.002 & 0.002 & 0.001 & 0.001 \\[0.2cm]

\multicolumn{7}{@{}l}{\underline{$n=200$}} \\[0.05cm]
\textsc{e-MSE}    & 2.024 & 2.024 & 2.024 & 2.023 & 2.023 & 2.063 \\
\textsc{e-REGRET} & -0.043 & -0.043 & -0.043 & -0.043 & -0.043 & -0.003 \\
CPU time & 0.003 & 0.002 & 0.002 & 0.002 & 0.000 & 0.000 \\[0.2cm]

\multicolumn{7}{@{}l}{\underline{$n=400$}} \\[0.05cm]
\textsc{e-MSE}    & 1.947 & 1.947 & 1.947 & 1.947 & 1.946 & 1.963 \\
\textsc{e-REGRET} & 0.362 & 0.362 & 0.362 & 0.361 & 0.361 & 0.378 \\
CPU time & 0.003 & 0.002 & 0.002 & 0.002 & 0.000 & 0.000 \\[0.1cm]
\hline
\hline
\end{tabular}
}
\label{weib_tab_sens}
\end{table}

Table \ref{weib_compare_tab} compares the \textsc{e-MSE} and \textsc{e-REGRET} for the Oracle, Robbins, ML, MHD and QB estimates. For the QB estimate we consider: i) a fixed uniform grid of $d=1,000$ quadrature points over $\Theta=(0,U_{\Theta})$; ii) an initial guess $G_{0}$ that is Uniform over $\Theta$; iii) a learning rate $\alpha_{n}=(1+n)^{-0.99}$. This setting corresponds to the second column of Table \ref{weib_tab_sens}. Figure \ref{weib_fig} displays the Oracle, Robbins, ML, MHD and QB estimates against the true values.

\begin{table}[ht]
\centering
\caption{Weibull prior: \textsc{e-MSE} and \textsc{e-REGRET} of Oracle, Robbins, ML, MHD and QB.}
{
\setlength{\tabcolsep}{0pt}
\begin{tabular}{@{}l@{\hspace{1cm}}*{5}{>{\centering\arraybackslash}p{1.75cm}}@{}}
\hline
\hline
 & Oracle & Robbins & ML & MHD & QB \\[0.1cm]
\hline
\multicolumn{6}{@{}l}{\underline{$n=50$}} \\[0.05cm]
\textsc{e-MSE}    & 2.385 & 3.921 & 2.204 & 2.228 & 2.485 \\
\textsc{e-REGRET} & 0.000 & 1.537 & -0.180 & -0.157 & 0.100 \\[0.3cm]

\multicolumn{6}{@{}l}{\underline{$n=100$}} \\[0.05cm]
\textsc{e-MSE}    & 1.796 & 11.319 & 1.848 & 1.977 & 1.989 \\
\textsc{e-REGRET} & 0.000 & 9.523 & 0.052 & 0.181 & 0.193 \\[0.3cm]

\multicolumn{6}{@{}l}{\underline{$n=200$}} \\[0.05cm]
\textsc{e-MSE}    & 2.066 & 7.600 & 2.051 & 2.084 & 2.024 \\
\textsc{e-REGRET} & 0.000 & 5.533 & -0.015 & 0.017 & -0.043 \\[0.3cm]

\multicolumn{6}{@{}l}{\underline{$n=400$}} \\[0.05cm]
\textsc{e-MSE}    & 1.585 & 6.072 & 2.910 & 1.678 & 1.947 \\
\textsc{e-REGRET} & 0.000 & 4.487 & 1.324 & 0.093 & 0.362 \\[0.1cm]
\hline
\hline
\end{tabular}
}
\label{weib_compare_tab}
\end{table}

\begin{figure}[t]
\centering
\includegraphics[width=.85\textwidth]{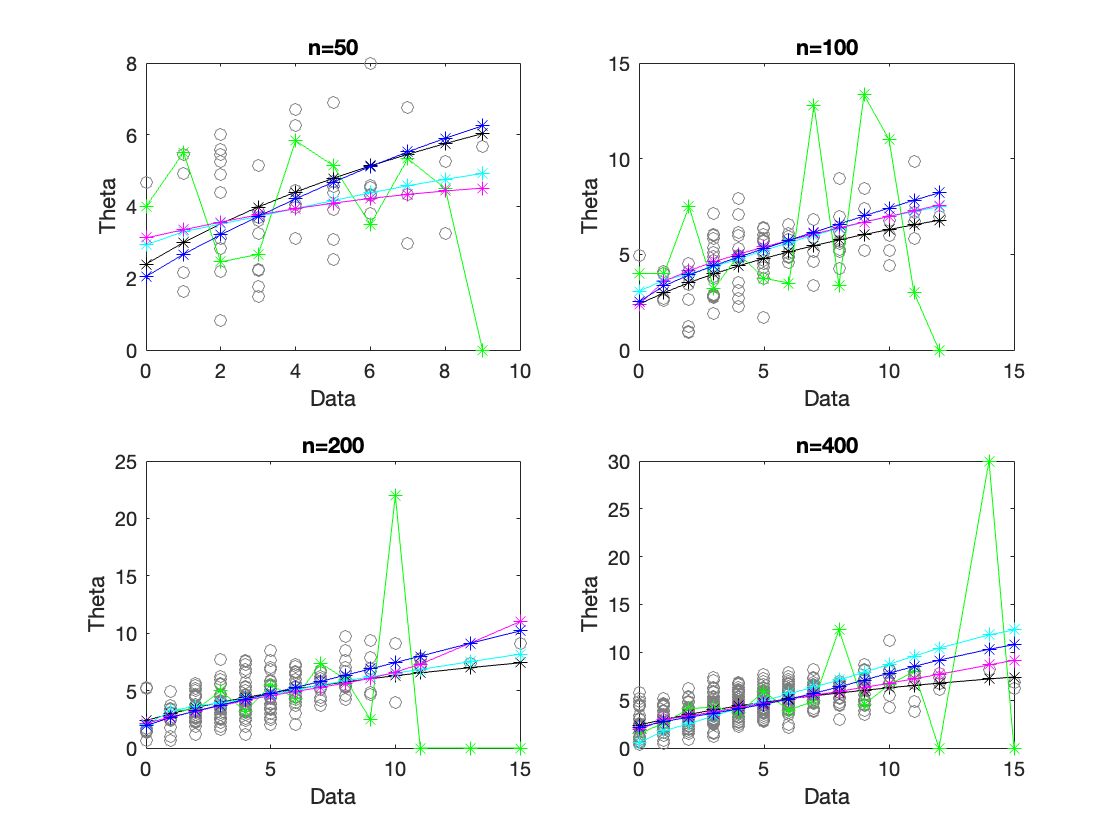}
\caption{Data points $Y$ versus true parameters $\theta$ (grey) under Weibull prior, and estimates: Oracle (black), Robbins (green), ML (cyan), MHD (magenta) and QB (blue).}
\label{weib_fig}
\end{figure}

To conclude, we provide QB credibles interval at level $1-\alpha=0.95$. Intervals are constructed by relying on Newton's algorithm initialized as in Table \ref{weib_tab_sens},  with the optimal choice of the tuning parameters $\beta_{1}$ and $\beta_{2}$. For such QB credible intervals, Table \ref{weib_tab_sens_interval} provides \textsc{e-MARLD} as the sample size $n$ and the grid resolution $d$ vary. 

\begin{table}[ht]
\centering
\caption{Weibull prior: \textsc{e-MARLD} of QB as $n$ and $d$ vary.}
{
\setlength{\tabcolsep}{0pt}
\begin{tabular}{@{}l@{\hspace{1cm}}*{6}{>{\centering\arraybackslash}p{1.75cm}}@{}}
\hline
\hline
 & $d=5{,}000$ & $d=1{,}000$ & $d=500$ & $d=100$ & $d=50$ & $d=10$ \\[0.1cm]
\hline
\multicolumn{7}{@{}l}{\underline{$n=50$}} \\[0.05cm]
\textsc{e-MARLD} & 0.228 & 0.228 & 0.232 & 0.273 & 0.316 & 0.792 \\[0.2cm]

\multicolumn{7}{@{}l}{\underline{$n=100$}} \\[0.05cm]
\textsc{e-MARLD} & 0.325 & 0.326 & 0.333 & 0.369 & 0.421 & 1.078 \\[0.2cm]

\multicolumn{7}{@{}l}{\underline{$n=200$}} \\[0.05cm]
\textsc{e-MARLD} & 0.294 & 0.294 & 0.300 & 0.342 & 0.403 & 0.956 \\[0.2cm]

\multicolumn{7}{@{}l}{\underline{$n=400$}} \\[0.05cm]
\textsc{e-MARLD} & 0.315 & 0.315 & 0.316 & 0.360 & 0.410 & 0.992 \\[0.1cm]
\hline
\hline
\end{tabular}
}
\label{weib_tab_sens_interval}
\end{table}

Under the setting of Table \ref{weib_tab_sens_interval}, and assuming $d=1,000$ quadrature points over $\Theta=(0,U_{\Theta})$, Figure \ref{weib_fig_interval} display the QB and Oracle credible intervals. Figure \ref{weib_fig_interval1} display the Oracle and QB credible intervals for datasets with larger sample sizes, namely $n\in\{1,000,\,2,000,\,4,000,\,8,000\}$, generated under the same Poisson mixture model.

\begin{figure}[t]
\centering
\includegraphics[width=.85\textwidth]{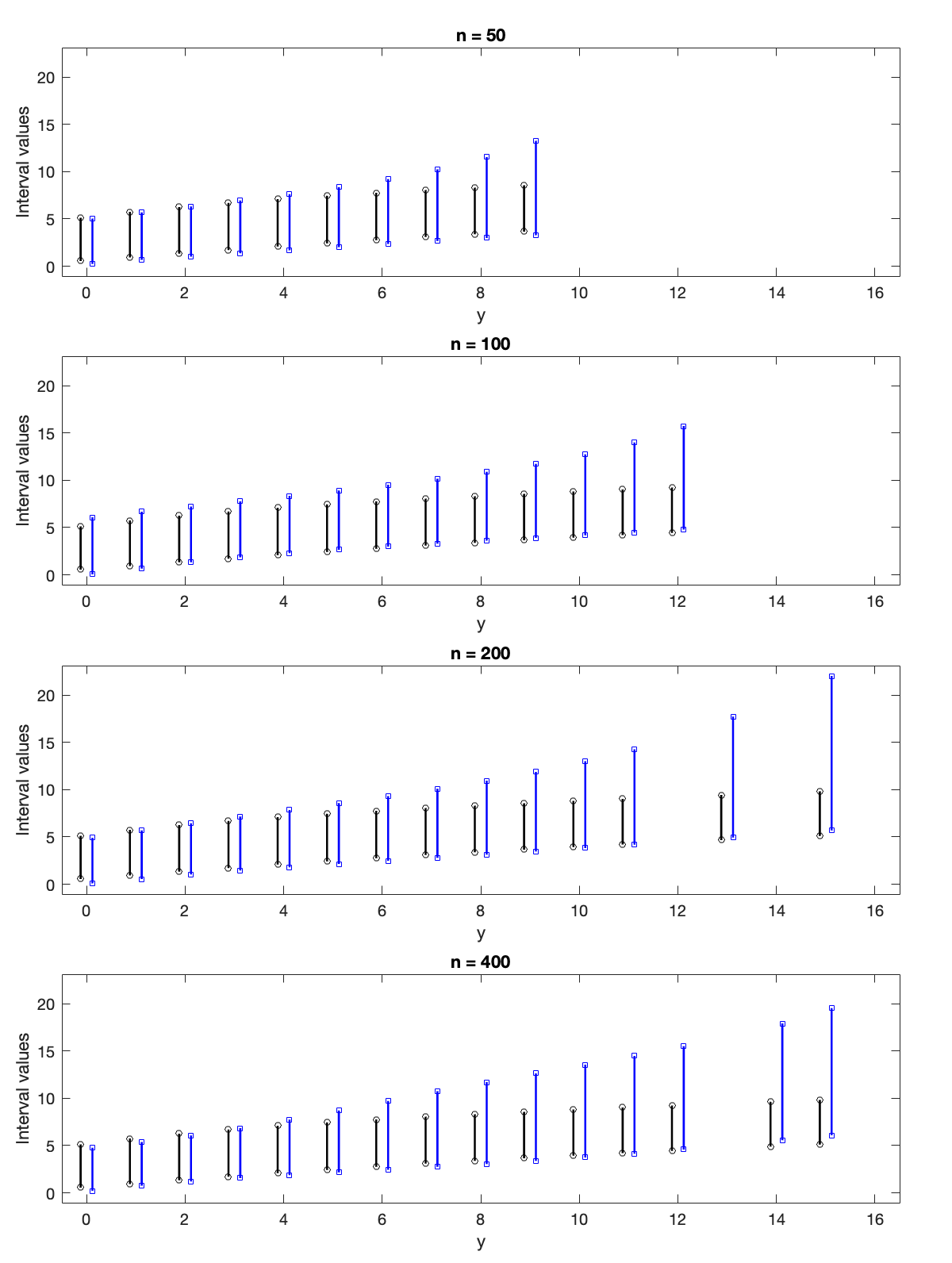}
\caption{Weibull prior: Oracle credible interval $I_{G^{\ast}}(y)$ (black), QB credible interval $I_{G_{n}}(y)$ (blue) for $Y\in\{Y_1,\dots,Y_n\}$.}
\label{weib_fig_interval}
\end{figure}

\begin{figure}[t]
\centering
\includegraphics[width=.85\textwidth]{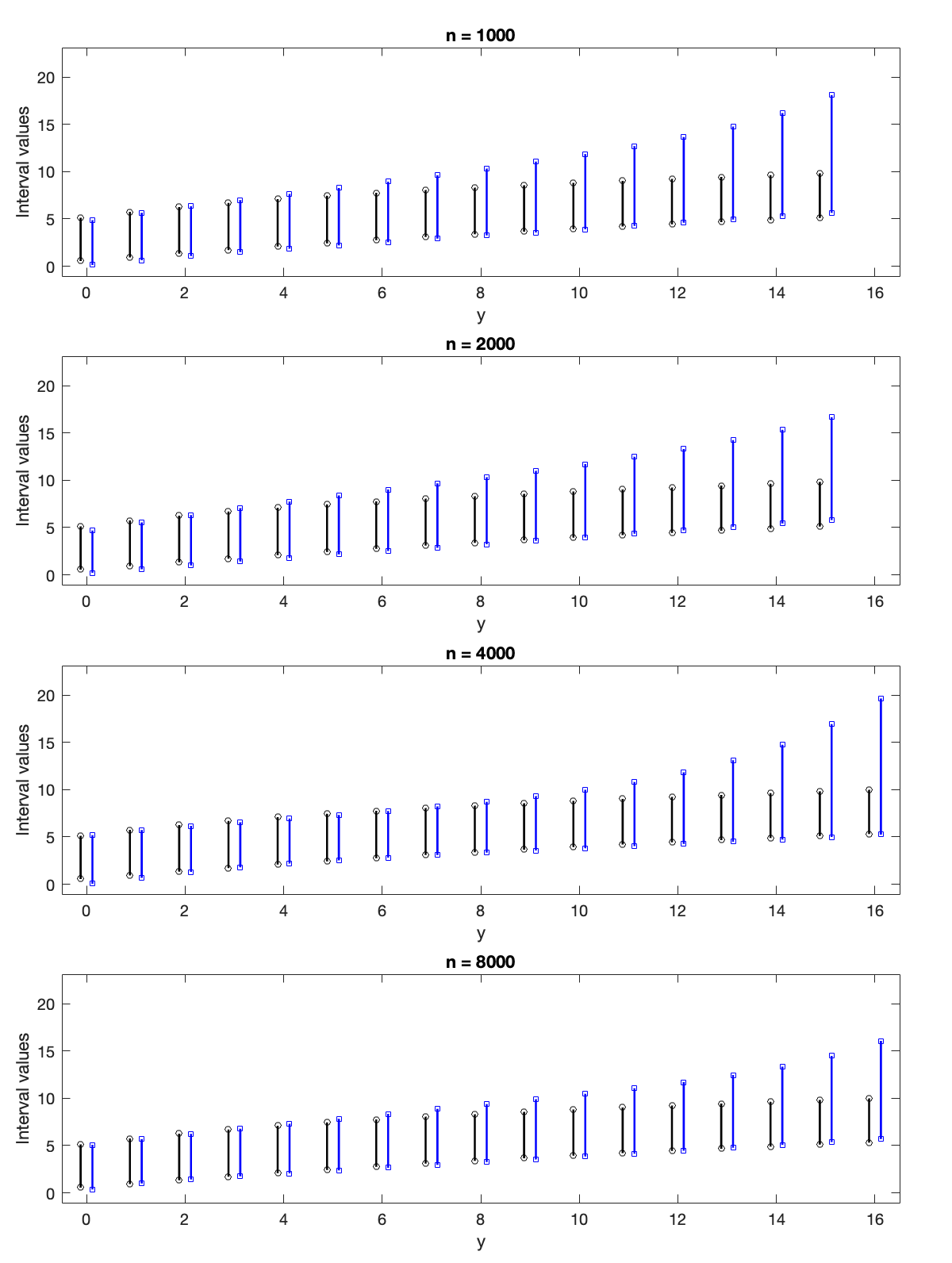}
\caption{Weibull prior: Oracle credible interval $I_{G^{\ast}}(y)$ (black), QB credible interval $I_{G_{n}}(y)$ (blue) for $Y\in\{Y_1,\dots,Y_n\}$.}
\label{weib_fig_interval1}
\end{figure}

\subsection{Half-Gaussian prior}\label{app5_num14}

For $n\in\{50,\,100,\,200,\,400\}$, we generate data $\mathbf{Y}_{n}=(Y_{1},\ldots,Y_{n})$ from the Poisson mixture model \eqref{eq:mixture_model_poisson} with a half-Gaussian prior $G$. For the QB estimate, the marginal likelihood in Newton's algorithm \eqref{eq:newton} is evaluated numerically via the trapezoidal rule. To perform this evaluation, the density function of the mixing distribution $G_{n}$ is represented through its values on a fixed uniform grid of $d\in\{5,000;\,1,000;\,500;\,100;\,50;\,10\}$ quadrature points over $\Theta=(0,U_{\Theta})$, where $U_\Theta=\max\{\max\{\mathbf{Y}_n\},\lceil Q_{n,0.99}+4(\max\{Q_{n,0.99},1\})^{1/2}\rceil\}$, with $Q_{n,0.99}=\textsc{Quantile}(\mathbf{Y}_n;0.99)$. Further, we set $G_{0}$ to be Uniform over $\Theta$, and take the learning rate to be $\alpha_{n}=(1+n)^{-0.99}$. For this setting of Newton's algorithm, Table \ref{gauss_tab_sens} provides \textsc{e-MSE}, \textsc{e-REGRET} and CPU time as the sample size $n$ and the grid resolution $d$ vary. The CPU time refers to the time (in seconds) for processing a new observation on a laptop MacBook Pro (M1 processor).

\begin{table}[ht]
\centering
\caption{Half-Gaussian prior: \textsc{e-MSE}, \textsc{e-REGRET} and CPU time (in seconds) of QB as $n$ and $d$ vary.}
{
\setlength{\tabcolsep}{0pt}
\begin{tabular}{@{}l@{\hspace{1cm}}*{6}{>{\centering\arraybackslash}p{1.75cm}}@{}}
\hline
\hline
 & $d=5{,}000$ & $d=1{,}000$ & $d=500$ & $d=100$ & $d=50$ & $d=10$ \\[0.1cm]
\hline
\multicolumn{7}{@{}l}{\underline{$n=50$}} \\[0.05cm]
\textsc{e-MSE}    & 0.280 & 0.281 & 0.281 & 0.286 & 0.294 & 0.651 \\
\textsc{e-REGRET} & 0.108 & 0.109 & 0.109 & 0.114 & 0.122 & 0.478 \\
CPU time & 0.003 & 0.002 & 0.002 & 0.001 & 0.001 & 0.001 \\[0.2cm]

\multicolumn{7}{@{}l}{\underline{$n=100$}} \\[0.05cm]
\textsc{e-MSE}    & 0.361 & 0.362 & 0.363 & 0.376 & 0.398 & 1.030 \\
\textsc{e-REGRET} & 0.117 & 0.118 & 0.120 & 0.132 & 0.154 & 0.786 \\
CPU time & 0.003 & 0.002 & 0.002 & 0.001 & 0.001 & 0.001 \\[0.2cm]

\multicolumn{7}{@{}l}{\underline{$n=200$}} \\[0.05cm]
\textsc{e-MSE}    & 0.296 & 0.295 & 0.295 & 0.293 & 0.291 & 0.492 \\
\textsc{e-REGRET} & 0.004 & 0.004 & 0.003 & 0.001 & -0.001 & 0.200 \\
CPU time & 0.003 & 0.002 & 0.002 & 0.001 & 0.001 & 0.001 \\[0.2cm]

\multicolumn{7}{@{}l}{\underline{$n=400$}} \\[0.05cm]
\textsc{e-MSE}    & 0.280 & 0.281 & 0.281 & 0.285 & 0.291 & 0.597 \\
\textsc{e-REGRET} & 0.010 & 0.010 & 0.011 & 0.014 & 0.020 & 0.327 \\
CPU time & 0.003 & 0.002 & 0.002 & 0.001 & 0.001 & 0.001 \\[0.1cm]
\hline
\hline
\end{tabular}
}
\label{gauss_tab_sens}
\end{table}

Table \ref{gauss_compare_tab} compares the \textsc{e-MSE} and \textsc{e-REGRET} for the Oracle, Robbins, ML, MHD and QB estimates. For the QB estimate we consider: i) a fixed uniform grid of $d=1,000$ quadrature points over $\Theta=(0,U_{\Theta})$; ii) an initial guess $G_{0}$ that is Uniform over $\Theta$; iii) a learning rate $\alpha_{n}=(1+n)^{-0.99}$. This setting corresponds to the second column of Table \ref{gauss_tab_sens}. Figure \ref{gauss_fig} displays the Oracle, Robbins, ML, MHD and QB estimates against the true values.

\begin{table}[ht]
\centering
\caption{Half-Gaussian prior: \textsc{e-MSE} and \textsc{e-REGRET} of Oracle, Robbins, ML, MHD and QB.}
{
\setlength{\tabcolsep}{0pt}
\begin{tabular}{@{}l@{\hspace{1cm}}*{5}{>{\centering\arraybackslash}p{1.75cm}}@{}}
\hline
\hline
 & Oracle & Robbins & ML & MHD & QB \\[0.1cm]
\hline
\multicolumn{6}{@{}l}{\underline{$n=50$}} \\[0.05cm]
\textsc{e-MSE}    & 0.172 & 1.133 & 0.456 & 0.176 & 0.281 \\
\textsc{e-REGRET} & 0.000 & 0.961 & 0.284 & 0.004 & 0.109 \\[0.3cm]

\multicolumn{6}{@{}l}{\underline{$n=100$}} \\[0.05cm]
\textsc{e-MSE}    & 0.244 & 0.980 & 0.311 & 1.287 & 0.362 \\
\textsc{e-REGRET} & 0.000 & 0.736 & 0.067 & 1.043 & 0.118 \\[0.3cm]

\multicolumn{6}{@{}l}{\underline{$n=200$}} \\[0.05cm]
\textsc{e-MSE}    & 0.292 & 0.315 & 0.295 & 0.309 & 0.295 \\
\textsc{e-REGRET} & 0.000 & 0.023 & 0.003 & 0.017 & 0.004 \\[0.3cm]

\multicolumn{6}{@{}l}{\underline{$n=400$}} \\[0.05cm]
\textsc{e-MSE}    & 0.270 & 0.284 & 0.278 & 0.282 & 0.281 \\
\textsc{e-REGRET} & 0.000 & 0.013 & 0.008 & 0.012 & 0.010 \\[0.1cm]
\hline
\hline
\end{tabular}
}
\label{gauss_compare_tab}
\end{table}

\begin{figure}[t]
\centering
\includegraphics[width=.85\textwidth]{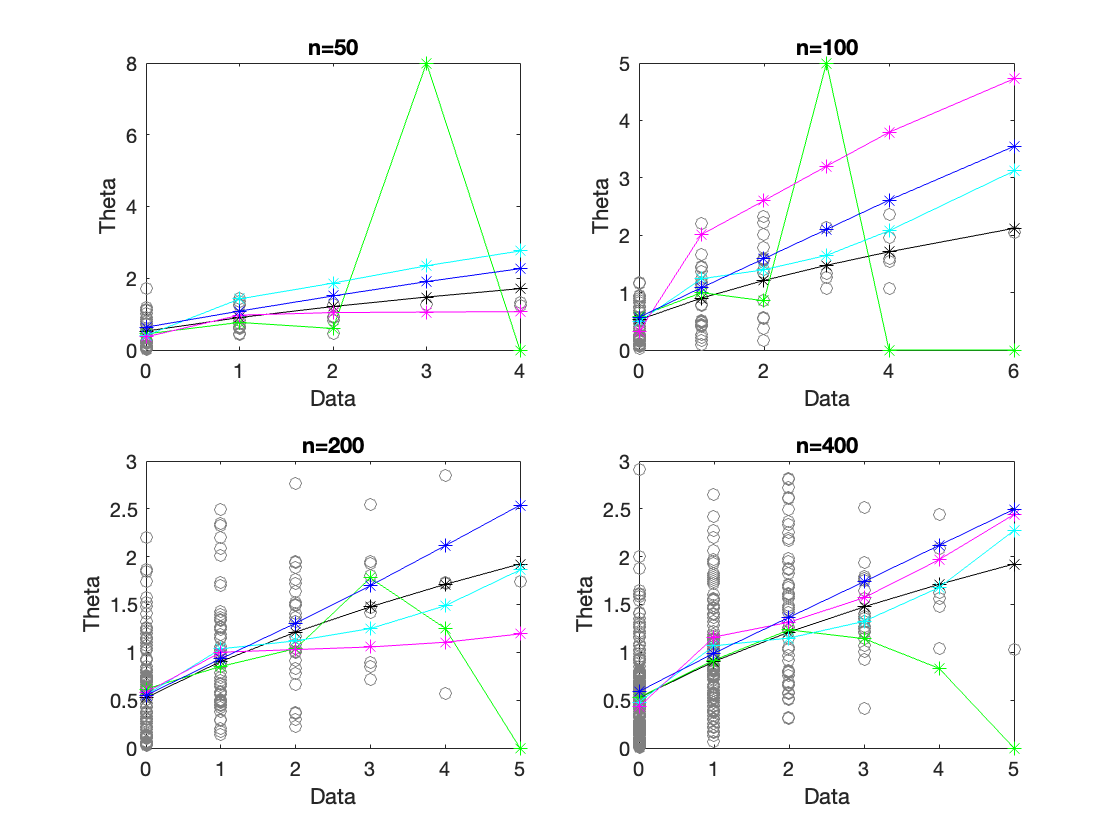}
\caption{Data points $Y$ versus true parameters $\theta$ (grey) under half-Gaussian prior, and estimates: Oracle (black), Robbins (green), ML (cyan), MHD (magenta) and QB (blue).}
\label{gauss_fig}
\end{figure}

To conclude, we provide QB credible intervals at level $1-\alpha=0.95$. Intervals are constructed by relying on Newton's algorithm initialized as in Table \ref{gauss_tab_sens},  with the optimal choice of the tuning parameters $\beta_{1}$ and $\beta_{2}$. For such QB credible intervals, Table \ref{gauss_tab_sens_interval} provides \textsc{e-MARLD} as the sample size $n$ and the grid resolution $d$ vary. 

\begin{table}[ht]
\centering
\caption{Half-Gaussian prior: \textsc{e-MARLD} of QB as $n$ and $d$ vary.}
{
\setlength{\tabcolsep}{0pt}
\begin{tabular}{@{}l@{\hspace{1cm}}*{6}{>{\centering\arraybackslash}p{1.75cm}}@{}}
\hline
\hline
 & $d=5{,}000$ & $d=1{,}000$ & $d=500$ & $d=100$ & $d=50$ & $d=10$ \\[0.1cm]
\hline
\multicolumn{7}{@{}l}{\underline{$n=50$}} \\[0.05cm]
\textsc{e-MARLD} & 0.490 & 0.491 & 0.497 & 0.547 & 0.614 & 1.161 \\[0.2cm]

\multicolumn{7}{@{}l}{\underline{$n=100$}} \\[0.05cm]
\textsc{e-MARLD} & 0.560 & 0.560 & 0.568 & 0.617 & 0.643 & 0.976 \\[0.2cm]

\multicolumn{7}{@{}l}{\underline{$n=200$}} \\[0.05cm]
\textsc{e-MARLD} & 0.203 & 0.203 & 0.211 & 0.238 & 0.325 & 0.658 \\[0.2cm]

\multicolumn{7}{@{}l}{\underline{$n=400$}} \\[0.05cm]
\textsc{e-MARLD} & 0.218 & 0.218 & 0.227 & 0.262 & 0.287 & 0.638 \\[0.1cm]
\hline
\hline
\end{tabular}
}
\label{gauss_tab_sens_interval}
\end{table}

Under the setting of Table \ref{gauss_tab_sens_interval}, and assuming $d=1,000$ quadrature points over $\Theta=(0,U_{\Theta})$, Figure \ref{gauss_fig_interval} display the QB and Oracle credible intervals. Figure \ref{gauss_fig_interval1} display the Oracle and QB credible intervals for datasets with larger sample sizes, namely $n\in\{1,000,\,2,000,\,4,000,\,8,000\}$, generated under the same Poisson mixture model.

\begin{figure}[t]
\centering
\includegraphics[width=.85\textwidth]{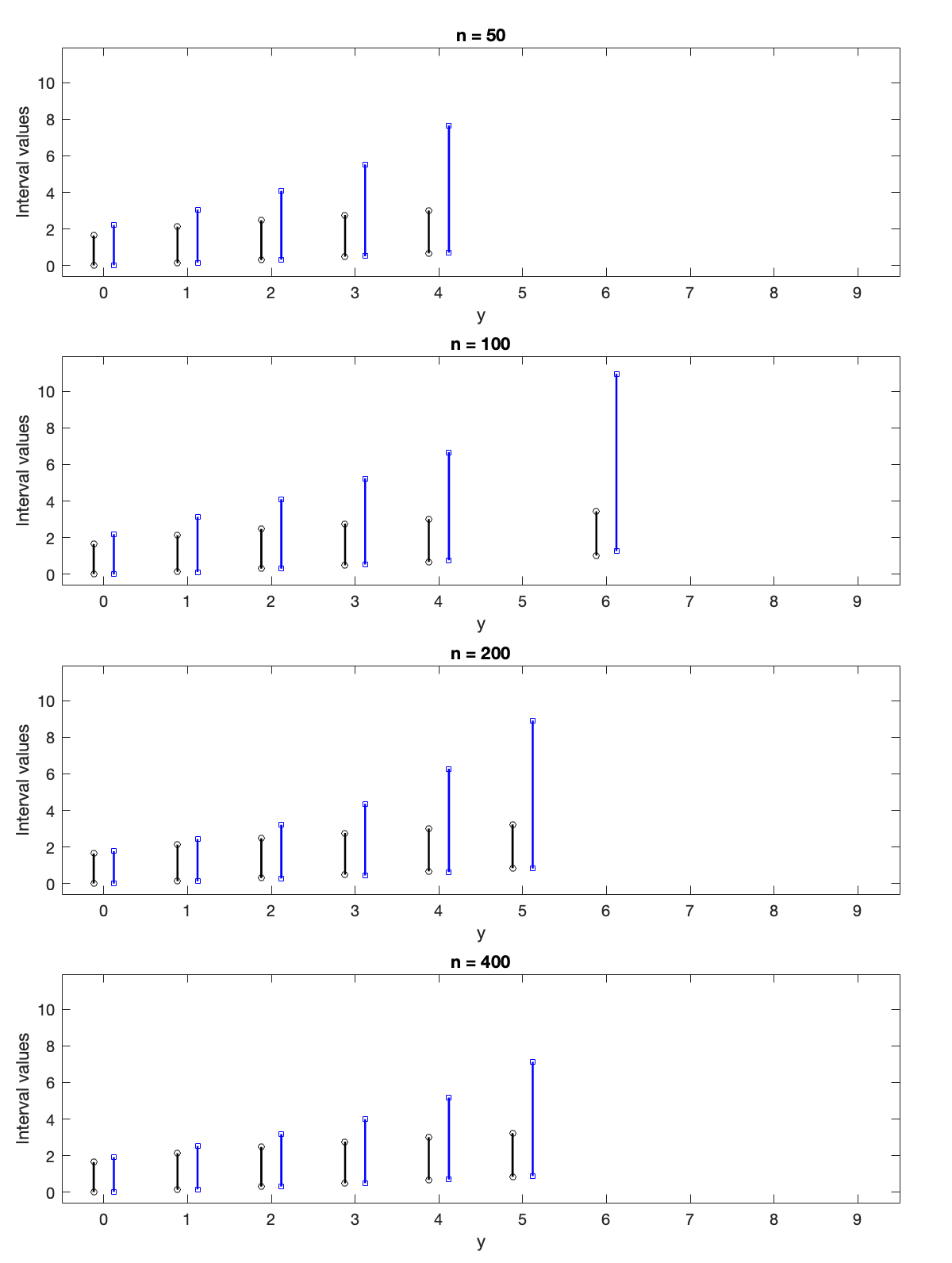}
\caption{Half-Gaussian prior: Oracle credible interval $I_{G^{\ast}}(y)$ (black), QB credible interval $I_{G_{n}}(y)$ (blue) for $Y\in\{Y_1,\dots,Y_n\}$.}
\label{gauss_fig_interval}
\end{figure}

\begin{figure}[t]
\centering
\includegraphics[width=.85\textwidth]{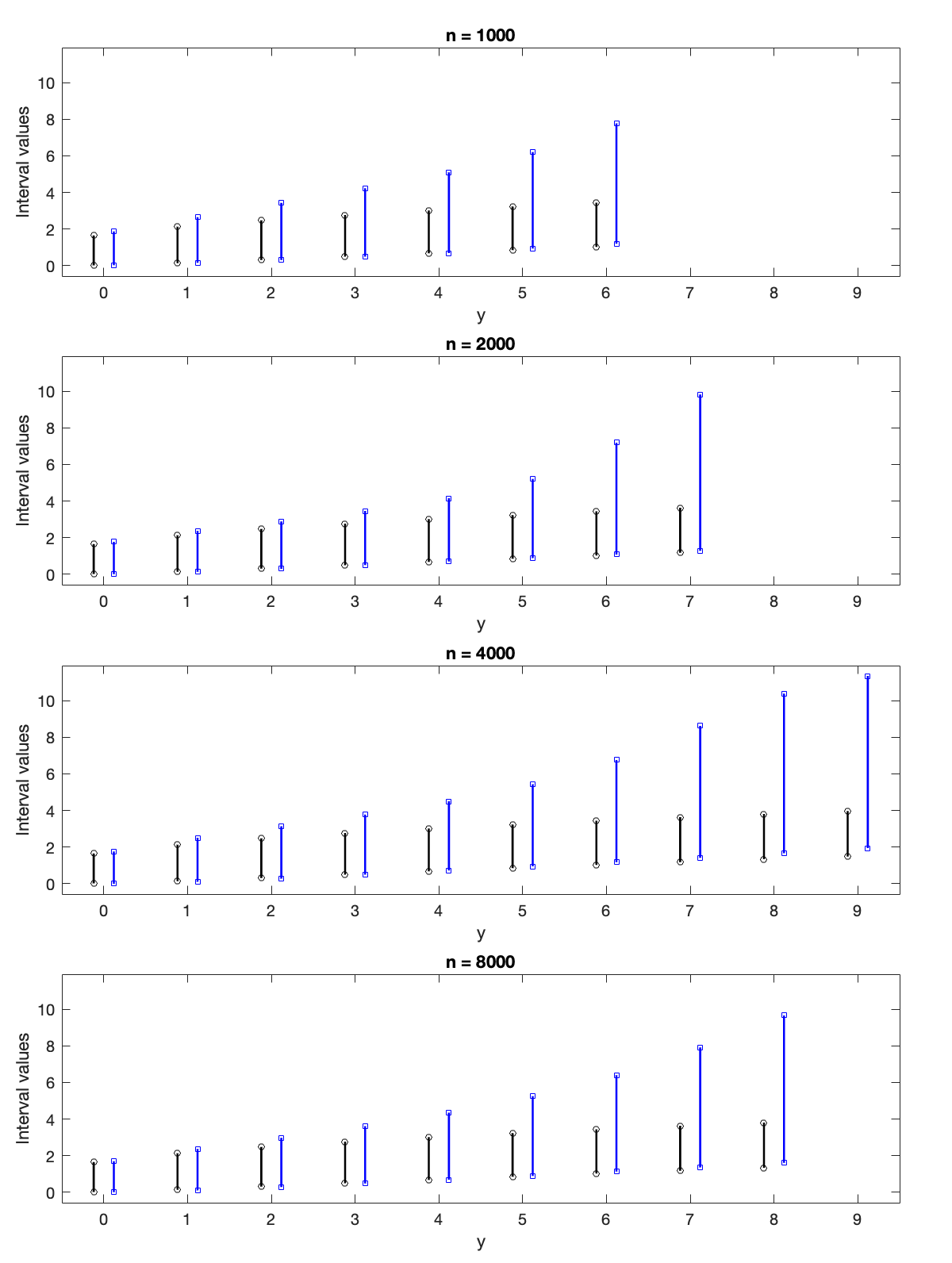}
\caption{Half-Gaussian prior: Oracle credible interval $I_{G^{\ast}}(y)$ (black), QB credible interval $I_{G_{n}}(y)$ (blue) for $Y\in\{Y_1,\dots,Y_n\}$.}
\label{gauss_fig_interval1}
\end{figure}

\subsection{Square-root of half-Cauchy prior}\label{app5_num15}

For $n\in\{50,\,100,\,200,\,400\}$, we generate data $\mathbf{Y}_{n}=(Y_{1},\ldots,Y_{n})$ from the Poisson mixture model \eqref{eq:mixture_model_poisson} with a square-root of half-Cauchy prior $G$. For the QB estimate, the marginal likelihood in Newton's algorithm \eqref{eq:newton} is evaluated numerically via the trapezoidal rule. To perform this evaluation, the density function of the mixing distribution $G_{n}$ is represented through its values on a fixed uniform grid of $d\in\{5,000;\,1,000;\,500;\,100;\,50;\,10\}$ quadrature points over $\Theta=(0,U_{\Theta})$, where $U_\Theta=\max\{\max\{\mathbf{Y}_n\},\lceil Q_{n,0.99}+4(\max\{Q_{n,0.99},1\})^{1/2}\rceil\}$, with $Q_{n,0.99}=\textsc{Quantile}(\mathbf{Y}_n;0.99)$. Further, we set $G_{0}$ to be Uniform over $\Theta$, and take the learning rate to be $\alpha_{n}=(1+n)^{-0.99}$. For this setting of Newton's algorithm, Table \ref{cau_tab_sens} provides \textsc{e-MSE}, \textsc{e-REGRET} and CPU time as the sample size $n$ and the grid resolution $d$ vary. The CPU time refers to the time (in seconds) for processing a new observation on a laptop MacBook Pro (M1 processor).

\begin{table}[ht]
\centering
\caption{Square-root of half-Cauchy prior: \textsc{e-MSE}, \textsc{e-REGRET} and CPU time (in seconds) of QB as $n$ and $d$ vary.}
{
\setlength{\tabcolsep}{0pt}
\begin{tabular}{@{}l@{\hspace{1cm}}*{6}{>{\centering\arraybackslash}p{1.75cm}}@{}}
\hline
\hline
 & $d=5{,}000$ & $d=1{,}000$ & $d=500$ & $d=100$ & $d=50$ & $d=10$ \\[0.1cm]
\hline
\multicolumn{7}{@{}l}{\underline{$n=50$}} \\[0.05cm]
\textsc{e-MSE}    & 0.773 & 0.773 & 0.773 & 0.774 & 0.775 & 1.136 \\
\textsc{e-REGRET} & 0.108 & 0.108 & 0.108 & 0.108 & 0.109 & 0.470 \\
CPU time & 0.003 & 0.002 & 0.002 & 0.001 & 0.001 & 0.001 \\[0.2cm]

\multicolumn{7}{@{}l}{\underline{$n=100$}} \\[0.05cm]
\textsc{e-MSE}    & 0.515 & 0.514 & 0.513 & 0.509 & 0.545 & 10.613 \\
\textsc{e-REGRET} & 0.105 & 0.104 & 0.103 & 0.099 & 0.135 & 10.203 \\
CPU time & 0.003 & 0.002 & 0.002 & 0.001 & 0.001 & 0.001 \\[0.2cm]

\multicolumn{7}{@{}l}{\underline{$n=200$}} \\[0.05cm]
\textsc{e-MSE}    & 0.585 & 0.585 & 0.585 & 0.586 & 0.588 & 1.062 \\
\textsc{e-REGRET} & -0.014 & -0.014 & -0.014 & -0.013 & -0.011 & 0.463 \\
CPU time & 0.003 & 0.002 & 0.002 & 0.001 & 0.001 & 0.001 \\[0.2cm]

\multicolumn{7}{@{}l}{\underline{$n=400$}} \\[0.05cm]
\textsc{e-MSE}    & 0.594 & 0.594 & 0.594 & 0.593 & 0.592 & 1.798 \\
\textsc{e-REGRET} & 0.015 & 0.015 & 0.015 & 0.014 & 0.013 & 1.218 \\
CPU time & 0.003 & 0.002 & 0.002 & 0.001 & 0.001 & 0.001 \\[0.1cm]
\hline
\hline
\end{tabular}
}
\label{cau_tab_sens}
\end{table}

Table \ref{cau_compare_tab} compares the \textsc{e-MSE} and \textsc{e-REGRET} for the Oracle, Robbins, ML, MHD and QB estimates. For the QB estimate we consider: i) a fixed uniform grid of $d=1,000$ quadrature points over $\Theta=(0,U_{\Theta})$; ii) an initial guess $G_{0}$ that is Uniform over $\Theta$; iii) a learning rate $\alpha_{n}=(1+n)^{-0.99}$. This setting corresponds to the second column of Table \ref{cau_tab_sens}. Figure \ref{cau_fig} displays the Oracle, Robbins, ML, MHD and QB estimates against the true values.

\begin{table}[ht]
\centering
\caption{Square-root of half-Cauchy prior: \textsc{e-MSE} and \textsc{e-REGRET} of Oracle, Robbins, ML, MHD and QB.}
{
\setlength{\tabcolsep}{0pt}
\begin{tabular}{@{}l@{\hspace{1cm}}*{5}{>{\centering\arraybackslash}p{1.75cm}}@{}}
\hline
\hline
 & Oracle & Robbins & ML & MHD & QB \\[0.1cm]
\hline
\multicolumn{6}{@{}l}{\underline{$n=50$}} \\[0.05cm]
\textsc{e-MSE}    & 0.666 & 2.861 & 1.614 & 0.939 & 0.773 \\
\textsc{e-REGRET} & 0.000 & 2.196 & 0.948 & 0.273 & 0.108 \\[0.3cm]

\multicolumn{6}{@{}l}{\underline{$n=100$}} \\[0.05cm]
\textsc{e-MSE}    & 0.410 & 15.520 & 1.283 & 2.577 & 0.514 \\
\textsc{e-REGRET} & 0.000 & 15.110 & 0.873 & 2.166 & 0.104 \\[0.3cm]

\multicolumn{6}{@{}l}{\underline{$n=200$}} \\[0.05cm]
\textsc{e-MSE}    & 0.599 & 1.106 & 1.409 & 0.621 & 0.585 \\
\textsc{e-REGRET} & 0.000 & 0.507 & 0.810 & 0.022 & -0.014 \\[0.3cm]

\multicolumn{6}{@{}l}{\underline{$n=400$}} \\[0.05cm]
\textsc{e-MSE}    & 0.579 & 1.971 & 0.668 & 0.649 & 0.594 \\
\textsc{e-REGRET} & 0.000 & 1.391 & 0.088 & 0.070 & 0.015 \\[0.1cm]
\hline
\hline
\end{tabular}
}
\label{cau_compare_tab}
\end{table}

\begin{figure}[t]
\centering
\includegraphics[width=.85\textwidth]{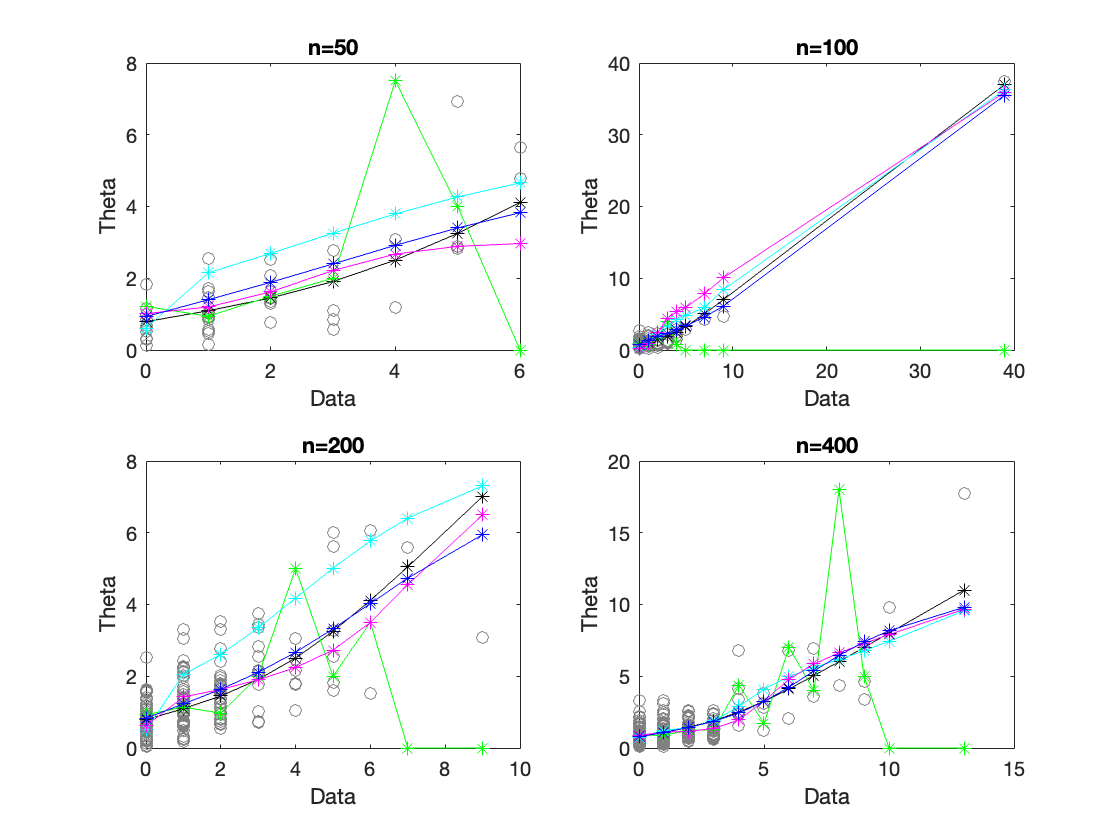}
\caption{Data points $Y$ versus true parameters $\theta$ (grey) under square-root of half-Cauchy prior, and estimates: Oracle (black), Robbins (green), ML (cyan), MHD (magenta) and QB (blue).}
\label{cau_fig}
\end{figure}

To conclude, we provide QB credible intervals at level $1-\alpha=0.95$. Intervals are constructed by relying on Newton's algorithm initialized as in Table \ref{cau_tab_sens},  with the optimal choice of the tuning parameters $\beta_{1}$ and $\beta_{2}$. For such QB credible intervals, Table \ref{cau_tab_sens_interval} provides \textsc{e-MARLD} as the sample size $n$ and the grid resolution $d$ vary. 

\begin{table}[ht]
\centering
\caption{Square-root of half-Cauchy prior: \textsc{e-MARLD} of QB as $n$ and $d$ vary.}
{
\setlength{\tabcolsep}{0pt}
\begin{tabular}{@{}l@{\hspace{1cm}}*{6}{>{\centering\arraybackslash}p{1.75cm}}@{}}
\hline
\hline
 & $d=5{,}000$ & $d=1{,}000$ & $d=500$ & $d=100$ & $d=50$ & $d=10$ \\[0.1cm]
\hline
\multicolumn{7}{@{}l}{\underline{$n=50$}} \\[0.05cm]
\textsc{e-MARLD} & 0.382 & 0.382 & 0.388 & 0.435 & 0.495 & 0.742 \\[0.2cm]

\multicolumn{7}{@{}l}{\underline{$n=100$}} \\[0.05cm]
\textsc{e-MARLD} & 0.374 & 0.374 & 0.378 & 0.555 & 0.753 & 1.737 \\[0.2cm]

\multicolumn{7}{@{}l}{\underline{$n=200$}} \\[0.05cm]
\textsc{e-MARLD} & 0.208 & 0.208 & 0.216 & 0.284 & 0.379 & 0.810 \\[0.2cm]

\multicolumn{7}{@{}l}{\underline{$n=400$}} \\[0.05cm]
\textsc{e-MARLD} & 0.065 & 0.065 & 0.075 & 0.141 & 0.273 & 0.811 \\[0.1cm]
\hline
\hline
\end{tabular}
}
\label{cau_tab_sens_interval}
\end{table}

Under the setting of Table \ref{cau_tab_sens_interval}, and assuming $d=1,000$ quadrature points over $\Theta=(0,U_{\Theta})$, Figure \ref{cau_fig_interval} display the QB and Oracle credible intervals. Figure \ref{cau_fig_interval1} display the Oracle and QB credible intervals for datasets with larger sample sizes, namely $n\in\{1,000,\,2,000,\,4,000,\,8,000\}$, generated under the same Poisson mixture model.

\begin{figure}[t]
\centering
\includegraphics[width=.85\textwidth]{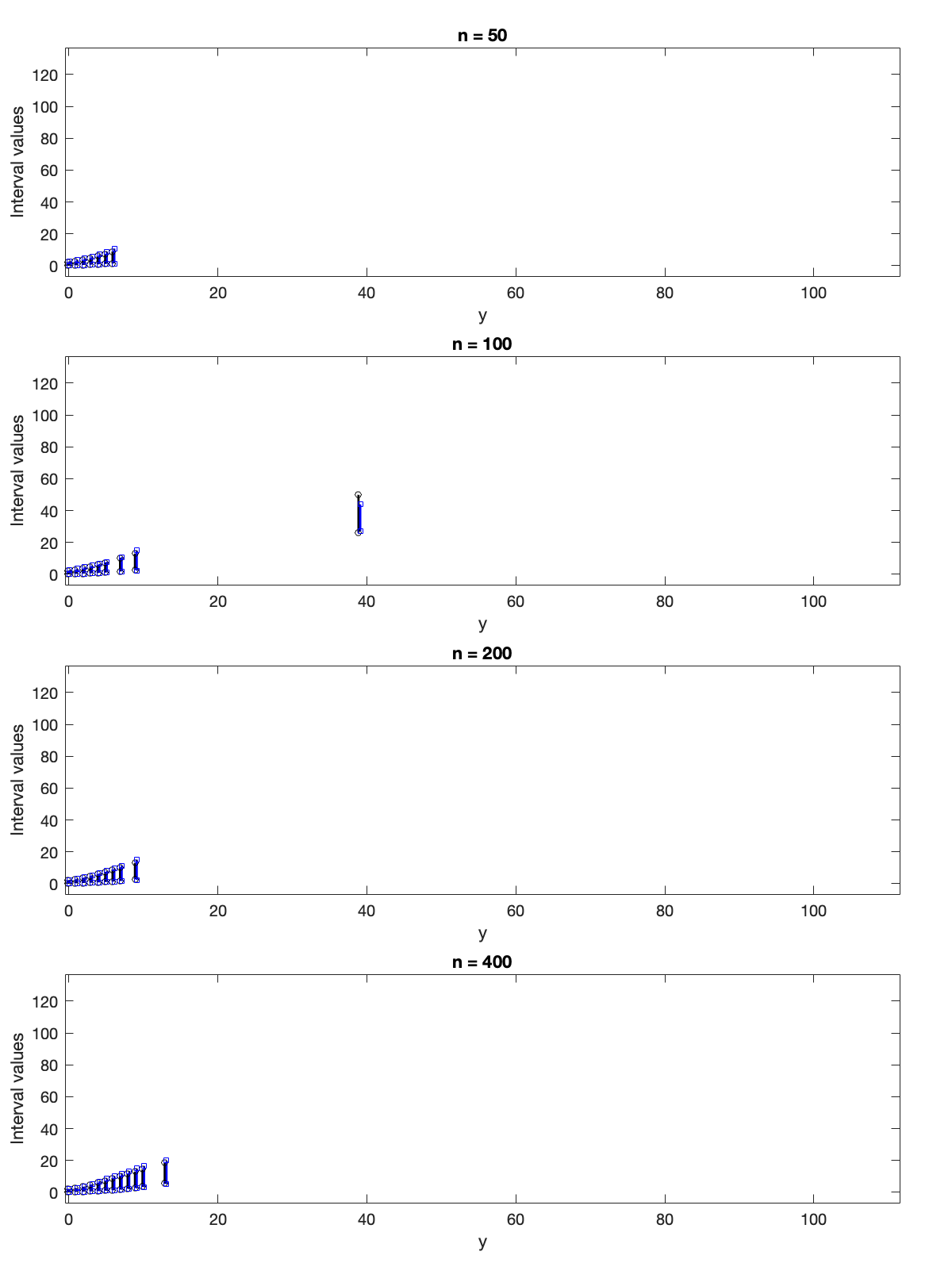}
\caption{Square-root of half-Cauchy prior: Oracle credible interval $I_{G^{\ast}}(y)$ (black), QB credible interval $I_{G_{n}}(y)$ (blue) for $Y\in\{Y_1,\dots,Y_n\}$.}
\label{cau_fig_interval}
\end{figure}

\begin{figure}[t]
\centering
\includegraphics[width=.85\textwidth]{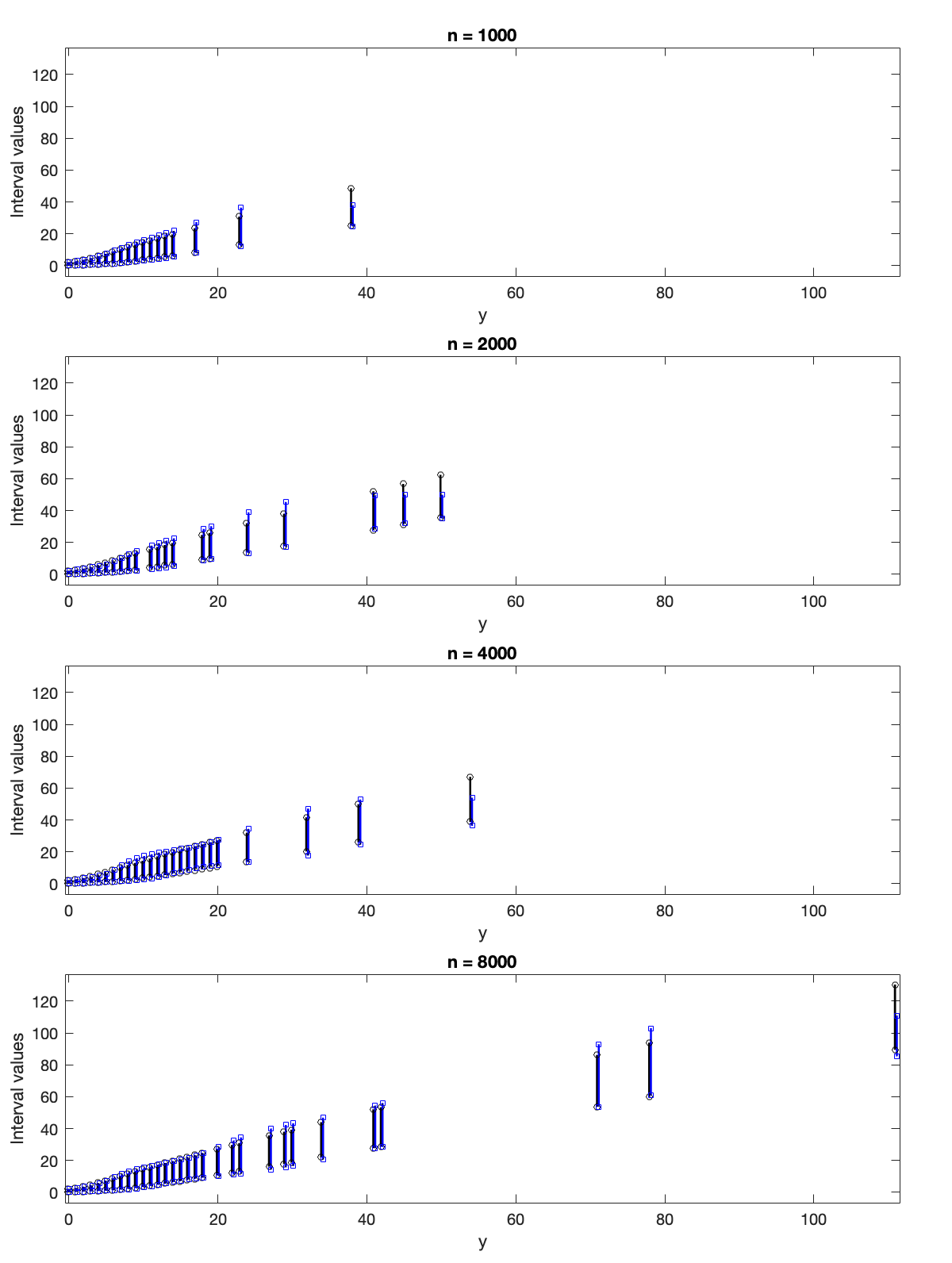}
\caption{Square-root of half-Cauchy prior: Oracle credible interval $I_{G^{\ast}}(y)$ (black), QB credible interval $I_{G_{n}}(y)$ (blue) for $Y\in\{Y_1,\dots,Y_n\}$.}
\label{cau_fig_interval1}
\end{figure}

\section{Numerical illustrations: Twitter dataset}\label{app5_num2}

We analyze a dataset of tweets that received at least $50$ retweets on Twitter between October 7 and November 6, 2011, available at \url{https://snap.stanford.edu/seismic/}. It contains $166,783$ tweets generating $34,617,704$ retweets, along with metadata such as tweet ID, posting time, retweet time, and follower counts. For each tweet, we focus on the number of retweets received within the first $30$ seconds after posting.

For each tweet, we consider the number of retweets generated within $30$ seconds from the tweet post time. From the first two rows of Table \ref{tweet_tab1}, $40,259$ tweets have $0$ retweets in $30$ seconds, $28,339$ tweets have $1$ retweet in $30$ seconds, etc. Table \ref{tweet_tab1} and Figure \ref{tweet_fig1} display the estimates of the number of retweets expected for a tweet.

For each tweet, we focus on the number of followers generated by retweets within $30$ seconds from the tweet post time, namely the total number of followers of fast retweets of a tweet, with the proviso that $0$ retweets corresponds to $0$ followers generated. From the first two rows of Table \ref{tweet_tab_fol1}, $40,287$ tweets have generated $0$ followers in $30$ seconds, $67$ tweets have generated $1$ follower in $30$ seconds, etc. Table \ref{tweet_tab_fol1} and Figure \ref{tweet_fig_fol1} display the estimates of the number of followers expected to be generated by the retweets of a tweet.

\begin{table}[ht]
\centering
\caption{Counts $n_y$ of the number of tweets that have $y$ retweets within $30$ seconds from the tweet post time, and corresponding estimates of the expected number of retweets.}
\label{tweet_tab1}
{
\setlength{\tabcolsep}{0pt}

\begin{tabular}{@{}l@{\hspace{0.8cm}}*{8}{>{\centering\arraybackslash}p{1.55cm}}@{}}
\hline
\hline
$y$    & 0 & 1 & 2 & 3 & 4 & 5 & 6 & 7 \\[0.1cm]
$n_y$  & 40,259 & 28,339 & 21,581 & 16,479 & 12,130 & 9,238 & 7,193 & 5,464 \\[0.1cm]
\hline
Robbins  & 0.703 & 1.521 & 2.294 & 2.940 & 3.812 & 4.673 & 5.324 & 6.322 \\
ML  & 0.391 & 2.134 & 2.932 & 3.703 & 4.460 & 5.333 & 6.331 & 7.422 \\
MHD  & 0.752 & 2.201 & 3.144 & 4.152 & 4.993 & 5.691 & 6.392 & 7.164 \\[0.2cm]
QB  & 0.614 & 0.821 & 1.633 & 3.084 & 4.012 & 4.523 & 5.024 & 5.711 \\[0.05cm]
\hline
\end{tabular}

\par\medskip

\begin{tabular}{@{}l@{\hspace{0.8cm}}*{8}{>{\centering\arraybackslash}p{1.55cm}}@{}}
\hline
$y$    & 8 & 11 & 14 & 17 & 20 & 23 & 26 & 29 \\[0.1cm]
$n_y$  & 4,319 & 2,291 & 1,352 & 817 & 475 & 293 & 188 & 129 \\[0.1cm]
\hline
Robbins  & 7.180 & 10.102 & 11.774 & 15.331 & 18.663 & 21.131 & 22.402 & 28.374 \\
ML  & 8.542 & 11.790 & 15.014 & 18.302 & 20.894 & 22.963 & 25.871 & 30.473 \\
MHD  & 8.044 & 10.873 & 13.020 & 14.832 & 16.681 & 18.823 & 22.501 & 29.214 \\[0.2cm]
QB  & 6.674 & 10.372 & 12.691 & 14.643 & 17.554 & 21.381 & 24.782 & 27.531 \\[0.05cm]
\hline
\end{tabular}

\par\medskip

\begin{tabular}{@{}l@{\hspace{0.8cm}}*{8}{>{\centering\arraybackslash}p{1.55cm}}@{}}
\hline
$y$    & 32 & 35 & 38 & 41 & 44 & 47 & 50 & 53 \\[0.1cm]
$n_y$  & 89 & 69 & 60 & 61 & 36 & 36 & 32 & 21 \\[0.1cm]
\hline
Robbins  & 30.784 & 32.350 & 30.552 & 26.853 & 48.751 & 45.334 & 43.032 & 41.143 \\
ML  & 34.880 & 37.752 & 39.813 & 41.734 & 43.791 & 46.183 & 49.131 & 52.722 \\
MHD  & 35.482 & 39.081 & 41.634 & 44.051 & 46.583 & 49.212 & 51.934 & 54.701 \\[0.2cm]
QB  & 30.263 & 33.322 & 36.594 & 39.712 & 42.504 & 45.051 & 47.543 & 50.092 \\[0.05cm]
\hline
\end{tabular}

\par\medskip

\begin{tabular}{@{}l@{\hspace{0.8cm}}*{8}{>{\centering\arraybackslash}p{1.55cm}}@{}}
\hline
$y$    & 56 & 59 & 62 & 65 & 68 & 71 & 74 & 77 \\[0.1cm]
$n_y$  & 24 & 16 & 14 & 13 & 14 & 11 & 2 & 5 \\[0.1cm]
\hline
Robbins  & 42.752 & 48.751 & 45.004 & 91.393 & 14.790 & 19.642 & 225.003 & 62.401 \\
ML  & 56.584 & 60.080 & 62.881 & 65.062 & 66.843 & 68.381 & 69.772 & 71.010 \\
MHD  & 57.471 & 60.213 & 62.842 & 65.313 & 67.544 & 69.491 & 71.152 & 72.531 \\[0.2cm]
QB  & 52.744 & 55.401 & 57.943 & 60.262 & 62.332 & 64.184 & 65.871 & 67.452 \\[0.1cm]
\hline
\hline
\end{tabular}
}
\end{table}

\begin{figure}[h!]
\begin{center}
\includegraphics[width = \textwidth]{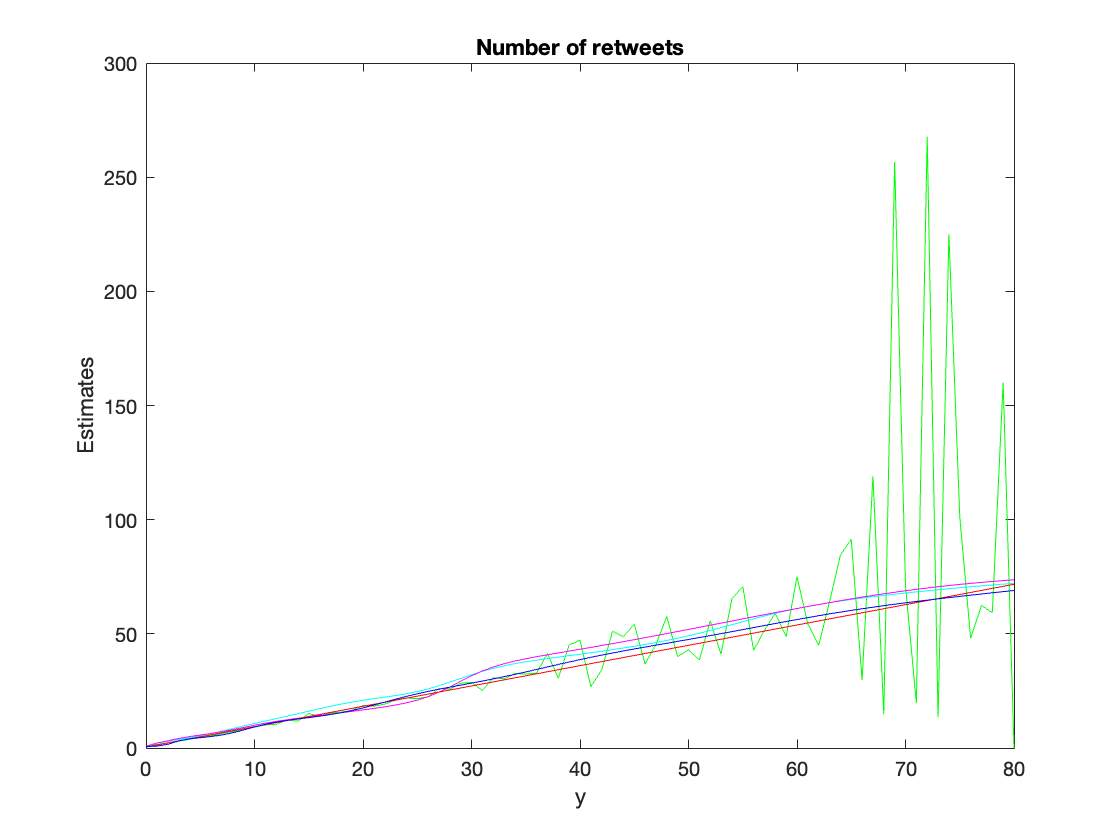}
\end{center}
\caption{\scriptsize{Estimates of the number of retweets expected for a tweet with $y$ retweets (in 30 seconds): Robbins (green line), ML (cyan line), MHD (magenta line), P-EB (red line) and QB (blue line).}}
\label{tweet_fig1}
\end{figure}

\begin{table}[ht]
\centering
\caption{Counts $n_y$ of the number of tweets that have generated $y$ followers within $30$ seconds from the tweet post time, for $y=0,1,\ldots,450$, and corresponding estimates of the expected number of followers.}
\label{tweet_tab_fol1}
{
\setlength{\tabcolsep}{0pt}

\begin{tabular}{@{}l@{\hspace{0.8cm}}*{8}{>{\centering\arraybackslash}p{1.55cm}}@{}}
\hline
\hline
$y$    & 0 & 30 & 60 & 90 & 120 & 150 & 180 & 210 \\[0.1cm]
$n_y$  & 40,287 & 67 & 87 & 83 & 98 & 82 & 84 & 63 \\[0.1cm]
\hline
Robbins  & 0.002 & 31.462 & 55.391 & 104.163 & 101.252 & 147.324 & 178.851 & 247.843 \\
ML  & 0.001 & 31.671 & 61.103 & 91.012 & 120.991 & 149.462 & 180.194 & 210.792 \\
MHD  & 0.782 & 29.933 & 60.621 & 93.261 & 121.192 & 148.303 & 181.521 & 213.794 \\[0.2cm]
QB  & 0.541 & 31.414 & 60.902 & 91.193 & 120.681 & 150.894 & 180.312 & 210.853 \\[0.05cm]
\hline
\end{tabular}

\par\medskip

\begin{tabular}{@{}l@{\hspace{0.8cm}}*{8}{>{\centering\arraybackslash}p{1.55cm}}@{}}
\hline
$y$    & 240 & 270 & 300 & 330 & 360 & 390 & 420 & 450 \\[0.1cm]
$n_y$  & 85 & 76 & 65 & 81 & 81 & 63 & 75 & 71 \\[0.1cm]
\hline
Robbins  & 275.023 & 274.571 & 342.684 & 314.652 & 352.094 & 378.592 & 370.481 & 330.314 \\
ML  & 240.812 & 270.901 & 300.963 & 330.974 & 361.002 & 391.053 & 421.044 & 450.462 \\
MHD  & 244.113 & 270.944 & 300.951 & 330.982 & 360.984 & 390.963 & 420.814 & 449.933 \\[0.2cm]
QB  & 240.402 & 271.034 & 300.572 & 330.713 & 360.281 & 390.421 & 418.993 & 442.982 \\[0.05cm]
\hline
\hline
\end{tabular}
}
\end{table}

\begin{figure}[h!]
\begin{center}
\includegraphics[width = \textwidth]{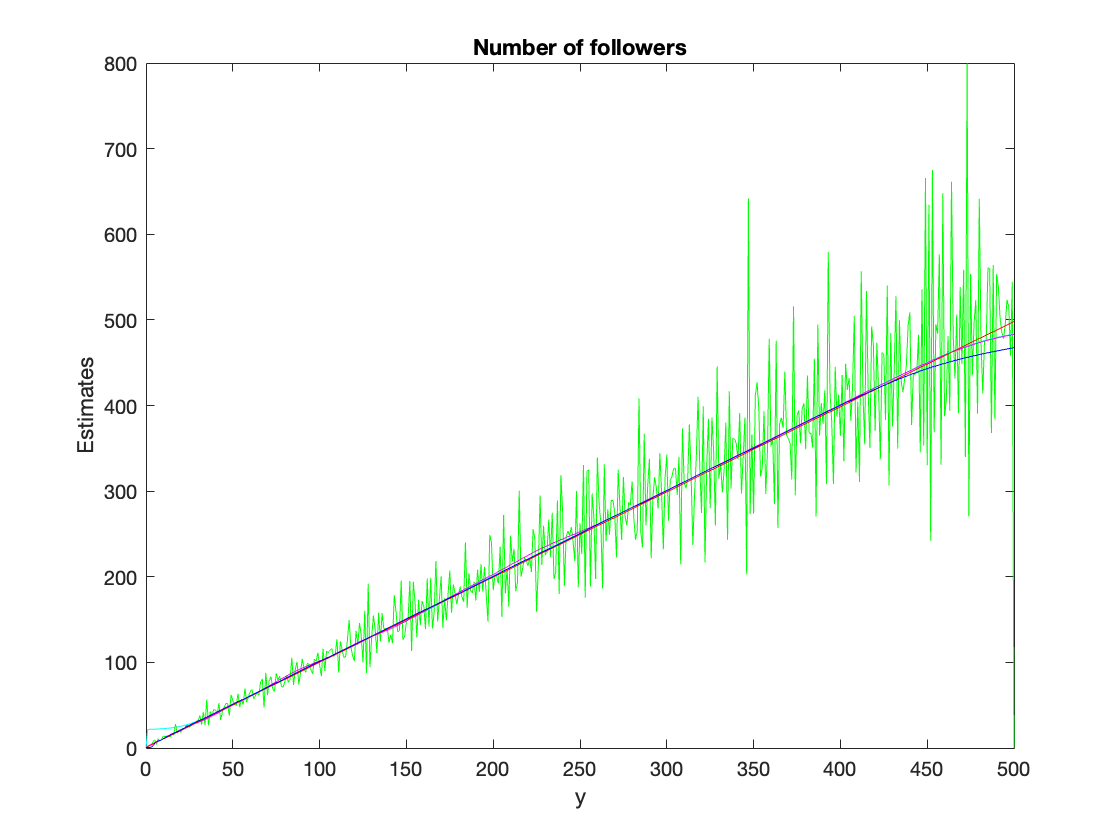}
\end{center}
\caption{\scriptsize{Estimates of the number of followers expected to be generated by a tweet whose retweets have $y$ followers (in 30 seconds): Robbins (green line), ML (cyan line), MHD (magenta line) and QB (blue line).}}
\label{tweet_fig_fol1}
\end{figure}

\section{Numerical illustrations: auto accident dataset}\label{app5_num3}

We consider the auto accident data of \citet[Table 6.1]{Efr(21)}. This is a benchmark dataset in empirical Bayes. These are one year of claims data for a European automobile insurance company on $9461$ automobile insurance policy holders. From the first two rows of Table \ref{efron_tab}, $7840$ of the $9461$ policy holders made no claims during the year, $1317$ made a single claim, $239$ made two claims each, etc., continuing to the one person who made seven claims. Interest is in estimating the number of claims expected in a succeeding year by each automobile insurance policy holder. Table \ref{efron_tab} displays the estimates of the number of claims expected. 

\begin{table}[ht]
\centering
\caption{Counts $n_y$ of number of claims $y$ made in a single year, for $y=0,1,\ldots,7$, and corresponding estimates of the expected number of claims.}
\label{efron_tab}
{
\setlength{\tabcolsep}{0pt}
\begin{tabular}{@{}l@{\hspace{0.8cm}}*{8}{>{\centering\arraybackslash}p{1.55cm}}@{}}
\hline
\hline
Claims $y$ & 0 & 1 & 2 & 3 & 4 & 5 & 6 & 7 \\[0.1cm]
Counts $n_y$ & 7,840 & 1,317 & 239 & 42 & 14 & 4 & 4 & 1 \\[0.1cm]
\hline
Robbins & 0.173 & 0.361 & 0.534 & 1.333 & 1.432 & 6.004 & 1.752 & 0.001 \\
ML & 0.334 & 1.391 & 2.313 & 3.392 & 4.301 & 4.934 & 5.362 & 5.691 \\
MHD & 0.154 & 0.413 & 0.822 & 1.003 & 1.171 & 1.794 & 3.241 & 4.462 \\[0.2cm]
QB & 0.121 & 0.254 & 0.473 & 0.912 & 1.853 & 3.174 & 4.331 & 5.222 \\[0.1cm]
\hline
\hline
\end{tabular}
}
\end{table}

\section{Multidimensional extension} \label{sec:multiple}

Throughout, we denote by $k(\cdot\,|\,\theta)$ the Poisson kernel with mean $\theta>0$, i.e. $k(y\,|\,\theta)=\theta^{y}\text{e}^{-\theta}/y!$ for $y\in\mathbb{N}_{0}$. The results in Section  \ref{sec:quasi}  extend, with minor modifications, to the k-dimension coordinatewise
independent Poisson model.
  We assume that $\Yb=(\Yb_{1},\dots,\Yb_{k})$ is a $k$-dimensional random vector with independent coordinates $\Yb_{i}\sim \text{Poisson}(\thetab_{i})$, and denote
\begin{equation}
    \label{eq:kmultiple}
k(\mathbf y\mid \thetab)=\prod_{j=1}^ke^{-\thetab_j}\frac{\thetab_j^{\yb_j}}{\yb_j!}, \quad \yb=(\yb_1,\dots,\yb_k)\in\mathbb N_0^k, \thetab_j\in\Theta\; (j=1,\dots,k).
\end{equation}
Denote $\thetab=(\thetab_1,\dots,\thetab_k)\in\Theta^k$. 
If $\thetab$ has probability distribution $G$ on $\Theta^k$, then the marginal probability mass function of $\Yb=(\Yb_{1},\dots,\Yb_k)$ is
\begin{equation}\label{eq:marginalmultiple}
p_G(\yb)=\int_{\Theta^k}k(\yb\mid\thetab)G(\ddr\thetab).
\end{equation}
 Given a sequence $(\Yb_n)_{n\geq 1}$ of random vectors with independent coordinates $\Yb_{n,j}\sim\text{Poisson}(\thetab_{n,j})$, we consider the problem of estimating the unknown parameters $\thetab_{n,j}$. 
 
If we knew that the random vectors $(\thetab_{n})_{n\geq 1}$ were i.i.d. with a known probability distribution $G$ on $\Theta^k$, we could estimate $\thetab_{n,j}$, given that $\Yb_n=\yb$, with
\begin{equation}
    \label{eq:thetahatmultiple}
    \hat\thetab_{G,j}(\yb)
=\EE_g[\thetab_{n,j}|\Yb_n=\yb]
=(\yb_j+1)\frac{p_G(\yb+\eb_j)}{p_G(\yb)},
\end{equation}
where $\eb_j$ is a $k$-dimensional vector with $\eb_{j,i}=0$ for $i\neq j$ and $\eb_{j,j}=1$. 

Being $G$ unknown, we can estimate it recursively as 
\begin{equation}
    \label{eq:newtonmultiple}
G_{n+1}(\ddr\thetab)=(1-\alpha_{n+1})G_{n}(\ddr\thetab)+\alpha_{n+1}\frac{k(\Yb_{n+1}\mid\thetab)G_{n}(\ddr\thetab)}{\int_{\Theta^k}k(\Yb_{n+1}\mid\thetab)G_{n}(\ddr\thetab)},
\end{equation}
and adopt for $\thetab_{G,j}(\yb)$ the plug-in estimate
\begin{equation}\label{seq_estim2}
\hat{\thetab}_{G_{n},j}(\yb)=(\yb_j+1)\frac{p_{G_{n}}(\yb+\eb_j)}{p_{G_{n}}(\yb)}\qquad \yb\in\mathbb{N}_{0}^k,
\end{equation}
where
\begin{equation}
    \label{eq:pgnmultiple}
p_{G_n}(\yb)=\int_{\Theta^k}k(\yb\mid\thetab)G_n(\ddr\thetab).
\end{equation}
We now model the sequence of observations $(\Yb_n)$, assuming that, given the information available at time $n$, $\Yb_{n+1}$ follows the probability mass function $p_{G_n}$. In other words, denoting by $(\mathcal G_n)_{n\geq 0}$ the natural filtration of $(\Yb_{n})_{n\geq1}$, we assume that $\Yb_{n+1}\mid\mathcal G_n\sim p_{G_n}$.
Under this assumption, the sequence $(G_{n})_{n\geq0}$ is a martingale with respect to $(\mathcal G_n)_{n\geq0}$, converging $\PP$\,-\,a.s. as $n\rightarrow\infty$ to a random probability distribution $\tilde{G}$ on $\Theta^k$. In particular, for every $\yb\in\mathbb{N}_{0}^k$, as $n\rightarrow+\infty$
\begin{equation}\label{conv_est2}
\hat{\thetab}_{G_{n},j}(\yb)\rightarrow \hat{\thetab}_{\tilde{G},j}(\yb)=(\yb_j+1)\frac{p_{\tilde G}(\yb+\eb_j)}{p_{\tilde G}(\yb)}\qquad \PP\,\text{-}\,a.s.
\end{equation}

\begin{thm}\label{th:cltmultiple}
Let $(\alpha_n)_{n\geq1}$ be non increasing and satisfy $\sum_{n\geq 1}(\alpha_n^2/\sum_{k\geq n}\alpha_n^2)^2<+\infty$, and let $b_n=(\sum_{k\geq n}\alpha_k^2)^{-1}$.
For every $\yb\in\mathbb N_0^k$,  $b_n^{1/2}(  \hat{\thetab}_{G_{n}}(\yb)-\hat{\thetab}_{\tilde{G}}(\yb))$ converges, in the sense of almost-sure conditional convergence, as $n\rightarrow+\infty$, to a centered (i.e., zero-mean) Gaussian kernel with the random covariance matrix $\Wb_{\tilde G}(\yb)$, whose $(j,j')$ element is
\begin{align*}\label{eq:W2}
&\Wb_{\tilde G,j,j'}(\yb)\\&=\hat\thetab_{\tilde G,j}(\yb)\hat\thetab_{\tilde G,j'}(\yb)\sum_{z\in\mathbb N_0}p_{\tilde G}(z)
\left[   \left(\frac{p_{\tilde G}(\yb+\eb_j\mid z)}{p_{\tilde G}(\yb+\eb_j)}-\frac{p_{\tilde G}(\yb\mid z)}{p_{\tilde G}(\yb)}
\right)\left(\frac{p_{\tilde G}(\yb+\eb_{j'}\mid z)}{p_{\tilde G}(\yb+\eb_{j'})}-\frac{p_{\tilde G}(\yb\mid z)}{p_{\tilde G}(\yb)}
\right)\right].
\end{align*}
\end{thm}
The proof follows the same steps as the proof of Theorem \ref{th:clt2}. For $\alpha_n=(\alpha+n)^{-\gamma}$ with $\gamma\in (1/2,1]$, the theorem holds with $b_n=(2\gamma-1)/n^{2\gamma-1}$. In particular, with $\alpha_n=(1+n)^{-1}$, $b_n=n$.  
Theorem \ref{th:clt2} provides a tool to obtain asymptotic credible regions for $\hat{\thetab}_{\tilde G}(\yb)$, for $\yb\in\mathbb{N}_{0}^k$.

We now consider the Bayesian regret under a ``true"  Poisson compound model for $(\Yb_{n})_{n\geq1}$, under which, for any $n\geq1$: i) the distribution of $\Yb_{n}$ given $\thetab^{\ast}_{n}$ is a product of Poisson distributions with parameters $\theta^{\ast}_{n,j}$; ii) the $\thetab^{\ast}_{n}$'s are i.i.d. from the ``true" prior distribution, say the oracle prior $G^{\ast}$. By the same techniques used in the univariate case, the frequentist guarantees can be studied in this more general setting. The regret becomes
\begin{equation}\label{reg2}
\textsc{Regret}(G_n,G^*)=\sum_{\yb\in\mathbb{N}_{0}^k}||\hat\thetab_{G_n}(\yb)-\hat\thetab_{G^*}(\yb)||^2p_{G^*}(\yb).
\end{equation}
 The following result extends Theorem \ref{th:regret} to the multivariate setting.

  \begin{thm}\label{th:regretmultiple}
Let  $G^{\ast}$ be an oracle prior on $(\mathbb R^+)^k$, and, for each $n\geq1$, let $G_{n}$  be as in \eqref{eq:newtonmultiple} with $\Theta$ compact. Let $\overline{\mathcal S}$ be the weak closure of the class $\mathcal S$ of probability measure on $(\mathbb R^+)^k$ that are absolutely continuous with respect to $G_0$ and let
$G^\dag=\text{argmin}_{G\in\overline{\mathcal S}}\textsc{KL}(p_{G^*} \,\|\, p_{G})$, where $\textsc{KL}$ denotes the Kullback--Leibler  divergence and, for every $G$,  $p_G(\yb)=\int k(\yb\mid\thetab)G(\ddr\thetab)$.  Then,
$
\hat\thetab_{G_n}(\yb)\rightarrow \hat\thetab_{G^{\dag}}(\yb)$  for every $\yb\in\mathbb N_0^k$ and $\textsc{Regret}(G_n,G^\dag)\!=\!\sum_{\yb\in\mathbb N_0^k}||\hat\thetab_{G_n}(\yb)-\hat\thetab_{G^\dag}(\yb)||^2p_{G^*}(\yb)\rightarrow 0$,  $\PP_{G^*}$-a.s. In particular, if $G^*$ is absolutely continuous with respect to $G_0$, then $\textsc{Regret}(G_n,G^*)\rightarrow 0$,  $\PP_{G^*}$-a.s.
 \end{thm}
The proof is analogous to that of Theorem \ref{th:regret}.


\end{document}